\documentclass[reprint,prx,aps,superscriptaddress,twocolumn,showpacs,longbibliography,floatfix,nofootinbib]{revtex4-2}
\usepackage{CSTheoryToolkitCMUStyle}
\usepackage{Custom}
\usepackage{cleveref}
\usepackage{multirow}
\usepackage{tikz-cd}
\usepackage{bbm}
\usetikzlibrary{calc, arrows.meta, positioning, decorations.markings, decorations.pathreplacing, fit, backgrounds, patterns, shapes.geometric}
\makeatletter
\let\PRX@numberline\numberline
\newcommand{\PRXTableOfContents}{%
  \begingroup
  \renewcommand{\numberline}[1]{\PRX@numberline{##1}\hspace{0.5em}}%
  \tableofcontents
  \endgroup
}
\makeatother

\newcommand{\col}{\mathrm{col}}
\newcommand{\row}{\mathrm{row}}
\newcommand{\BBc}{\mathbb{B}\mathbb{B}}
\newcommand{\syz}{\mathrm{syz}}
\newcommand{\coker}{\mathrm{coker}}
\newcommand{\ann}{\mathrm{ann}}

\newcommand{\TwoDTLTI}{2D TTI }

\begin{document}

\title{Generalized matching decoders for 2D topological translationally-invariant codes}


\author{Shi Jie Samuel Tan}
\thanks{These authors contributed equally to this work.}
\affiliation{Joint Center for Quantum Information and Computer Science, University of Maryland, College Park, Maryland 20742, USA}
\affiliation{QuEra Computing Inc., 1284 Soldiers Field Road, Boston, Massachusetts 02135, USA}

\author{Ian Gill}
\thanks{These authors contributed equally to this work.}
\affiliation{Yale Quantum Institute \& Department of Applied Physics, Yale University, New Haven, Connecticut 06511, USA}

\author{Eric Huang}
\thanks{These authors contributed equally to this work.}
\affiliation{Joint Center for Quantum Information and Computer Science, University of Maryland, College Park, Maryland 20742, USA}
\affiliation{QuEra Computing Inc., 1284 Soldiers Field Road, Boston, Massachusetts 02135, USA}

\author{Pengyu Liu}
\affiliation{QuEra Computing Inc., 1284 Soldiers Field Road, Boston, Massachusetts 02135, USA}
\affiliation{Department of Computer Science, Carnegie Mellon University, Pittsburgh, Pennsylvania 15213, USA}

\author{Chen Zhao}
\affiliation{QuEra Computing Inc., 1284 Soldiers Field Road, Boston, Massachusetts 02135, USA}

\author{Hossein Dehghani}
\affiliation{QuEra Computing Inc., 1284 Soldiers Field Road, Boston, Massachusetts 02135, USA}

\author{Aleksander Kubica}
\email{a.kubica@yale.edu}
\affiliation{Yale Quantum Institute \& Department of Applied Physics, Yale University, New Haven, Connecticut 06511, USA}

\author{Hengyun Zhou}
\email{hyzhou@mit.edu}
\affiliation{Department of Electrical Engineering and Computer Science, Massachusetts Institute of Technology, Cambridge, Massachusetts 02139, USA}

\author{Arpit Dua}
\email{adua@quera.com}
\affiliation{QuEra Computing Inc., 1284 Soldiers Field Road, Boston, Massachusetts 02135, USA}
\affiliation{Department of Physics, Virginia Tech, Blacksburg, Virginia 24061, USA}

\date{\today}

\begin{abstract}
Two-dimensional topological translationally-invariant (TTI) quantum codes, such as the toric code (TC) and bivariate bicycle (BB) codes, are promising candidates for fault-tolerant quantum computation.
For such codes to be practically relevant, their decoders must successfully correct the most likely errors while remaining computationally efficient.
For the TC, graph-matching decoders satisfy both requirements and, additionally, admit provable performance guarantees.
Given the equivalence between TTI codes and (multiple copies of) the TC, one may then ask whether TTI codes also admit analogous graph-matching decoders.
In this work, we develop a graph-matching approach to decoding general TTI codes.
Intuitively, our approach coarse-grains the TTI code to obtain an effective description of the syndrome in terms of TC excitations, which can then be removed using graph-matching techniques.
We prove that our decoders correct errors of weight up to a constant fraction of the code distance and achieve nonzero code-capacity and circuit-level thresholds in the specified noise models.
We further numerically study a variant optimized for practically relevant BB codes and observe performance comparable to that of belief propagation with ordered statistics decoding (BP--OSD). We also study the code-capacity threshold for the hexagonal color code as a proof of concept.
Our results indicate that graph-matching decoders are a viable approach to decoding BB codes and other TTI codes.
\end{abstract}

\maketitle

\PRXTableOfContents

\section{Introduction}
\label{sec:intro}
Quantum error correction (QEC)~\cite{shor1995scheme,steane1996error} is a key ingredient in building reliable quantum computers~\cite{ShorFT, aharonov1997fault,preskill1998}, as physical qubits are inherently noisy, and without protection, errors quickly accumulate and corrupt any long computation.
QEC addresses this challenge by encoding logical information into many physical qubits and detecting errors via parity-check measurements.
Among the large landscape of QEC codes, two-dimensional (2D) topological quantum codes~\cite{Kitaev2003,bombin2013introductiontopologicalquantumcodes} are especially attractive.
By definition, a 2D topological quantum code is realized by placing qubits on a lattice and introducing geometrically-local parity checks.
For concreteness, we restrict our attention to topological translationally-invariant (TTI) CSS codes defined on a square lattice with periodic boundary conditions.
For TTI codes, it suffices to specify a local generating set of parity checks and then translate it across the lattice.
Importantly, as we increase the lattice size, the code distance increases, while the range of parity checks does not.
Canonical examples of TTI codes include the toric code (TC)~\cite{Kitaev2003,dennis2002topological} and the color code~\cite{bombin2006topological}.
Recently, there has been significant interest in other TTI codes, such as bivariate bicycle (BB) codes~\cite{Kovalev2013,bravyi2024high} and tile codes~\cite{steffan2025tile}, as they provide better encoding rates than the TC and possess interesting code automorphisms~\cite{yoder2025tour,sayginel2025fault,breuckmann2026logical,koh2026achieving}.

To perform QEC in practice, one relies on specialized algorithms, called decoders, which solve the decoding problem---given the measurement outcomes of parity checks (the syndrome), return an appropriate correction.
Decoders are carefully designed to handle the most likely errors for a given noise model; they also have to be computationally efficient to avoid the so-called backlog problem~\cite{Terhal2015}, where classical processing cannot keep up with the syndrome-extraction rate.
In general, the decoding problem can be viewed as a matching problem in a hypergraph, where vertices correspond to (violated) parity checks and hyperedges correspond to independent errors; as such, decoding can be computationally hard~\cite{karp2009reducibility}.
However, for the 2D TC and independent bit-flip and phase-flip noise, the decoding problem reduces to the graph matching problem~\cite{dennis2002topological}, which admits an efficient solution via the minimum-weight perfect matching (MWPM) algorithm~\cite{edmonds1965paths}.
Importantly, its runtime is compatible with the syndrome extraction rate of current quantum hardware~\cite{Wu2023,higgott2025sparse}.
Furthermore, although the MWPM decoder does not find the optimal correction, it yields high QEC thresholds~\cite{dennis2002topological, fowler2012surface, higgott2023}, making it a standard benchmark for QEC protocols.

The idea of using graph-matching techniques to solve the decoding problem has also been explored in the context of the color code~\cite{wang2009graphicalalgorithmsthresholderror,bombin2012universal, delfosse2014decoding, kubica2023efficient, sahay2022decoder, gidney2023new,lee2025color}.
In fact, these previous results can be understood from the perspective of unitary equivalence of the color code and the TC~\cite{bombin2012universal, kubica2015unfolding}---using a constant-depth circuit comprising geometrically-local Clifford gates, one can map the decoding problem for the color code to the TC setting, apply graph-matching techniques there, and then map back the identified TC correction to the color code.
In a similar way, one may ask whether other 2D TTI codes can admit analogous graph-matching decoders, since 2D TTI codes are known to be equivalent to (multiple copies of) the TC~\cite{yoshida2011classification,bombin2012universal,bombin2014structure,haah2013commuting,haah2016algebraic}.
Understanding how this equivalence translates into efficient decoders is therefore both a natural and practically relevant question for developing QEC protocols based on 2D TTI codes.

In our work, we develop graph-matching decoders for 2D TTI codes that rely on the polynomial representation of 2D TTI~\cite{bombin2014structure,haah2016algebraic} and their equivalence to (multiple copies of) the 2D TC.
The basic idea is to coarse-grain the TTI code to obtain an effective description of the syndrome in terms of the TC excitations, which can then be removed using graph-matching techniques, yielding an appropriate correction for the original TTI code. 
We prove that our decoders:
(i) correct errors of weight up to a constant fraction (determined by the locality of parity checks) of the code distance,
(ii) have non-zero QEC thresholds in the code-capacity setting with independent and identically distributed (i.i.d.) bit-flip and phase-flip noise.
We also explain how the same sector maps can be used with repeated noisy syndrome measurements under circuit-level Pauli noise, where the induced sector problems are decoded using 3D space--time matching. This extension assumes the fixed local extraction circuit and temporal boundary conditions of Section~\ref{sec:circuit_noise}.
We also numerically study a variant of our graph-matching decoder optimized for the i.i.d. bit-flip noise and practically-relevant instances of the BB codes.
In particular, we consider the distance-12 gross code on the $12 \times 6$  lattice and the distance-24 code (with the same polynomial) on the $ 24 \times 24$ lattice, respectively, as well as the two-gross code.
On the $12\times 6$  lattice, our graph-matching decoder outperforms the standard belief propagation (BP) decoder at a high logical error rate and performs slightly worse than the belief propagation with ordered statistics decoder (BP-OSD)~\cite{poulin2008iterative,panteleev2021degenerate,roffe2020decoding,Roffe_LDPC_Python_tools_2022}; see Fig. ~\ref{fig:d12decodercomparison}.
Although our decoder does not outperform the BP-OSD, its competitive performance and favorable computational complexity provide a proof of concept for decoding BB codes and other TTI codes.
As such, our results open the door to future optimizations of graph-matching decoders for 2D TTI codes. Our numerical studies also include the hexagonal color code as evidence of a non-zero threshold in Fig.~\ref{fig:color_threshold} using our decoder.

\begin{figure*}[t]
    \centering
\includegraphics[width=\linewidth]{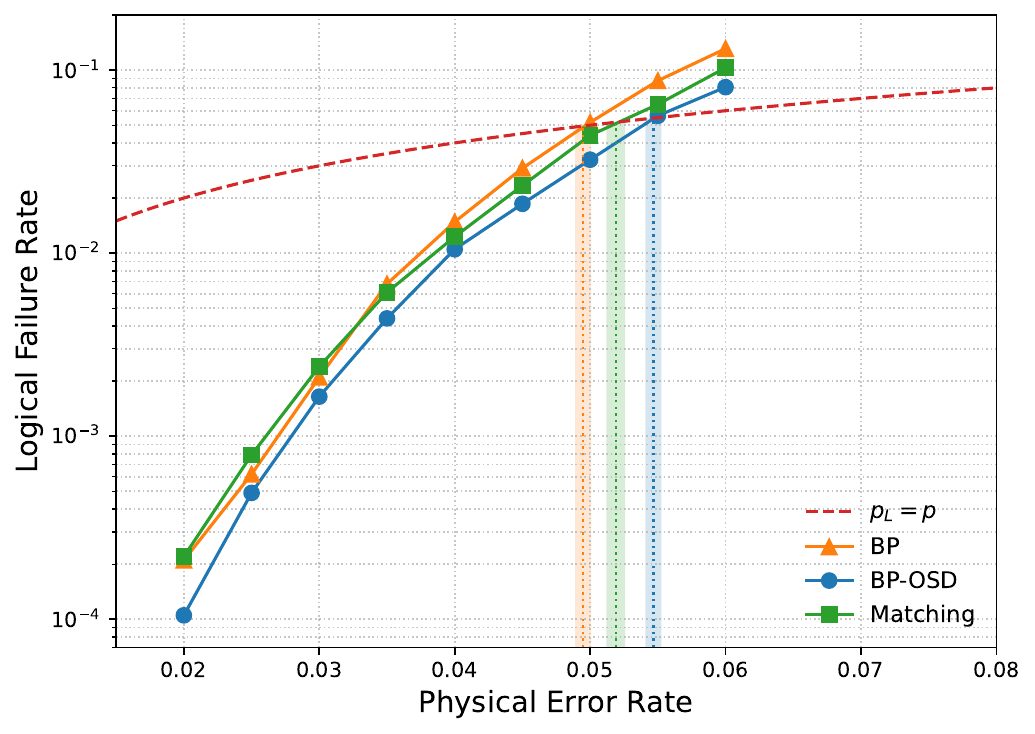}
    \caption{Decoding the 144-qubit BB code from  \cite{bravyi2024high} under i.i.d bit-flip noise using BP, BP-OSD, and the small-code-size adaptation of the cell-matching decoder. Logical error rates are determined through Monte Carlo sampling of error configurations at different physical error rates followed by decoding. For a data point with logical error rate of order $10^{-k}$,  $10^{k+2}$ trials were run. Statistical uncertainties are contained within marker size. Vertical lines indicate pseudothresholds for each decoding method, and shaded regions indicate pseudothreshold uncertainty estimated from logical error rate uncertainty.}
    \label{fig:d12decodercomparison}
\end{figure*}

Readers primarily interested in practical BB decoding can begin with Sections~\ref{sec:cellular_decoding} and \ref{sec:gross}, consulting the asymptotic proofs and layer-decoupling construction as needed.

The rest of this paper is organized as follows.
Section~\ref{sec:high_level_summary} provides an informal, high-level summary of the decoding problem and our two decoders.
Section~\ref{sec:preliminaries} reviews background on CSS codes, chain complexes, and the Laurent polynomial formalism.
Section~\ref{sec:decoding} presents the layer-decoupling decoder and explains how the decoupling outputs induce matching-decodable sector problems together with a lifting map.
Section~\ref{sec:cellular_decoding} presents the cell-matching decoder and analyzes its correctness and complexity.
Section~\ref{sec:gross} adapts the cell-matching decoder for small BB codes and provides numerical simulations of its performance for the code-capacity scenario with independent bit-flip noise.
Section~\ref{sec:circuit_noise} discusses the extension to the circuit-level noise model.
In Section~\ref{sec:conclusions}, we conclude with a discussion of open questions and future directions.
In the appendices, we provide a self-contained review of BB codes, the decoupling theorem and its algorithmic implementation, and the algebraic tools used in the decoupling process.
We also provide additional details on the decoding algorithms and their analysis.

\section{Overview}
\label{sec:high_level_summary}
A CSS code lets us correct $X$-type and $Z$-type Pauli errors separately~\cite{calderbank1996good, steane1996error}.
Throughout this summary, we will be focused on the $Z$-error decoding problem: a physical $Z$ error pattern anticommutes with some $X$ checks which is the measured syndrome.
A decoder is a classical algorithm that takes the syndrome as input and returns a $Z$-type Pauli correction; decoding succeeds if the product of the unknown error and the returned correction is a $Z$ stabilizer, implying that it acts trivially on logical qubits.
The problem of decoding $X$ errors is analogous.
While our exposition focuses on CSS codes, the same principles we use to produce the presented decoding procedure apply to general 2D TTI stabilizer codes because the results that we leverage are not specific to the CSS setting~\cite{haah2013commuting,bombin2014structure}. Going forward, all codes discussed are assumed to be CSS for simplicity.

The TC is the canonical example of a quantum code that is both physically local and efficiently decodable. A $Z$ error creates a pair of point-like excitations at the endpoints of error strings; since each string has two endpoints, the total number of excitations will have even parity. One can then model the observed excitations as vertices of a complete graph where edges are weighted as the number of qubits associated to the minimum length string connecting the endpoints of the edge. The matching strategy is to select a set of edges with minimum total length such that each vertex is incident to exactly one edge from this set. The proposed correction is the sum modulo two of the minimum length strings of $Z$ errors corresponding to the selected edges. This decoding approach can be implemented efficiently by minimum-weight perfect matching (MWPM) as shown in Fig.~\ref{fig:toric_code_mwpm_intuition}.

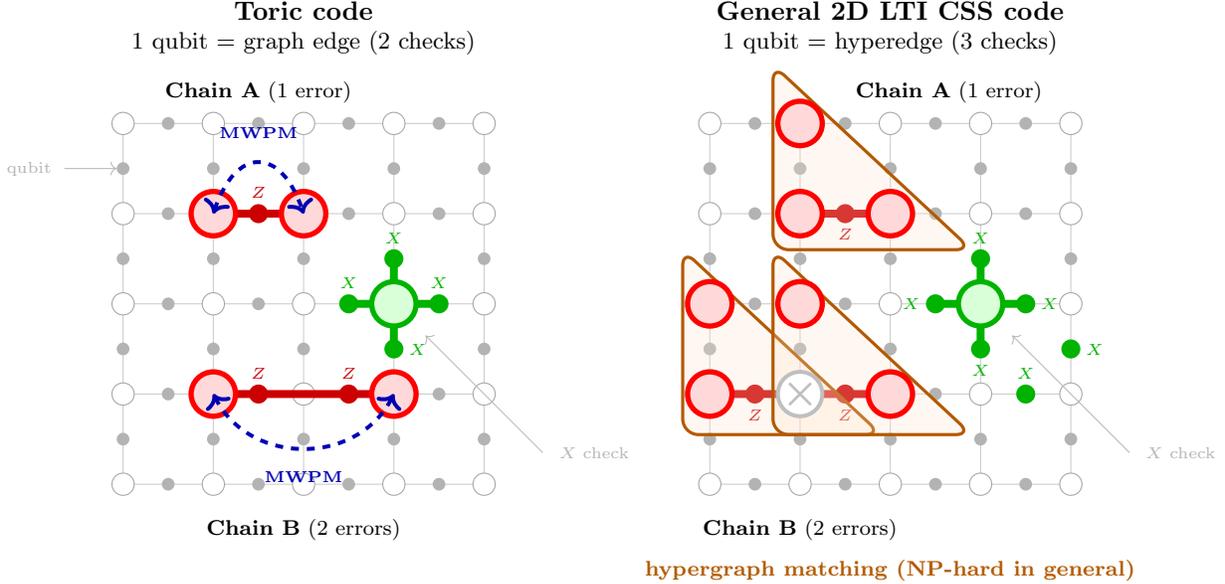
\begin{figure*}[t]
\centering
\begin{tikzpicture}[scale=1.2]


\foreach \i in {0,...,4} {
  \foreach \j in {0,...,3} { \draw[gray!40, thin] (\i,\j) -- (\i,\j+1); }
}
\foreach \i in {0,...,3} {  
  \foreach \j in {0,...,4} { \draw[gray!40, thin] (\i,\j) -- (\i+1,\j); }
}
\foreach \i in {0,...,4} {
  \foreach \j in {0,...,3} { \fill[gray!60] (\i, \j+0.5) circle (2pt); }
}
\foreach \i in {0,...,3} {
  \foreach \j in {0,...,4} { \fill[gray!60] (\i+0.5, \j) circle (2pt); }
}
\foreach \i in {0,...,4} {
  \foreach \j in {0,...,4} {
    \fill[white] (\i,\j) circle (3.5pt);
    \draw[gray!70, thin] (\i,\j) circle (3.5pt);
  }
}

\draw[red!80!black, line width=3pt, line cap=round] (1,3) -- (2,3);
\fill[red!80!black] (1.5,3) circle (3pt);
\node[red!80!black, font=\tiny\bfseries, above=2pt] at (1.5,3) {$Z$};
\fill[red!15] (1,3) circle (7pt); \draw[red, line width=2pt] (1,3) circle (7pt);
\fill[red!15] (2,3) circle (7pt); \draw[red, line width=2pt] (2,3) circle (7pt);
\draw[blue!70!black, dashed, line width=1.5pt, <->]
  (1,3) .. controls (1.25,3.75) and (1.75,3.75) .. (2,3);
\node[blue!70!black, font=\tiny\bfseries] at (1.5,3.9) {MWPM};
\node[font=\scriptsize, align=center] at (1.5,4.35) {\textbf{Chain A} (1 error)};

\draw[red!80!black, line width=3pt, line cap=round] (1,1) -- (2,1) -- (3,1);
\fill[red!80!black] (1.5,1) circle (3pt);
\fill[red!80!black] (2.5,1) circle (3pt);
\node[red!80!black, font=\tiny\bfseries, above=2pt] at (1.5,1) {$Z$};
\node[red!80!black, font=\tiny\bfseries, above=2pt] at (2.5,1) {$Z$};
\fill[red!15] (1,1) circle (7pt); \draw[red, line width=2pt] (1,1) circle (7pt);
\fill[red!15] (3,1) circle (7pt); \draw[red, line width=2pt] (3,1) circle (7pt);
\draw[blue!70!black, dashed, line width=1.5pt, <->]
  (1,1) .. controls (1.4,0.2) and (2.6,0.2) .. (3,1);
\node[blue!70!black, font=\tiny\bfseries] at (2,0.08) {MWPM};
\node[font=\scriptsize, align=center] at (2,-0.5) {\textbf{Chain B} (2 errors)};


\draw[green!70!black, line width=3pt, line cap=round] (3,2) -- (2.5,2);
\draw[green!70!black, line width=3pt, line cap=round] (3,2) -- (3.5,2);
\draw[green!70!black, line width=3pt, line cap=round] (3,2) -- (3,1.5);
\draw[green!70!black, line width=3pt, line cap=round] (3,2) -- (3,2.5);
\fill[green!70!black] (2.5,2) circle (3pt);
\node[green!70!black, font=\tiny\bfseries, above=2pt]  at (2.5,2) {$X$};
\fill[green!70!black] (3.5,2) circle (3pt);
\node[green!70!black, font=\tiny\bfseries, above=2pt] at (3.5,2) {$X$};
\fill[green!70!black] (3,1.5) circle (3pt);
\node[green!70!black, font=\tiny\bfseries, right=2pt] at (3,1.5) {$X$};
\fill[green!70!black] (3,2.5) circle (3pt);
\node[green!70!black, font=\tiny\bfseries, above=2pt] at (3,2.5) {$X$};
\fill[green!15] (3,2) circle (7pt); \draw[green!70!black, line width=2pt] (3,2) circle (7pt);

\draw[->, gray!55, thin] (-0.65,3.5) -- (-0.07,3.5);
\node[gray!60, font=\tiny, left=1pt] at (-0.65,3.5) {qubit};
\draw[->, gray!55, thin] (4.65,0.35) -- (3.35,1.65);
\node[gray!60, font=\tiny, anchor=west] at (4.72,0.35) {$X$ check};

\node[font=\small\bfseries] at (2, 5.25) {Toric code};
\node[font=\footnotesize] at (2, 4.9) {1 qubit $=$ graph edge (2 checks)};


\def\dx{6.5}

\foreach \i in {0,...,4} {
  \foreach \j in {0,...,3} { \draw[gray!40, thin] (\i+\dx,\j) -- (\i+\dx,\j+1); }
}
\foreach \i in {0,...,3} {
  \foreach \j in {0,...,4} { \draw[gray!40, thin] (\i+\dx,\j) -- (\i+1+\dx,\j); }
}
\foreach \i in {0,...,4} {
  \foreach \j in {0,...,3} { \fill[gray!60] (\i+\dx, \j+0.5) circle (2pt); }
}
\foreach \i in {0,...,3} {
  \foreach \j in {0,...,4} { \fill[gray!60] (\i+0.5+\dx, \j) circle (2pt); }
}
\foreach \i in {0,...,4} {
  \foreach \j in {0,...,4} {
    \fill[white] (\i+\dx,\j) circle (3.5pt);
    \draw[gray!70, thin] (\i+\dx,\j) circle (3.5pt);
  }
}


\draw[red!80!black, line width=3pt, line cap=round] (\dx+1,3) -- (\dx+2,3);
\fill[red!80!black] (\dx+1.5, 3) circle (3pt);
\node[red!80!black, font=\tiny\bfseries, below=2pt] at (\dx+1.5,3) {$Z$};
\fill[orange!20, fill opacity=0.25, rounded corners=6pt]
  (\dx+0.70,2.60) -- (\dx+2.90,2.60) -- (\dx+0.70,4.65) -- cycle;
\draw[orange!70!black, line width=1.2pt, rounded corners=6pt]
  (\dx+0.70,2.60) -- (\dx+2.90,2.60) -- (\dx+0.70,4.65) -- cycle;
\fill[red!15] (\dx+1,3) circle (7pt); \draw[red, line width=2pt] (\dx+1,3) circle (7pt);
\fill[red!15] (\dx+2,3) circle (7pt); \draw[red, line width=2pt] (\dx+2,3) circle (7pt);
\fill[red!15] (\dx+1,4) circle (7pt); \draw[red, line width=2pt] (\dx+1,4) circle (7pt);
\node[font=\scriptsize, align=center] at (\dx+2.65,4.35) {\textbf{Chain A} (1 error)};

\begin{scope}[shift={(-1,0)}]
\draw[red!80!black, line width=3pt, line cap=round] (\dx+1,1) -- (\dx+2,1) -- (\dx+3,1);
\fill[red!80!black] (\dx+1.5,1) circle (3pt);
\fill[red!80!black] (\dx+2.5,1) circle (3pt);
\node[red!80!black, font=\tiny\bfseries, below=2pt] at (\dx+1.5,1) {$Z$};
\node[red!80!black, font=\tiny\bfseries, below=2pt] at (\dx+2.5,1) {$Z$};
\fill[orange!20, fill opacity=0.25, rounded corners=6pt]
  (\dx+0.70,0.55) -- (\dx+2.90,0.55) -- (\dx+0.70,2.60) -- cycle;
\draw[orange!70!black, line width=1.2pt, rounded corners=6pt]
  (\dx+0.70,0.55) -- (\dx+2.90,0.55) -- (\dx+0.70,2.60) -- cycle;
\fill[orange!20, fill opacity=0.25, rounded corners=6pt]
  (\dx+1.70,0.55) -- (\dx+3.90,0.55) -- (\dx+1.70,2.60) -- cycle;
\draw[orange!70!black, line width=1.2pt, rounded corners=6pt]
  (\dx+1.70,0.55) -- (\dx+3.90,0.55) -- (\dx+1.70,2.60) -- cycle;
\fill[red!15] (\dx+1,1) circle (7pt); \draw[red, line width=2pt] (\dx+1,1) circle (7pt);
\fill[red!15] (\dx+3,1) circle (7pt); \draw[red, line width=2pt] (\dx+3,1) circle (7pt);
\fill[red!15] (\dx+1,2) circle (7pt); \draw[red, line width=2pt] (\dx+1,2) circle (7pt);
\fill[red!15] (\dx+2,2) circle (7pt); \draw[red, line width=2pt] (\dx+2,2) circle (7pt);

\fill[white] (\dx+2,1) circle (7pt);
\draw[gray!50, line width=1.5pt] (\dx+2,1) circle (7pt);
\draw[gray!50, line width=1pt] (\dx+1.88,0.88) -- (\dx+2.12,1.12);
\draw[gray!50, line width=1pt] (\dx+1.88,1.12) -- (\dx+2.12,0.88);

\node[font=\scriptsize, align=center] at (\dx+2,-0.5) {\textbf{Chain B} (2 errors)};
\node[orange!70!black, font=\scriptsize\bfseries, align=center] at (\dx+3, -0.95)
  {hypergraph matching (NP-hard in general)};
\end{scope}

\draw[green!70!black, line width=3pt, line cap=round] (\dx+3,2) -- (\dx+2.5,2);
\draw[green!70!black, line width=3pt, line cap=round] (\dx+3,2) -- (\dx+3.5,2);
\draw[green!70!black, line width=3pt, line cap=round] (\dx+3,2) -- (\dx+3,1.5);
\draw[green!70!black, line width=3pt, line cap=round] (\dx+3,2) -- (\dx+3,2.5);
\fill[green!70!black] (\dx+2.5,2) circle (3pt);
\node[green!70!black, font=\tiny\bfseries, left=2pt]  at (\dx+2.5,2) {$X$};
\fill[green!70!black] (\dx+3.5,2) circle (3pt);
\node[green!70!black, font=\tiny\bfseries, right=2pt] at (\dx+3.5,2) {$X$};
\fill[green!70!black] (\dx+3,1.5) circle (3pt);
\node[green!70!black, font=\tiny\bfseries, below=2pt] at (\dx+3,1.5) {$X$};
\fill[green!70!black] (\dx+3,2.5) circle (3pt);
\node[green!70!black, font=\tiny\bfseries, above=2pt] at (\dx+3,2.5) {$X$};
\fill[green!15] (\dx+3,2) circle (7pt); \draw[green!70!black, line width=2pt] (\dx+3,2) circle (7pt);
\draw[->, gray!55, thin] (\dx+4.65,0.35) -- (\dx+3.35,1.65);
\node[gray!60, font=\tiny, anchor=west] at (\dx+4.72,0.35) {$X$ check};

\fill[green!70!black] (\dx+4,1.5) circle (3pt);
\node[green!70!black, font=\tiny\bfseries, right=2pt] at (\dx+4,1.5) {$X$};
\fill[green!70!black] (\dx+3.5,1) circle (3pt);
\node[green!70!black, font=\tiny\bfseries, above=2pt] at (\dx+3.5,1) {$X$};

\node[font=\small\bfseries] at (\dx+2, 5.25) {General 2D TTI CSS code};
\node[font=\footnotesize] at (\dx+2, 4.9) {1 qubit $=$ hyperedge (3 checks)};

\end{tikzpicture}
\caption{%
  \textbf{Left:} Two $Z$-error chains (red) on a $5{\times}5$ TC patch, each
  producing exactly two excitations at the string endpoints regardless of chain length.
  Excitations always come in \emph{pairs}, reducing decoding to a graph matching
  problem. MWPM (blue arcs) finds the minimum-weight pairing.
  \textbf{Right:} In a general \TwoDTLTI code, one check (green) can act on more than the four nearest-neighbor qubits and one qubit can participate in
  three or more $X$ checks, so a single-qubit error can create three or more excitations simultaneously which is equivalent to a \emph{hyperedge}.  Minimum-weight hypergraph
  matching is NP-hard in general, breaking the naive MWPM approach.
}
\label{fig:toric_code_mwpm_intuition}
\end{figure*}

Many other \TwoDTLTI codes, such as the BB code, look more complicated than the TC at the level of syndrome patterns. A qubit can participate in more than two checks of the same type, and a single-qubit error can create more than two excitations. That breaks the simplest graph matching picture, since the correction we seek to obtain is no longer represented as an edge joining two excitations, but a \emph{hyperedge} joining more than two defects as depicted in Fig.~\ref{fig:toric_code_mwpm_intuition}. The equivalent minimum-weight decoding problem is a minimum-weight hypergraph matching problem, which is generally NP-hard~\cite{karp2009reducibility}. In fact, it was recently shown that the minimum-weight decoding problem for the color code is NP-hard~\cite{walters2026minimum}.

For the color code, the way around this obstruction is to project the syndrome into simpler pieces by restricting to subsets of check types before decoding those pieces using matching decoders, and then lift the answer back~\cite{delfosse2014decoding,kubica2023efficient,sahay2022decoder}.
That approach is very successful, but it is also tailored to the specific structure of the color code.
The main objective of this work is to recover a graph matching based decoding procedure for general \TwoDTLTI CSS codes. To achieve this, we use the decoupling theorem~\cite{haah2013commuting,bombin2014structure} that can be summarized informally as follows.

\begin{quote}
\emph{If a 2D translationally-invariant commuting-Pauli Hamiltonian has topological order (no non-trivial finite-support logical operators on the infinite lattice), then after coarse-graining, it can be transformed by a constant-depth, locality-preserving Clifford circuit into a tensor product of finitely many copies of the TC and qubits in a trivial state.}
\end{quote}

\begin{figure*}[t]
\centering
\begin{tikzpicture}[
  x={(0.62cm,0cm)},
  y={(0.30cm,0.19cm)},
  z={(0cm,0.80cm)},
  qubit/.style={circle, fill=gray!65, minimum size=3pt, inner sep=0pt},
  vtx/.style  ={circle, fill=white, draw=gray!55, line width=0.5pt,
                minimum size=5pt, inner sep=0pt},
  ktcq/.style ={circle, fill=gray!60, minimum size=2.5pt, inner sep=0pt},
  ktcv/.style ={circle, fill=white, draw=gray!50, line width=0.4pt,
                minimum size=4.5pt, inner sep=0pt}
]



\fill[gray!10] (0,0,0) -- (4,0,0) -- (4,4,0) -- (0,4,0) -- cycle;

\fill[orange!28, opacity=0.85]
  (0.08,2.08,0) -- (2.92,2.08,0) -- (2.92,3.92,0) -- (0.08,3.92,0) -- cycle;
\fill[blue!22, opacity=0.85]
  (1.08,0.08,0) -- (3.92,0.08,0) -- (3.92,3.92,0) -- (1.08,3.92,0) -- cycle;
\fill[violet!22, opacity=0.85]
  (0.08,0.08,0) -- (2.92,0.08,0) -- (2.92,2.92,0) -- (0.08,2.92,0) -- cycle;

\foreach \i in {0,...,4} {
  \draw[gray!45, very thin] (\i,0,0) -- (\i,4,0);
  \draw[gray!45, very thin] (0,\i,0) -- (4,\i,0);
}
\draw[gray!60] (0,0,0) -- (4,0,0) -- (4,4,0) -- (0,4,0) -- cycle;

\foreach \i in {0,...,3} {
  \foreach \j in {0,...,4} { \node[qubit] at (\i+0.5,\j,0) {}; }
}
\foreach \i in {0,...,4} {
  \foreach \j in {0,...,3} { \node[qubit] at (\i,\j+0.5,0) {}; }
}
\foreach \i in {0,...,4} {
  \foreach \j in {0,...,4} { \node[vtx] at (\i,\j,0) {}; }
}
\node[font=\small\bfseries, anchor=west] at (5.05,0,0) {2D TTI code};


\def\zcnotbot{2.75}
\def\zcnotmid{3.00}
\def\zcnottop{3.25}

\fill[yellow!12, opacity=0.45]
  (0,0,\zcnotbot) -- (4,0,\zcnotbot) -- (4,4,\zcnotbot) -- (0,4,\zcnotbot) -- cycle;
\fill[yellow!12, opacity=0.70]
  (0,0,\zcnottop) -- (4,0,\zcnottop) -- (4,4,\zcnottop) -- (0,4,\zcnottop) -- cycle;
\fill[yellow!10, opacity=0.22]
  (0,0,\zcnotbot) -- (4,0,\zcnotbot) -- (4,0,\zcnottop) -- (0,0,\zcnottop) -- cycle;
\fill[yellow!10, opacity=0.18]
  (4,0,\zcnotbot) -- (4,4,\zcnotbot) -- (4,4,\zcnottop) -- (4,0,\zcnottop) -- cycle;

\draw[gray!55] (0,0,\zcnotbot) -- (4,0,\zcnotbot) -- (4,4,\zcnotbot) -- (0,4,\zcnotbot) -- cycle;
\draw[gray!55] (0,0,\zcnottop) -- (4,0,\zcnottop) -- (4,4,\zcnottop) -- (0,4,\zcnottop) -- cycle;
\draw[gray!45, very thin] (0,0,\zcnotbot) -- (0,0,\zcnottop);
\draw[gray!45, very thin] (4,0,\zcnotbot) -- (4,0,\zcnottop);
\draw[gray!45, very thin] (4,4,\zcnotbot) -- (4,4,\zcnottop);
\draw[gray!45, very thin] (0,4,\zcnotbot) -- (0,4,\zcnottop);

\foreach \k/\z/\op/\xshift in {
  0/\zcnotbot/0.55/0.00,
  1/\zcnotmid/0.75/0.10,
  2/\zcnottop/0.95/0.20
} {
  \foreach \row in {0.5, 1.0, 1.5, 2.0, 2.5, 3.0, 3.5} {
    \draw[gray!50, very thin, opacity=\op] (0,\row,\z) -- (4,\row,\z);
  }

  \ifnum\k=0\relax
    \foreach \xpos/\pa/\pb in {
      0.5/0.5/1.0, 0.5/1.5/2.0, 0.5/2.5/3.0,
      1.3/1.0/1.5, 1.3/2.0/2.5, 1.3/3.0/3.5,
      2.1/0.5/1.0, 2.1/1.5/2.0, 2.1/2.5/3.0,
      2.9/1.0/1.5, 2.9/2.0/2.5, 2.9/3.0/3.5,
      3.7/0.5/1.0, 3.7/1.5/2.0, 3.7/2.5/3.0
    } {
      \fill[gray!70, opacity=\op] (\xpos+\xshift,\pa,\z) circle (2pt);
      \draw[gray!60, thin, opacity=\op] (\xpos+\xshift,\pa,\z) -- (\xpos+\xshift,\pb,\z);
      \draw[gray!60, opacity=\op] (\xpos+\xshift,\pb,\z) circle (3pt);
    }
  \else\ifnum\k=1\relax
    \foreach \xpos/\pa/\pb in {
      0.5/1.0/1.5, 0.5/2.0/2.5, 0.5/3.0/3.5,
      1.3/0.5/1.0, 1.3/1.5/2.0, 1.3/2.5/3.0,
      2.1/1.0/1.5, 2.1/2.0/2.5, 2.1/3.0/3.5,
      2.9/0.5/1.0, 2.9/1.5/2.0, 2.9/2.5/3.0,
      3.7/1.0/1.5, 3.7/2.0/2.5, 3.7/3.0/3.5
    } {
      \fill[gray!70, opacity=\op] (\xpos+\xshift,\pa,\z) circle (2pt);
      \draw[gray!60, thin, opacity=\op] (\xpos+\xshift,\pa,\z) -- (\xpos+\xshift,\pb,\z);
      \draw[gray!60, opacity=\op] (\xpos+\xshift,\pb,\z) circle (3pt);
    }
  \else
    \foreach \xpos/\pa/\pb in {
      0.9/0.5/1.0, 0.9/2.0/2.5, 0.9/3.0/3.5,
      1.7/1.0/1.5, 1.7/2.5/3.0,
      2.5/0.5/1.0, 2.5/1.5/2.0, 2.5/2.5/3.0,
      3.3/1.0/1.5, 3.3/2.0/2.5, 3.3/3.0/3.5
    } {
      \fill[gray!70, opacity=\op] (\xpos,\pa,\z) circle (2pt);
      \draw[gray!60, thin, opacity=\op] (\xpos,\pa,\z) -- (\xpos,\pb,\z);
      \draw[gray!60, opacity=\op] (\xpos,\pb,\z) circle (3pt);
    }
  \fi\fi
}

\node[font=\small\bfseries, anchor=west, align=left] at (5.05,0,\zcnotmid)
  {constant-depth\\CNOT circuit};


\fill[blue!12] (0,0,5.0) -- (4,0,5.0) -- (4,4,5.0) -- (0,4,5.0) -- cycle;

\foreach \i/\j in {0/0, 0/2, 1/1, 1/3, 2/0, 2/2, 3/1, 3/3} {
  \fill[blue!32]
    (\i,\j,5.0) -- (\i+1,\j,5.0) -- (\i+1,\j+1,5.0) -- (\i,\j+1,5.0) -- cycle;
}
\foreach \i in {0,...,4} {
  \draw[blue!40, very thin] (\i,0,5.0) -- (\i,4,5.0);
  \draw[blue!40, very thin] (0,\i,5.0) -- (4,\i,5.0);
}
\draw[blue!65] (0,0,5.0) -- (4,0,5.0) -- (4,4,5.0) -- (0,4,5.0) -- cycle;
\foreach \i in {0,...,3} {
  \foreach \j in {0,...,4} { \node[ktcq] at (\i+0.5,\j,5.0) {}; }
  \foreach \j in {0,...,3} { \node[ktcq] at (\i,\j+0.5,5.0) {}; }
}
\foreach \j in {0,...,3} { \node[ktcq] at (4,\j+0.5,5.0) {}; }
\foreach \i in {0,...,4} {
  \foreach \j in {0,...,4} { \node[ktcv] at (\i,\j,5.0) {}; }
}
\node[font=\small\bfseries, blue!72!black, anchor=west] at (5.05,0,5.0)
  {TC copy 1};


\fill[orange!12] (0,0,6.5) -- (4,0,6.5) -- (4,4,6.5) -- (0,4,6.5) -- cycle;
\foreach \i/\j in {0/0, 0/2, 1/1, 1/3, 2/0, 2/2, 3/1, 3/3} {
  \fill[orange!32]
    (\i,\j,6.5) -- (\i+1,\j,6.5) -- (\i+1,\j+1,6.5) -- (\i,\j+1,6.5) -- cycle;
}
\foreach \i in {0,...,4} {
  \draw[orange!50, very thin] (\i,0,6.5) -- (\i,4,6.5);
  \draw[orange!50, very thin] (0,\i,6.5) -- (4,\i,6.5);
}
\draw[orange!75!black] (0,0,6.5) -- (4,0,6.5) -- (4,4,6.5) -- (0,4,6.5) -- cycle;
\foreach \i in {0,...,3} {
  \foreach \j in {0,...,4} { \node[ktcq] at (\i+0.5,\j,6.5) {}; }
  \foreach \j in {0,...,3} { \node[ktcq] at (\i,\j+0.5,6.5) {}; }
}
\foreach \j in {0,...,3} { \node[ktcq] at (4,\j+0.5,6.5) {}; }
\foreach \i in {0,...,4} {
  \foreach \j in {0,...,4} { \node[ktcv] at (\i,\j,6.5) {}; }
}
\node[font=\small\bfseries, orange!80!black, anchor=west] at (5.05,0,6.5)
  {TC copy 2};

\node[font=\normalsize, gray!55, anchor=west] at (5.05,0,7.4) {$\vdots$};

\draw[gray!25, dashed, very thin] (0,0,0) -- (0,0,6.5);
\draw[gray!25, dashed, very thin] (4,0,0) -- (4,0,6.5);

\end{tikzpicture}
\caption{%
  Schematic illustration of the decoupling theorem,
  shown as a stack of 2D planes in 3D.
  \textbf{Bottom plane:} A patch of a general \TwoDTLTI code.
  Overlapping colored regions represent the entangled TC copies in the \TwoDTLTI code.
  \textbf{Middle plane:} A constant-depth CNOT circuit that decouples the \TwoDTLTI code into independent TC copies.
  \textbf{Top planes:} The code decouples into independent TC copies.
}
\label{fig:decoupling_theorem}
\end{figure*}
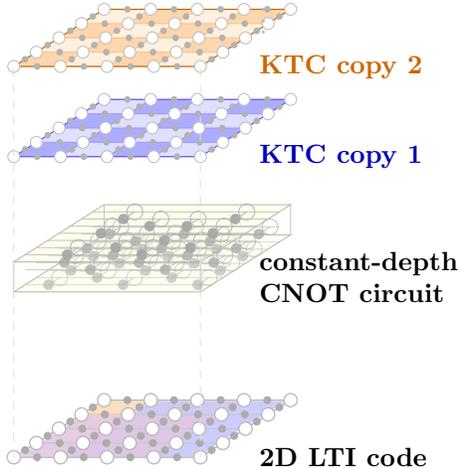

One way to interpret this statement is that every \TwoDTLTI code can be viewed as finitely many copies of the TC entangled together; see Fig.~\ref{fig:decoupling_theorem}. This was made explicit for the color codes in Ref.~\cite{kubica2015unfolding}.
Here is another picture that is useful for decoding intuition: while the different 2D TTI codes can have different microscopic stabilizer generators, we can always identify a finite set of
Pauli $X$ operators that create two single $Z$ stabilizer excitations. We refer to these Pauli $X$ operators as
short strings. In other words, these short strings create pairs of excitations in a TC-like way. Recall that
increasing the length of an error chain in one direction simply moves one of the $Z$ stabilizer violations away
from the other violation. Eventually, the two $Z$ stabilizer violations will meet and annihilate each other when the error chain is long enough and now forms a logical operator that wraps around the torus. The same can
be understood for the 2D TTI codes—by composing the aforementioned short strings , we can extend these
Pauli operators to form a loop around the torus, giving us a representative for the logical operators of the
2D TTI code. By translating these short strings along the 2D lattice, we can generate all other short strings
that create pairs of defects. In fact, all logical operators can be generated by these short strings and their number is related to the number of TC copies.
Thus, if we can \emph{algorithmically} extract those hidden toric-code blocks, then we can reduce TTI code decoding to decoding several TC-like instances.

Note that the decoupling theorem is an equivalence statement.
To build a decoder, we need to turn it into an explicit pipeline which accepts the observed syndrome and outputs a correction. The natural thing to do is to first map the measured syndrome of the original code into several ``sector syndromes'' that live on different TC-like structures. Subsequently, we decode each sector using a standard TC decoder, e.g., MWPM or Union-Find~\cite{delfosse2021almost}. Lastly, we map the sector corrections back into a correction on the original physical qubits. We present two ways to implement this idea.
\begin{enumerate}
  \item The \emph{layer-decoupling decoder} that uses an explicit decoupling-induced homomorphism between the original code and several TC copies.
  \item The \emph{cell-matching decoder} that cellulates the code and works with the violated checks in each intrinsic unit cell and applies local transport relations to perform matching on the underlying TC check violations.
\end{enumerate}
Even though these two decoders differ in how they produce the matching instances, the basic idea remains the same. Let $r$ be the number of TC sectors we decouple into and $\Lambda_b$ be the coarse-grained lattice that exposes the $r$ sectors. The two decoders both compute sets of violated checks
\begin{equation}
D_1,\dots,D_r \subseteq V(\Lambda_b)
\end{equation}
where $V(\Lambda_b)$ is the vertex set of $\Lambda_b$ and each $D_i$ has an even number of stabilizer violations (since we assume periodic boundary conditions).
Then they run MWPM/UF on each $D_i$ to get a coarse correction chain $\widehat\gamma^{(i)}$.
Finally, they lift $\{\widehat\gamma^{(i)}\}$ back to the original code to obtain a physical Pauli correction. In the layer-decoupling decoder, the map $s\mapsto (D_1,\dots,D_r)$ is induced by a global algebraic change-of-basis.
In the cell-matching decoder, the map is implemented by local transport moves in bounded neighborhoods of the unit cells.

The color code takes a similar approach, identifying subsets of checks on which matching is possible and lifting the matching results to a qubit level correction. The difference is in how these subsets of checks are identified.
For the color code, the check types are easily identifiable by the 3-coloring.
For BB codes and general TTI codes, the relevant check types live naturally in an algebraic object, the torsion of a cokernel,
and the projection is most naturally described as a module homomorphism. We refer to check configurations
that can be created by finite-support errors as \emph{physical} or \emph{locally-trivial} excitations. The quotient $\coker(H)$, where $H$ is the parity-check matrix of the code, describes exactly the check configurations modulo locally creatable patterns. Two check configurations are the same element of the cokernel if their difference can be created by a local Pauli operator. We also refer to these configurations as ``locally equivalent''. In
the families we target, $r$ is determined by the dimension of the cokernel. For example, the TC has two stabilizer violation patterns so its cokernel is isomorphic to $\mathbb{Z}_2$. The trivial (non-trivial) group element would be every check configuration involving an even (odd) number of violated checks. For a code that satisfies the topological order condition (i.e., the distance grows with the system size), the dimension of the cokernel is always finite, so $r$ is always a constant~\cite{haah2016algebraic}. Thus, it remains an efficient decoding primitive even if it involves running matching $r$ times.

\subsection{Decoder 1: the layer-decoupling decoder}

The key idea behind the layer-decoupling decoder is to use the guarantee that the decoupling theorem provides of the existence of a pair of invertible transformations to map between the 2D TTI code and the copies of the TC. These two transformations decouple the stabilizer generators and qubits of the 2D TTI code into those of the TC copies.
Once those maps are known, we can project the measured syndrome into different $r$ sector syndromes. Subsequently, we decode each sector using matching on the corresponding TC decoding graph before lifting the resulting sector corrections back to the original qubits by applying the inverse of the decoupling map
on the correction operators. We provide a diagram to depict the decoding pipeline in Fig.~\ref{fig:decoupling_pipeline}.

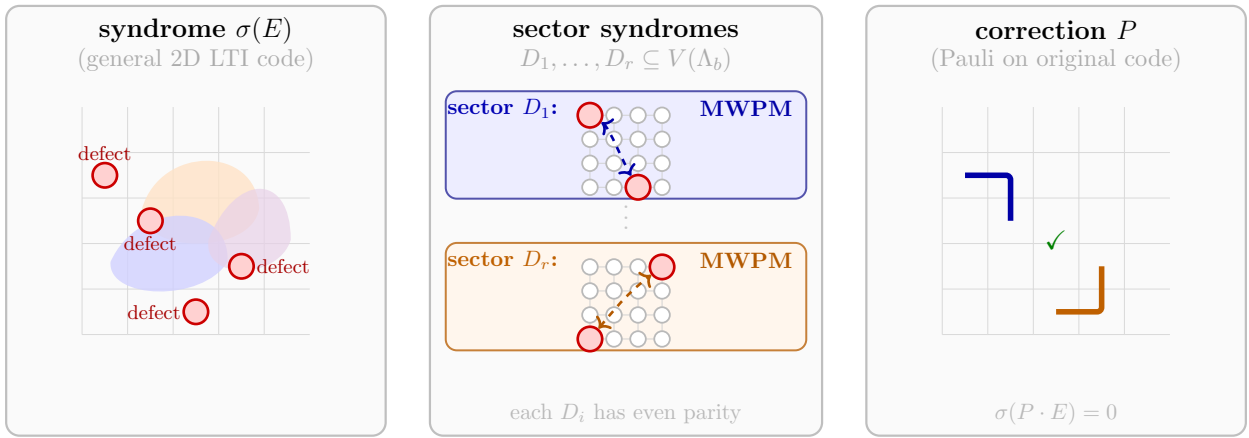
\begin{figure*}[t]
\centering
\resizebox{\linewidth}{!}{%
\begin{tikzpicture}[
  defect/.style={circle, draw=red!80!black, fill=red!18, line width=1.2pt,
                 minimum size=11pt, inner sep=0pt},
  gvtx/.style={circle, fill=white, draw=gray!55, line width=0.8pt,
               minimum size=9pt, inner sep=0pt},
  box/.style={rounded corners=6pt, draw=gray!50, fill=gray!4, line width=1pt}
]

\begin{scope}[shift={(0,0)}]
  \draw[box] (-3.0,-3.4) rectangle (3.0,3.4);
  \node[font=\large\bfseries] at (0, 3.02) {syndrome $\sigma(E)$};
  \node[font=\normalsize, gray!65] at (0, 2.55) {(general 2D TTI code)};

  \foreach \i in {0,...,4} {
    \foreach \j in {0,...,4} {
      \draw[gray!30, thin]
        (\i*0.72-1.44-0.36, \j*0.72-1.44-0.36) -- (\i*0.72-1.44+0.72-0.36, \j*0.72-1.44-0.36);
      \draw[gray!30, thin]
        (\i*0.72-1.44-0.36, \j*0.72-1.44-0.36) -- (\i*0.72-1.44-0.36, \j*0.72-1.44+0.72-0.36);
    }
  }
  \fill[orange!22, opacity=0.80]
    (-0.85, 0.20) .. controls (-0.55, 1.20) and (0.95, 1.20) ..
    (1.00, 0.20) .. controls (0.90,-0.55) and (-0.90,-0.55) .. cycle;
  \fill[blue!18, opacity=0.80]
    (-1.35,-0.65) .. controls (-0.90, 0.35) and (0.55, 0.30) ..
    (0.50,-0.65) .. controls (0.10,-1.30) and (-1.40,-1.25) .. cycle;
  \fill[violet!18, opacity=0.70]
    ( 0.20,-0.30) .. controls ( 0.60, 0.75) and (1.50, 0.80) ..
    (1.50,-0.30) .. controls (1.30,-0.95) and ( 0.10,-0.95) .. cycle;
  \node[defect] (p1d1) at (-1.44,  0.72) {};
  \node[defect] (p1d2) at ( 0.72,  -0.72) {};
  \node[defect] (p1d3) at (-0.72, 0) {};
  \node[defect] (p1d4) at ( 0, -1.44) {};
  \node[font=\footnotesize, red!65!black, above=3pt] at (p1d1) {defect};
  \node[font=\footnotesize, red!65!black, right=3pt] at (p1d2) {defect};
  \node[font=\footnotesize, red!65!black, below=3pt] at (p1d3) {defect};
  \node[font=\footnotesize, red!65!black, left=3pt] at (p1d4) {defect};
\end{scope}

\begin{scope}[shift={(6.8,0)}]
  \draw[box] (-3.1,-3.4) rectangle (3.1,3.4);
  \node[font=\large\bfseries] at (0, 3.02) {sector syndromes};
  \node[font=\normalsize, gray!65] at (0, 2.55)
    {$D_1,\ldots,D_r\subseteq V(\Lambda_b)$};

  \draw[blue!35!gray, rounded corners=5pt, fill=blue!7, line width=0.9pt]
    (-2.85, 0.35) rectangle (2.85, 2.05);
  \node[font=\small\bfseries, blue!70!black] at (-1.95, 1.78) {sector $D_1$:};
  \begin{scope}[shift={(-0.57,0.53)}]
    \def\s{0.38}
    \foreach \i in {0,1,2,3} {
      \draw[gray!35, thin] (\i*\s,0) -- (\i*\s,3*\s);
      \draw[gray!35, thin] (0,\i*\s) -- (3*\s,\i*\s);
    }
    \foreach \i in {0,1,2,3} { \foreach \j in {0,1,2,3} {
      \node[gvtx, minimum size=7pt] at (\i*\s,\j*\s) {};
    } }
    \node[defect, minimum size=11pt] (d1a) at (0,3*\s) {};
    \node[defect, minimum size=11pt] (d1b) at (2*\s,0) {};
    \draw[blue!65!black, dashed, line width=1.2pt, <->]
      (d1a) .. controls (0.30,0.95) and (0.55,0.35) .. (d1b);
  \end{scope}
  \node[font=\small\bfseries, blue!65!black] at (1.90, 1.78) {MWPM};

  \node[font=\large, gray!45] at (0, 0.18) {$\vdots$};

  \draw[orange!55!gray, rounded corners=5pt, fill=orange!7, line width=0.9pt]
    (-2.85,-2.05) rectangle (2.85,-0.35);
  \node[font=\small\bfseries, orange!75!black] at (-1.95,-0.62) {sector $D_r$:};
  \begin{scope}[shift={(-0.57,-1.87)}]
    \def\s{0.38}
    \foreach \i in {0,1,2,3} {
      \draw[gray!35, thin] (\i*\s,0) -- (\i*\s,3*\s);
      \draw[gray!35, thin] (0,\i*\s) -- (3*\s,\i*\s);
    }
    \foreach \i in {0,1,2,3} { \foreach \j in {0,1,2,3} {
      \node[gvtx, minimum size=7pt] at (\i*\s,\j*\s) {};
    } }
    \node[defect, minimum size=11pt] (dr1) at (0,0) {};
    \node[defect, minimum size=11pt] (dr2) at (3*\s,3*\s) {};
    \draw[orange!70!black, dashed, line width=1.2pt, <->]
      (dr1) .. controls (0.25,0.30) and (0.55,0.70) .. (dr2);
  \end{scope}
  \node[font=\small\bfseries, orange!70!black] at (1.90,-0.62) {MWPM};

  \node[font=\small, gray!55] at (0,-3.05) {each $D_i$ has even parity};
\end{scope}

\begin{scope}[shift={(13.6,0)}]
  \draw[box] (-3.0,-3.4) rectangle (3.0,3.4);
  \node[font=\large\bfseries] at (0, 3.02) {correction $P$};
  \node[font=\normalsize, gray!65] at (0, 2.55) {(Pauli on original code)};

  \foreach \i in {0,...,4} {
    \foreach \j in {0,...,4} {
      \draw[gray!30, thin]
        (\i*0.72-1.44-0.36, \j*0.72-1.44-0.36) -- (\i*0.72-1.44+0.72-0.36, \j*0.72-1.44-0.36);
      \draw[gray!30, thin]
        (\i*0.72-1.44-0.36, \j*0.72-1.44-0.36) -- (\i*0.72-1.44-0.36, \j*0.72-1.44+0.72-0.36);
    }
  }
  \draw[blue!65!black, line width=2.5pt, rounded corners=2pt]
    (-1.08-0.36, 1.08-0.36) -- (-0.36-0.36, 1.08-0.36) -- (-0.36-0.36, 0.36-0.36);
  \draw[orange!75!black, line width=2.5pt, rounded corners=2pt]
    ( 0,-1.44) -- ( 0.72,-1.44) -- ( 0.72,-0.72);
  \node[font=\large, green!50!black] at (0, -0.30) {$\checkmark$};
  \node[font=\small, gray!55] at (0,-3.05) {$\sigma(P \cdot E) = 0$};
\end{scope}

\end{tikzpicture}%
}
\caption{%
  Pipeline of the layer-decoupling decoder.
  \textbf{Project:} An algebraic change of basis maps the
  measured syndrome $\sigma$ of the \TwoDTLTI code into $r$ sector syndromes
  $D_1,\ldots,D_r$, each living on a TC-like coarse lattice whose
  check violations come in pairs.
  \textbf{Decode:} Each sector syndrome $D_i$ is decoded
  independently by matching.
  \textbf{Lift:} The sector-correction chains are mapped back to physical qubits via the inverse of the decoupling map, yielding a Pauli correction
  $P$ that clears the original syndrome.%
}
\label{fig:decoupling_pipeline}
\end{figure*}
At an abstract level, the syndrome projection step is exactly how we turn a hypergraph matching problem where check violations can be created in groups into a collection of graph matching problems where new check violations now come in pairs.
For the color code, this kind of projection can be described geometrically by restricting to sublattices.

However, the induced noise on the decoupled TC sectors can be correlated. In the worst case, a single-qubit error in the original code can map (under the decoupling transformation) to a \emph{pattern} of errors spread across multiple TC copies. If one hands MWPM an independent decoding graph but the true induced noise is highly correlated, the decoder can still work, but its performance can be noticeably reduced. The premise of matching is that edge weights in the decoding graph are supposed to correspond to error likelihoods. For the layer-decoupling decoder,
we assume that the physical noise is i.i.d. phase flips and the induced noise for a qubit in any one of the $r$ TC decoding graphs is the sum of the i.i.d. phase-flip probabilities from all physical qubits that are connected to the qubit via constant-depth CNOT circuit used for decoupling.

\subsection{Decoder 2: the cell-matching decoder}

The cell-matching decoder does not use the decoupling theorem as a prescription to globally rewrite the code by decoupling it into TC copies. Instead, it moves violated checks using local transport relations to expose the TC-like pairing structure at the level of unit cells.

Concretely, we first coarse-grain the lattice into $b\times b$ unit cells for a constant $b$ that depends only on the polynomials that define the code family and not on system size.
Inside each constant-size $b \times b$ unit cell, we perform a fixed local flushing procedure---we move the violated checks within the unit cell into a $c \times c$ basis subcell (for some $c < b$) at the top-left corner of the unit cell. Note that the basis subcell need not be a square, but making this choice is convenient. For each possible set of violated checks $s_u$ inside a unit cell $u$, a local Pauli operator $\Psi(s_u)$ contained within $u$ can be applied to move the violated checks within $u$ into the $c\times c$ fixed corner subcell. Note that a single violated check outside of the basis subcell can be transformed into several violated checks inside the basis subcell by a local operator. This is shown in the middle panel in Fig.~\ref{fig:cell-matching_decoder_pipeline}.
If $s_u$ happens to be a local violation (it can be created by an operator supported entirely in $u$), then flushing simply annihilates it.
We can keep $\Psi(s_u)$ as the local correction whose support is completely contained in $u$ that removes those violated checks.
If $s_u$ is nonlocal (it is the restriction of a global physical syndrome that crosses unit-cell boundaries), then flushing cannot annihilate it locally.
Instead, it produces a canonical residual check-violation type label in the basis subcell. 

The crucial point is that this residual class can be represented by a $r$-bit label, where $r$ is the number of independent check-violation types.
From a high-level perspective, the violated checks within the basis subcell can be identified with violated checks on $r$ different TC copies. 
For each of the $r$ sectors, the set of unit cells where that bit is $1$ has even parity for a physical syndrome on a compatible torus, allowing us to run MWPM. Circuit-level errors require the spacetime construction of Section~\ref{sec:circuit_noise}.

At that point, we run matching separately for each check-violation type on the coarse lattice. By pairing these violated checks within the basis subcell with their corresponding violated checks in adjacent basis subcells, we correct the errors on each TC copy. 
A matching edge corresponds to applying a fixed local ``edge generator'' Pauli short string supported near the boundary between the two unit cells, which creates or annihilates a pair of violated checks across the shared boundary.
Finally, we multiply all such edge generators together, and also include the flushing operators used during the first stage to get a physical Pauli correction that clears the original syndrome.

The cell-matching decoder therefore has two steps. The first step acts locally within each unit cell and finds the different check-violation types within the unit cell via flushing. The second step involves solving TC-like pairing problems for each check violation type on the coarse lattice of unit cells using matching.

\begin{figure*}[t]
\centering
\input{sections/tikz_cell-matching_decoder_pipeline.tex}
\caption{%
  Pipeline of the cell-matching decoder, using a color code example for illustrative purposes.
  \textbf{(a)} The 2D lattice with $3\times 3$ unit cells. The qubits are edges, the $X$ checks are vertices, and the $Z$ checks are plaquettes. Gold boundaries are drawn around sets of 9 $Z$ checks that belong in the same unit cell.
  A Pauli $X$ error configuration has resulted in some violated $Z$ checks shown in red.
  \textbf{(b)} After local flushing: local Pauli $X$ operators (orange $X$ markers)
  move each violated check into the $2\times 1$ basis subcell (purple) within the unit cell that the violated check is in. If the violated check is already in the basis subcell, no flushing is done for the check.
  \textbf{(c)} Global MWPM on the coarse lattice: there are two possible check-violation types in the $2 \times 1$ basis subcells, both of which remain and require matching on the two independent toric graphs.  
  Each check-violation type is paired by matching edges (green dashed) that correspond to short Pauli strings that generate pairs of excitations in the coarse lattice. 
}
\label{fig:cell-matching_decoder_pipeline}
\end{figure*}

At this point, it should be clear that both decoders rely on a constant-size coarse-graining.
The correct choice of $b$ depends on the polynomials used to define the \TwoDTLTI code.
At a high level, $b$ should be large enough that the unit-cell translation symmetries capture all single violated check mobility.
In other words, every single violated check in the $b\times b$ unit cell should be transportable to the same coordinate in a neighboring unit cell by some finite-support Pauli short string operator.
Once $b$ is fixed, the cell-matching decoder chooses a smaller basis subcell of size $c\times c$ inside each $b\times b$ unit cell.
This basis subcell is where we store canonical representatives of the different check-violation types.
The choice of $c$ is mostly a convenience---it should be large enough to hold a $r$-dimensional basis for the check-violation types.
In practice, for a set of polynomials used to define a \TwoDTLTI code family with growing distance, these parameters can be precomputed once and reused for all system sizes.

\subsection{Main Results}
\label{subsec:main_results_informal}

We now informally state our main results.
Throughout, $r$ denotes the number of TC sectors (equal to half the number of encoded logical qubits on tori compatible with the structural reduction). Recall that $b$ is the fixed coarse-graining parameter for the decoder under consideration.

\begin{theorem}[Performance guarantees for the layer-decoupling and cell-matching decoders; see Theorems~\ref{thm:decoupling-decoder-finite-size}, \ref{thm:time-complexity-decoupling-decoder-finite-size}, \ref{thm:cellular-decoder-finite-size}, and \ref{thm:time-complexity-cellular-decoder-finite-size}]
Consider an arbitrary set of polynomials that define a 2D TTI code on a torus.
There is a precomputable, polynomial-dependent constant $b$ that divides the length $L_x$ and width $L_y$ of the torus. 
Let the \emph{coarse distance} $d_{\mathrm{TC}} = \min\{L_x/b, L_y/b\}$ be the length of a shortest nontrivial cycle of its coarse square torus.
For sufficiently large system sizes, both the layer-decoupling and cell-matching decoders correct every error in the decoded Pauli sector of weight at most $\left\lfloor(d_{\mathrm{TC}}-1)/(2k_\ast)\right\rfloor$, where $k_\ast=k_{\mathrm{lay}}$ or $k_{\mathrm{cell}}$ is the corresponding constant forward error spread.
Both decoders have nonzero thresholds under local-stochastic physical noise, including i.i.d. bit-flip or phase-flip noise. In particular, if the physical error rate is below some code-family-dependent constant determined by the locality of the reduction, then the logical error rate for the TTI code decays exponentially with $d_{\mathrm{TC}}$ for families with growing sizes but with fixed length-and-width ratio.
The time complexity of both decoders scales asymptotically with the cost of running the MWPM decoder on the lattice with dimensions $L_x \times L_y$.
\end{theorem}

Section~\ref{sec:circuit_noise} extends the construction to circuit-level Pauli noise for the specified extraction circuit and memory boundaries, using 3D space--time matching and the same spatial lifting maps.

Both decoders ultimately rely on reducing the decoding problem to matching instances that live on a torus.
On a torus, a defect-pairing decoder can fail only if it chooses a correction whose difference from the error contains a non-contractible loop (a nontrivial homology class).
In our setting, there are $r$ such TC-like sectors.
The decoders are designed so that any  logical operator in the original code induces a non-trivial cycle in at least one sector, and boundaries in a sector lift to stabilizers.
Therefore, if every sector decoder returns a correction that differs from the induced sector error by a boundary, the combined lifted correction differs from the error by a stabilizer. 

The practical performance of the layer-decoupling decoder is highly dependent on our choice of the decoupling map which is, in general, not unique. One can consider a possibly inefficient way to decouple the TC copies using a CNOT circuit that involves certain physical qubits a large number of times. This translates into an induced noise model with a higher error probability on the qubits in certain TC copies, resulting in poor practical error correction performance. Because of the potential challenges associated with identifying a sparse decoupling map out of all possible decoupling maps, we choose to perform our numerical analysis on the cell-matching decoder and its finite-size adaptations.

Another issue with these decoding procedures is that if the coarse-graining parameter is large, $d_{TC}$ may be small, reducing the number of errors the decoder is capable of correcting.
To improve the practical performance of our matching decoder, it is pertinent to identify a basis for our violated checks where the edge generators for as many pairs of such check-violation types as possible are short so that their corresponding $d_{TC}$s are relatively larger.
Using the cell-matching decoder as a baseline framework, we develop specialized matching decoding procedures for small and intermediate code sizes which overcome this problem. The general idea behind these procedures is that we can use the  information about how individual violated checks project onto the various TC sectors of the cell-matching decoder to produce a set of matching graphs and associated edge weights which incorporates more information about the original syndrome than is captured in the coarse grained point of view. \\
\indent We evaluate these decoding procedures on BB codes, a practically relevant family of \TwoDTLTI codes that can be specified compactly by two bivariate polynomials. This family includes familiar examples, such as the TC and the color code. We focus on two well-studied BB instances, the gross and two-gross codes~\cite{yoder2025tour}, for which off-the-shelf matching constructions are not directly applicable. To probe larger lattice sizes, we also consider the 1152-qubit BB code (defined via the same polynomials as the gross code, but on the $24\times 24$ lattice instead of the $12\times 6$ gross code lattice).
We find that the adapted matching decoder improves over BP on some BB codes and performs comparably to BP-OSD. For the gross code under code-capacity noise, we estimate a pseudothreshold of $\approx 5.2\%$, compared to $\approx 5.0\%$ for BP and $\approx 5.5\%$ for BP-OSD.
Our numerical evidence also includes the color-code threshold study in Fig.~\ref{fig:color_threshold}. Overall, these results indicate that graph-matching decoders can be competitive for TTI codes and provide a proof of concept.
For the gross code, the adapted decoder approaches BP-OSD performance using a combinatorial procedure built from local translation structure of the excitations. This suggests that the basis decomposition of excitations into check-violation types captures actionable geometric information about likely error clusters, and that MWPM-style primitives remain useful beyond the TC setting. While our implementation is a proof of concept and evaluated under code-capacity noise, the underlying subroutines (shortest paths, MWPM, and parallel shifts) are implementation-friendly and provide a concrete starting point for more optimized and more realistic noise models.
The present implementation is slower though we note that Python/C++ implementation differences limit runtime comparisons, as discussed in Section~\ref{sec:gross}.
Note that our work does not conflict with the hardness result from Refs.~\cite{walters2026minimum,gu2026color}. We merely show that there exist efficient ways to decode 2D TTI codes using MWPM and do not claim that there exists a polynomial-time algorithm to perform minimum-weight decoding for arbitrary 2D TTI codes. In fact, there have been follow-up works on polynomial-time approximation schemes for minimum-weight decoding for 2D TTI codes~\cite{walters2026approximately,gu2026polynomial}.

Furthermore, there are many existing studies on the thresholds of the TC under more exotic noise models that include coherent errors, correlated errors, burst errors, etc.~\cite{wang2010quantum,stephens2014fault,auger2017fault,darmawan2018linear,tan2024resilience,chadwick2024averting,pataki2024coherent}.
Our work suggests that \TwoDTLTI codes may benefit from such TC decoding techniques whenever the sectorized noise can be represented by suitable toric-code matching instances.
Identifying when this representation is available for noise models beyond the circuit-level Pauli model of Section~\ref{sec:circuit_noise} is an important direction for future work.

\section{Preliminaries}
\label{sec:preliminaries}
This section provides the necessary notation and preliminaries for the paper.

\subsection{Notation}
\label{subsec:notation}
For a positive integer $n$, let $[n] = \{1, 2, \ldots, n\}$.
Given a vector $v \in \F_2^n$, let $|v| = \left|\set{i \in [n]: v_i \neq 0}\right|$ denote the Hamming weight of $v$, i.e., the number of nonzero coordinates in the standard basis.
We also assume that all $\F_2$-vector spaces in our note are finite-dimensional which implies that all Hamming weights are finite.

\subsection{Quantum codes}
\label{subsec:quantum-codes}
This section states the basic definitions in quantum coding theory.
In particular, we pay close attention to quantum CSS codes.

\begin{definition}[Quantum CSS Codes]
  \label{def:quantum-css-codes}
  For a finite field $\F_2$, a quantum CSS code of length $n$ over $\F_2$ is a pair $Q = \left(Q_X, Q_Z\right)$ of subspaces (i.e. classical codes) $Q_X, Q_Z \subseteq \F_2^n$ such that $Q_X^\perp \subseteq Q_Z$.
  The dimension of $Q$ is $k\coloneqq \dim\left(Q_Z\right) - \dim\left(Q_X^\perp\right)$, and the rate is $R \coloneqq k/n$.
  The distance of $Q$ is 
  \begin{equation}d \coloneqq \min_{c \in \left(Q_X \setminus Q_Z^\perp\right) \cup \left(Q_Z\setminus Q_X^\perp\right)} |c|.\end{equation}
  We refer to $Q$ as an $\llbracket n, k, d\rrbracket$ CSS code.
  Sometimes, we differentiate between the $X$ and $Z$ distance of the code and define them as follows:
  \begin{equation}d_X \coloneqq \min_{c \in Q_X \setminus Q_Z^\perp}|c|,\quad\quad d_Z \coloneqq \min_{c \in Q_Z \setminus Q_X^\perp}|c|.\end{equation}
  The locality $w$ of $Q$ is the maximum number of nonzero entries in any row or column of the parity-check matrices of $Q_X$ and $Q_Z$.
\end{definition}

For qLDPC codes, the locality $w$ is a constant that is independent of the code length $n$. This implies that each stabilizer generator checks at most a constant number of qubits and each qubit is checked by at most a constant number of stabilizer generators. 

It is common for us to associate the single qubit Pauli operators $X, Z, Y$ with the following expressions:
\begin{equation}X = \left(\begin{array}{c|c}0 & 1\end{array}\right),\quad\quad Z = \left(\begin{array}{c|c}1 & 0\end{array}\right),\quad\quad Y = \left(\begin{array}{c|c}1 & 1\end{array}\right).\end{equation}
Using the above expressions and extending it to the $n$ qubits case, we can write the parity-check matrix for a quantum CSS code $Q$ in a single compact matrix in $\F_2^{(n_X + n_Z) \times 2n}$ as follows:
\begin{equation}\left(\begin{array}{c|c}
  H_X & 0\\
  0 & H_Z
\end{array}\right)\end{equation}
where $H_X$ and $H_Z$ are the parity-check matrices for the $X$ and $Z$ stabilizers respectively and $n_X$ (corr. $n_Z$) corresponds to the number of $X$ (corr. $Z$) stabilizer generators. This expression is known as the binary symplectic matrix of the CSS code.

\subsection{Chain complexes}
\label{subsec:chain-complexes}
This section states the basic definitions in homological algebra.

\begin{definition}[Chain Complexes]\label{def:chain-complexes}
  A chain complex $\calC_\ast$ over a field $\F_2$ consists of a sequence of $\F_2$-vector spaces $\left(C_i\right)_{i \in \Z}$ and linear boundary maps $\left(\partial_i^\calC: C_i \to C_{i-1}\right)_{i \in \Z}$ satisfying $\partial_{i-1}^\calC \circ \partial_i^\calC = 0$ for all $i \in \Z$.
  When clear from context, we omit the superscript and subscript and write $\partial = \partial_i = \partial^\calC = \partial_i^\calC$. 
  Assuming that each $C_i$ has a fixed basis, then the locality $w^\calC$ of $\calC$ is the maximum number of nonzero entries in any row or column of any matrix $\partial_i$ in this fixed basis.
  If there exists bounds $\ell < m \in \Z$ such that for all $i < \ell$ and $i > m$ have $C_i = 0$, then we may truncate the sequence and say that $\calC$ is the $(m-\ell + 1)$-term chain complex
  \begin{equation}\calC_\ast = \left(C_m \xrightarrow[]{\partial_m} C_{m-1} \xrightarrow[]{\partial_{m-1}} \ldots \xrightarrow[]{\partial_{\ell + 1}} C_\ell\right).\end{equation}
  We furthermore define the following (standard) vector spaces for $i \in \Z$:
  \begin{align}
    \text{the space of } i\text{-cycles: } Z_i\left(\calC\right) &\coloneqq \ker\left(\partial_i\right) \subseteq C_i,\\
    \text{the space of } i\text{-boundaries: } B_i\left(\calC\right) &\coloneqq \mathrm{im}\left(\partial_{i+1}\right) \subseteq C_i,\\
    \text{the space of } i\text{-homology: } H_i\left(\calC\right) &\coloneqq Z_i\left(\calC\right)/B_i\left(\calC\right).
  \end{align}

  The cochain complex $\calC^\ast$ associated to $\calC_\ast$ is the chain complex with vector spaces $\left(C^i \coloneqq C_i\right)_{i \in \Z}$ and boundary maps given by the coboundary maps $\left(\delta_i = \partial_{i+1}^\top : C^i \to C^{i+1}\right)_{i \in \Z}$ obtained by transposing all the boundary maps of $\calC_\ast$.
  Thus, the cochain complex is defined as such:
  \begin{equation}\calC^\ast = \left(C^m \xleftarrow[]{\delta_{m-1}} C^{m - 1} \xleftarrow[]{\delta_{m-2}} \ldots \xleftarrow[]{\delta_{\ell}} C^\ell\right).\end{equation}
  We can analogously define the spaces of cohomology $H^i\left(\calC\right) = Z^i(\calC)/B^i(\calC)$, cocycles $Z^i\left(\calC\right) = \ker\left(\delta_i\right)$, and coboundaries $B^i\left(\calC\right) = \mathrm{im}\left(\delta_{i-1}\right)$.
\end{definition}

\begin{definition}
  For a chain complex $\calC$, the $i$-systolic distance $d_i(\calC)$ and the $i$-cosystolic distance $d^i(\calC)$ are defined as
  \begin{align}
    d_i(\calC) &= \min_{c \in Z_i(\calC)\setminus B_i(\calC)}|c|,\\
    d^i(\calC) &= \min_{c \in Z^i(\calC)\setminus B^i(\calC)}|c|.
  \end{align}
\end{definition}

It is well-known that classical linear codes can be described by 2-term chain complexes where the two vector spaces are the spaces of bits and checks respectively. These two vector spaces are connected by a linear boundary map that can be written as the check matrix $H$.
Quantum CSS codes can be described with a 3-term chain complex by associating the $X$ stabilizers, qubits, and $Z$ stabilizers with the three vector spaces in a 3-term chain complex. 
The condition ``the boundary of a boundary is trivial'' is compatible with the CSS orthogonality condition i.e.,  $\partial_{i-1} \circ \partial_i = 0$ and $H_X H_Z^\top = 0$.
To see how the Pauli logical operators fit in the chain complex picture, we associate the $X$ stabilizers, qubits, and $Z$ stabilizers to the $\mathbb{F}_2$-vector spaces $C_0, C_1$, and $C_2$, then the $X$ and $Z$ logical operator representatives are given by the basis elements of the $1$-cohomology space and $1$-homology space respectively. The $X$ and $Z$ distances of the quantum code are then $d^1(\mathcal{C})$ and $d_1(\mathcal{C})$.

\subsection{Polynomial representation of \TwoDTLTI codes}
\label{subsec:polynomial-representation-2d-topological-codes}

In this section, we explore the Laurent polynomial representation of \TwoDTLTI codes. This representation is particularly useful for analyzing the properties of these codes and their performance in quantum error correction.

\subsubsection{Laurent polynomials}
\label{subsubsec:laurent-polynomials}
A Laurent polynomial is a polynomial that allows for negative powers of its variables. Formally, a Laurent polynomial in two variables $x$ and $y$ over the finite field $\F_2$ is an expression of the form:
\begin{equation}f(x,y) = \sum_{i=-\infty}^{\infty} \sum_{j=-\infty}^{\infty} a_{ij} x^i y^j,\end{equation}
where $a_{ij} \in \F_2$ and only finitely many $a_{ij}$ are nonzero. The degree of a Laurent polynomial is defined as the maximum of the degrees of its terms, i.e., $\deg(f) = \max\{i + j : a_{ij} \neq 0\}$.
For the purpose of this work, we only consider the case where the exponents are finite. The bivariate Laurent polynomials form a ring $R = \F_2[x^{\pm 1}, y^{\pm 1}]$ with the usual addition and multiplication of polynomials. The ring $R$ is a commutative ring with unity, and it is also a $\F_2$-vector space with a basis given by the monomials $x^i y^j$ for $i, j \in \Z$.
\subsubsection{\TwoDTLTI codes}
\label{subsubsec:2d-topological-translationally-invariant-css-codes}
A \TwoDTLTI code is a quantum error-correcting code that is defined on a two-dimensional lattice with periodic boundary conditions. These codes are characterized by their stabilizer checks, which are local and translationally invariant (LTI), meaning that they can be translated across the lattice without changing their structure. We define local checks as stabilizer checks that act not only on a constant number of qubits but also only on qubits that are constant distance away. In other words, these qubits are effectively adjacent to the check when we scale the lattice to infinity. The stabilizer checks are typically defined using Laurent polynomials in two variables, which allows for a compact representation of the code's structure.

The key idea is to associate each site in the lattice with $q$ qubits. For most topological quantum codes, $q = 2$ and each site in the two-dimensional lattice is associated with the pair of qubits that lie on the edges to the north and east of the site. One can view this as the act of distinguishing between the qubits that reside on the horizontal and vertical edges of the lattice. In the more general case, we can associate $q$ qubits to a site in the lattice, where $q$ is a constant that is independent of the code length $n$. Each of these sites is then associated with a bivariate Laurent monomial $x^i y^j,$
where $i$ and $j$ are the horizontal and vertical coordinates of the site in the lattice. Similarly, the $X$ and $Z$ stabilizer checks of the code can also be labeled by these bivariate Laurent monomials. Multiplying by $x$ means translating one step in the horizontal direction, and multiplying by $y$ means translating one step vertically.
The antipode $f(x,y)^*$ (replace $x$ by $x^{-1}$ and $y$ by $y^{-1}$) corresponds to reversing the translation direction. A polynomial can be understood as a finite stencil. If a stabilizer generator is described by a polynomial like $1 + x + x^{-1}y,$
one should interpret it as it having to act on the qubit at the current site, the qubit one step to the right, and the qubit one step left-and-up. This is illustrated in Fig. \ref{fig:informal_dictionary_tikz}.
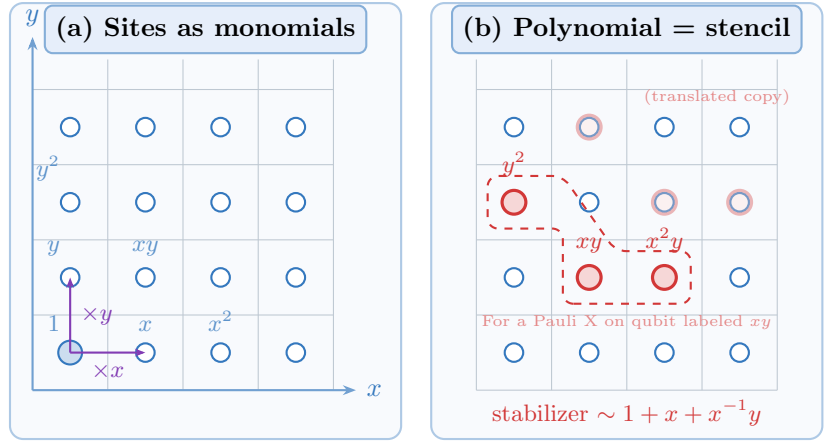
\begin{figure*}[t]
\centering
\definecolor{siteblue}{RGB}{52,120,190}
\definecolor{supportred}{RGB}{210,55,55}
\definecolor{supportgreen}{RGB}{40,160,80}
\definecolor{supportpurple}{RGB}{130,60,180}
\definecolor{gridgray}{RGB}{190,200,210}
\definecolor{panelbg}{RGB}{248,250,253}
\definecolor{headerbg}{RGB}{230,238,250}

\tikzset{
  qubit/.style={circle, draw=siteblue, fill=white, thick, minimum size=7pt, inner sep=0pt},
  qubitfill/.style={circle, draw=siteblue, fill=siteblue!25, thick, minimum size=7pt, inner sep=0pt},
  activequbit/.style={circle, draw=supportred, fill=supportred!20, very thick, minimum size=9pt, inner sep=0pt},
  stabilizer/.style={rectangle, rounded corners=3pt, draw=supportgreen!70!black, fill=supportgreen!12, thick},
  arrow/.style={-{Stealth[length=5pt]}, thick},
  panelbox/.style={rectangle, rounded corners=5pt, draw=siteblue!40, fill=panelbg, thick},
  headerbox/.style={rectangle, rounded corners=3pt, fill=headerbg, draw=siteblue!60, thick, inner sep=4pt},
}

\begin{tikzpicture}[font=\small]

\begin{scope}[shift={(0,0)}]

  \node[panelbox, minimum width=5.2cm, minimum height=5.8cm] at (2.3,2.55) {};
  \node[headerbox] at (2.3,5.1) {\textbf{(a) Sites as monomials}};

  \foreach \i in {0,...,4} {
    \draw[gridgray, thin] (\i*1, 0.3) -- (\i*1, 4.7);
    \draw[gridgray, thin] (0.0, 0.3+\i*1) -- (4.0, 0.3+\i*1);
  }

  \draw[arrow, siteblue!70] (0.0,0.3) -- (4.3,0.3) node[right, siteblue] {$x$};
  \draw[arrow, siteblue!70] (0.0,0.3) -- (0.0,5.0) node[above, siteblue] {$y$};

  \foreach \i in {0,...,3} {
    \foreach \j in {0,...,3} {
      \node[qubit] (q\i\j) at (\i*1+0.5, \j*1+0.8) {};
    }
  }

  \node[qubitfill, minimum size=9pt] at (0.5, 0.8) {};
  \node[above left, siteblue!80, font=\footnotesize] at (0.5,0.95) {$1$};
  \node[above, siteblue!80, font=\footnotesize] at (1.5,0.95) {$x$};
  \node[above, siteblue!80, font=\footnotesize] at (2.5,0.95) {$x^2$};
  \node[above left, siteblue!80, font=\footnotesize] at (0.5,1.95) {$y$};
  \node[above left, siteblue!80, font=\footnotesize] at (0.5,2.95) {$y^2$};
  \node[above, siteblue!80, font=\footnotesize] at (1.5,1.95) {$xy$};

  \draw[-{Stealth[length=4pt]}, supportpurple, thick]
    (0.5,0.8) -- (1.5,0.8)
    node[midway, below, font=\footnotesize, supportpurple] {$\times x$};
  \draw[-{Stealth[length=4pt]}, supportpurple, thick]
    (0.5,0.8) -- (0.5,1.8)
    node[midway, right, font=\footnotesize, supportpurple] {$\times y$};
\end{scope}

\begin{scope}[shift={(5.6,0)}]

  \node[panelbox, minimum width=5.2cm, minimum height=5.8cm] at (2.3,2.55) {};
  \node[headerbox] at (2.3,5.1) {\textbf{(b) Polynomial support}};

  \foreach \i in {0,...,4} {
    \draw[gridgray, thin] (\i*1 + 0.3, 0.3) -- (\i*1 + 0.3, 4.7);
    \draw[gridgray, thin] (0.3, 0.3+\i*1) -- (4.3, 0.3+\i*1);
  }

  \foreach \i in {0,...,3} {
    \foreach \j in {0,...,3} {
      \node[qubit] at (\i*1+0.8, \j*1+0.8) {};
    }
  }


  \node[activequbit] (s0) at (1.8,1.8) {};
  \node[activequbit] (s1) at (2.8,1.8) {};
  \node[activequbit] (s2) at (0.8,2.8) {};

  \draw[supportred, thick, dashed, rounded corners=4pt]
    (1.45,1.45) -- (3.15,1.45) -- (3.15,2.15) -- (2.15,2.15) --
    (1.45,3.15) -- (0.45,3.15) -- (0.45,2.45) -- (1.45,2.45) -- cycle;
  \node[font=\tiny, supportred!60] at (2.3,1.2) {For an $X$-type operator anchored at $xy$};

  \node[above, supportred, font=\footnotesize] at (1.8,2.0) {$xy$};
  \node[above, supportred, font=\footnotesize] at (2.8,2.0) {$x^2 y$};
  \node[above, supportred, font=\footnotesize] at (0.8,3.0) {$y^2$};

  \node[supportred, font=\footnotesize, align=center] at (2.3,0)
    {stabilizer $\sim 1+x+x^{-1}y$};

  \begin{scope}[opacity=0.35]
    \node[activequbit] at (2.8,2.8) {};
    \node[activequbit] at (3.8,2.8) {};
    \node[activequbit] at (1.8,3.8) {};
  \end{scope}
  \node[font=\tiny, supportred!60] at (3.5,4.2) {(translated copy)};

\end{scope}

\end{tikzpicture}
\caption{%
  \textbf{(a)} Lattice sites are labelled by monomials $x^i y^j$;
  multiplying by $x$ (resp.\ $y$) translates one step right (resp.\ up).
  The filled blue site is the origin monomial~$1$.
  \textbf{(b)} A polynomial such as $1+x+x^{-1}y$ is a repeating local
  stencil. A Pauli $X$ on qubit labeled $xy$ acts on the highlighted site and on its translates one step
  right and one step left-and-up.
  A faded copy of the same stencil anchored at $x^2 y^2$ illustrates translation
  invariance.%
}
\label{fig:informal_dictionary_tikz}
\end{figure*}

Finally, the parity-check matrix of a \TwoDTLTI code can be expressed as a collection of Laurent polynomials in two variables. Suppose we have a 2D lattice with $n$ qubits that reside in sites that can accommodate 2 qubits each, then a possible parity-check matrix can be expressed as:
\begin{widetext}
\begin{align}
  H = \left(\begin{array}{c|c}
  H_X & 0 \\ 0 & H_Z
\end{array}\right) &= \left(\begin{array}{cc|cc}
  1 + x & 1+y & 0 & 0 \\
  0 & 0 & (1 + y)^\ast & (1 + x)^\ast
\end{array}\right) \\
&= \left(\begin{array}{cc|cc}
  1 + x & 1+y & 0 & 0 \\
  0 & 0 & 1 + \bar{y} & 1 + \bar{x}
\end{array}\right) \\
&\coloneqq \left(\begin{array}{cc|cc}
  1 + x & 1+y & 0 & 0 \\
  0 & 0 & 1 + y^{-1} & 1 + x^{-1}
\end{array}\right)
\end{align}
\end{widetext}
where $(1 + x)^\ast = 1 + x^{-1}$ and $(1 + y)^\ast = 1 + y^{-1}$. Note that the first column on each side of the partition corresponds to the qubits that corresponds to the horizontal edges of the lattice and the second column on each side of the partition corresponds to the qubits that corresponds to the vertical edges of the lattice. The asterisk denotes the antipode map, which is an involutive $\F_2$-linear map from the ring $R$ to itself that sends a Laurent polynomial $f(x,y) = \sum_{i,j} a_{ij} x^i y^j$ to $f^\ast(x,y) = \sum_{i,j} a_{ij} x^{-i} y^{-j}$. 

When we consider the dagger operation on the parity-check matrix, we are effectively taking the transpose of the matrix and applying the antipode map to each entry. An example of taking the dagger operation on the parity-check matrix $H$ is given by:
\begin{widetext}
\begin{align}
  H^\dagger = \left(\begin{array}{cc}
  H_X^\dagger & 0 \\ 0 & H_Z^\dagger
\end{array}\right) 
= \left(\begin{array}{cc}
  (1 + x)^\ast & 0 \\ 
  (1 + y)^\ast & 0 \\
  0 & 1 + y \\
  0 & 1 + x
\end{array}\right) 
= \left(\begin{array}{cc}
  1 + x^{-1} & 0 \\ 
  1 + y^{-1} & 0 \\
  0 & 1 + y \\
  0 & 1 + x
\end{array}\right) 
\end{align}
\end{widetext}

From the above example involving $H$, we see that each horizontal qubit associated to the site $(i,j)$ in the lattice is acted upon by the $X$ stabilizer checks $x^{i}y^{j} \left(1 + x\right)$, i.e. the $X$ stabilizer checks associated with the monomials $x^{i}y^{j}$ and $x^{i+1}y^{j}$. Similarly, the same associated horizontal qubit is acted upon by the $Z$ stabilizer checks $x^{i}y^{j}$ and $x^{i}y^{j-1}$. The above observation can be obtained by performing the following matrix operation:
\begin{widetext}
\begin{equation}H\left(\begin{array}{cc|cc}x^iy^j & 0 & x^iy^j & 0\end{array}\right)^\top = \left(\begin{array}{cc}x^iy^j + x^{i+1}y^j &  x^iy^j + x^iy^{j-1} \end{array}\right)^\top.\end{equation}
\end{widetext}
The same observation can be made for the vertical qubits associated with the site $(i,j)$ in the lattice. By considering $x$ and $y$ as the horizontal and vertical directions in the lattice, we see that each qubit in this CSS code is acted upon by its two nearest neighbor $X$ and $Z$ stabilizer checks.
Conversely, the support of the $X$ stabilizer at site $(i,j)$ is
obtained from the adjoint check map:
\begin{equation}
H_X^\dagger x^iy^j
=\begin{pmatrix}
x^iy^j+x^{i-1}y^j\\
x^iy^j+x^iy^{j-1}
\end{pmatrix}.
\end{equation}
The first component labels horizontal qubits and the second vertical
qubits. Multiplication by $H_X$ computes the syndrome of a $Z$ error,
whereas multiplication by $H_X^\dagger$ specifies an $X$-stabilizer
support. The analogous statements hold for $H_Z$. These check
matrices describe the square toric code.

While we have mainly focused on the case where $q = 2$, the above formulation is highly general and can be used to describe any \TwoDTLTI code with a constant number of qubits per site. The polynomial representation is very powerful as it captures the periodic structure of the code and provides a succinct description of the interactions for the code Hamiltonian. A notable member of the \TwoDTLTI code family are the bivariate bicycle (BB) codes. We include a useful summary of BB codes in Appendix~\ref{app:summary-bb-codes} for the readers' reference.

\section{The layer-decoupling decoder for \TwoDTLTI codes}
\label{sec:decoding}
In this section, we present the \emph{layer-decoupling decoder}, a matching-based decoding algorithm for \TwoDTLTI codes.
We first review the decoding pipeline at a high level, then formalize how the algebraic outputs of the decoupling process induce matching-decodable graph chain complexes. For readers who would like a beginner-friendly introduction to the decoupling theorem and its algorithmic implementation, we recommend reading Appendix~\ref{app:decoupling} before proceeding with the rest of this section. In particular, we expect the reader to be comfortable with terms like defects, charge groups, as well as the cokernel of the parity-check matrix, etc.
Finally, we show how to lift the resulting graph corrections back to a physical Pauli correction and we state finite-size correctness and time complexity guarantees.

Throughout, we assume periodic boundary conditions on a torus of linear dimensions $L_x\times L_y$ and write
\begin{equation}\Lambda \coloneqq \Z_{L_x}\times \Z_{L_y}.\end{equation}
We fix a coarse-graining parameter $b\in\N$ and assume $b\mid L_x,L_y$ when discussing finite-size decoding.
The corresponding coarse lattice is
\begin{equation}\Lambda_b \coloneqq \Z_{L_x/b}\times \Z_{L_y/b}.\end{equation}

For concreteness, we focus on decoding $Z$-type errors; the $X$-type decoder is analogous.

\subsection{Matching decoder framework}
\label{sec:decoding-framework}
The decoding algorithm is designed to be a three-step process: 
\begin{itemize}
    \item We virtually decouple the code into independent copies of TCs and trivial product states.
    \item We construct a decoding graph for each of these copies based on the syndrome information obtained from the measurements and obtain their corresponding error correction operators using a suitable matching decoder such as the MWPM algorithm or the Union-Find algorithm.
    \item We lift the error correction operators from the TCs back to the original \TwoDTLTI code.
\end{itemize}

The key step is the decoupling process because it reduces the relevant decoding problem to a graph matching problem which is akin to the approach taken by Delfosse for the color code~\cite{delfosse2014decoding}.
The decoupling algorithm in Appendix~\ref{app:decoupling} provides two algebraic objects $U$ and $V$, which are matrices over the bivariate Laurent polynomial ring $R' = \F_2[x^{\pm b}, y^{\pm b}]$.
The matrix $U$ records the valid symplectic row operations applied to the coarse-grained parity-check matrix $H'$, while $V$ records the valid symplectic column operations.
Let the following ordered sets be the sequence of row and column operations that we have performed on $H'$:
\begin{align*}
    T_{\row} &= \left\{r_1, r_2, \ldots, r_{m_{\row}}\right\},\\
    &\hspace{1em} r_i \in R^{\prime\,n_r \times n_r},\quad i \in [m_{\row}],\\
    T_{\col} &= \left\{c_1, c_2, \ldots, c_{m_{\col}}\right\},\\
    &\hspace{1em} c_j \in R^{\prime\,n_c \times n_c},\quad j \in [m_{\col}].
\end{align*}
where $n_r$ and $n_c$ are the number of rows and columns in the coarse-grained parity-check matrix $H'$, respectively. Suppose we denote the final parity-check matrix after the decoupling process 
\begin{equation}\tilde{H} = \begin{pmatrix}
    \tilde{H}_X & 0\\
    0 & \tilde{H}_Z
\end{pmatrix}\end{equation}
where $\tilde{H}_X$ and $\tilde{H}_Z$ are the parity-check matrices for the X-type and Z-type stabilizer checks, respectively. Then, from the associativity of matrix multiplication, we can express the final parity-check matrix as follows:
\begin{equation}\tilde{H} = \left(\prod_{i=1}^{m_{\row}} r_i\right) H' \left(\prod_{j=1}^{m_{\col}} c_j\right) \eqqcolon U H' V.\end{equation}
Because the parity-check matrix is block-diagonal, we can express $U$ and $V$ as follows:
\begin{equation}U = \begin{pmatrix}
    U_X & 0\\
    0 & U_Z
\end{pmatrix},\quad V = \begin{pmatrix}
    V_X & 0\\
    0 & V_Z
\end{pmatrix}\end{equation}
such that $U_X H'_X V_X = \tilde{H}_X$ and $U_Z H'_Z V_Z = \tilde{H}_Z$. The matrices $U_X$, $U_Z$, $V_X$, and $V_Z$ are all square matrices over $R'$ of size $\frac{n_r}{2} \times \frac{n_r}{2}$ and $\frac{n_c}{2} \times \frac{n_c}{2}$, respectively.
Of course, there could be some additional steps of coarse-graining that we have performed in between row and column operations during the decoupling process, but we can always express all the intermediate row and column operations in the final coarse-grained base ring $\tilde{R}$ by replacing the ring elements in $R'$ with their corresponding elements in $\tilde{R}$ using the generator matrices formulated in Appendix~\ref{subsec:coarse_graining}. Lastly, we note that $U$ and $V$ are invertible matrices over $\tilde{R}$ because they are obtained from a sequence of valid symplectic row and column operations on the coarse-grained parity-check matrix $H'$.

We summarize the finite-size decoding procedure as an explicit algorithm.

\begin{algorithm}[H]
\caption{Layer-decoupling decoder for $Z$-type errors (finite size)}
\label{alg:decoupling-decoder}
\begin{algorithmic}[1]
    \Require{Measured $X$-check syndrome $s_X$ of a $Z$ error; fixed decoupling maps.}
    \Ensure{$Z$-type Pauli correction $\widehat{E}$.}
    \State Compute TC sector and product state syndromes $\tilde{s}^{(i)}$ and $\tilde{s}^A$ via the syndrome mapping (Definition~\ref{def:syndrome-mapping-decoupling}).
    \State Assign unit weight to every elementary toric edge.
    \ForAll{TC sectors $i\in[r]$}
        \State Use MWPM to obtain a minimum-weight chain $\tilde{e}^{(i)}_1\in\tilde{C}^{(i)}_1$ with boundary $\tilde{s}^{(i)}$.
    \EndFor
    \State Compute a local correction $\tilde{e}^A_1$ on the trivial product-state sector.
    \State Lift and combine: $\tilde{e}_1 \coloneqq \bigoplus_{i=1}^{r} \tilde{e}_1^{(i)} \oplus \tilde{e}^A_1$.
    \State Compute final correction $\widehat{E}\coloneqq \phi_1^{-1}(\tilde{e}_1)$ where $\phi_1^{-1}$ is the inverse of the decoupling map.
    \State \Return $\widehat{E}$.
\end{algorithmic}
\end{algorithm}

\subsection{Morphism of chain complexes}
\label{sec:morphism-chain-complexes}
In this subsection, we first discuss how we can obtain several chain complexes from the decoupling process. Subsequently, we show how to construct a bijective morphism from the direct sum of these chain complexes to a graph chain complex where we can perform matching decoding.

For $Z$-error decoding, let $t_X$ and $t_Z$ denote the numbers of
$X$- and $Z$-stabilizer generator types per lattice site, and let
$q$ be the number of qubits per site. The physical CSS chain complex is
\begin{equation}
\calC:\quad R^{t_Z}\xrightarrow{H_Z^\dagger}R^q
\xrightarrow{H_X}R^{t_X}.
\end{equation}
We write $C_2=R^{t_Z}$, $C_1=R^q$, $C_0=R^{t_X}$,
$\partial_2=H_Z^\dagger$, and $\partial_1=H_X$.
Here $C_1$ records $Z$ errors, and CSS commutation gives
$H_XH_Z^\dagger=0$. After imposing the finite-torus periodic
boundary conditions, the $Z$ logical operators are represented by
$H_1(\calC)=\ker\partial_1/\im\partial_2$.

Let $r$ be the number of TC sectors obtained from the decoupling of the code with parity-check matrix $H$.
Concretely, for the full two-sector excitation matrix $H$,
$r=\tfrac12\dim_{\F_2}\calT(\coker\,H)$, as in
Theorem~\ref{thm:decoupling_count}.
Each toric copy contributes two independent charge types.
For the subsequent discussion, we let $\tilde{R}$ be the final coarse-grained base ring obtained from the decoupling process and indicate all associated quantities with a tilde.
For each of these $r$ TC sectors, we can construct a chain complex $\tilde{\calC}^{(i)}$ for $i \in [r]$ as follows:
\begin{equation}
\tilde{\calC}^{(i)}:\quad
\tilde R\xrightarrow{H_{TC,Z}^\dagger}\tilde R^2
\xrightarrow{H_{TC,X}}\tilde R.
\end{equation}
Here $\tilde C_2^{(i)}=\tilde R$, $\tilde C_1^{(i)}=\tilde R^2$,
$\tilde C_0^{(i)}=\tilde R$, and the displayed maps are
$\tilde\partial_2^{TC},\tilde\partial_1^{TC}$, respectively,
where $H_{TC, X/Z}$ is the $X/Z$-type parity-check matrix of the TC. In addition, we can also construct a chain complex $\tilde{\calA}$ for the trivial product state where each of the qubits is either stabilized by an $X$-type or a $Z$-type stabilizer check. Let $\phi = (\phi_2, \phi_1, \phi_0)$ be a vector-space isomorphism between the chain complexes $\calC$ and the direct sum of the chain complexes $\tilde{\calC}^{(i)}$ and $\tilde{\calA}$. Note that because $\phi$ is bijective, there exists $\phi^{-1}$ maps from the direct sum of the chain complexes $\tilde{\calC}^{(i)}$ and $\tilde{\calA}$ to the chain complex $\calC$.  

Below is a diagram showing the vector-space isomorphism $\phi$ between the chain complex $\calC$ and the direct sum of the chain complexes $\tilde{\calC}^{(i)}$ and $\tilde{\calA}$:
\begin{center}
\begin{tikzcd}
    C_2 \arrow[r, "\partial_2"] \arrow[d, "\phi_2"] & C_1 \arrow[r, "\partial_1"] \arrow[d, "\phi_1"] & C_0 \arrow[d, "\phi_0"]\\
    \tilde{C}^{(1)}_2 \arrow[r, "\tilde{\partial}^{TC}_2"] & \tilde{C}^{(1)}_1 \arrow[r, "\tilde{\partial}^{TC}_1"] & \tilde{C}^{(1)}_0 \\[-2em]
    \oplus & \oplus & \oplus  \\[-2em]
    \tilde{C}^{(2)}_2 \arrow[r, "\tilde{\partial}^{TC}_2"] & \tilde{C}^{(2)}_1 \arrow[r, "\tilde{\partial}^{TC}_1"] & \tilde{C}^{(2)}_0 \\[-2em]
    \oplus & \oplus & \oplus \\[-2.5em]
    \vdots & \vdots & \vdots \\[-2em]
    \oplus & \oplus & \oplus \\[-2em]
    \tilde{C}^{(r)}_2 \arrow[r, "\tilde{\partial}^{TC}_2"] & \tilde{C}^{(r)}_1 \arrow[r, "\tilde{\partial}^{TC}_1"] & \tilde{C}^{(r)}_0 \\[-2em]
    \oplus & \oplus & \oplus \\[-2em]
    \tilde{A}_2 \arrow[r, "\tilde{\partial}^A_2"] & \tilde{A}_1 \arrow[r, "\tilde{\partial}^A_1"] & \tilde{A}_0
\end{tikzcd}
\end{center}
For the sake of brevity, we denote the direct sum of the chain complexes $\tilde{\calC}^{(i)}$ and $\tilde{\calA}$ as $\tilde{\calC}$.
In addition, we denote the boundary operators of the direct sum of chain complexes by
\begin{align}
\tilde{\partial}_2 &= \left(\bigoplus_{i=1}^r \tilde{\partial}_2^{TC}\right) \oplus \tilde{\partial}_2^A,\\
\tilde{\partial}_1 &= \left(\bigoplus_{i=1}^r \tilde{\partial}_1^{TC}\right) \oplus \tilde{\partial}_1^A.
\end{align}
where $\tilde{\partial}_2^{TC}$ and $\tilde{\partial}_1^{TC}$ are the boundary operators of the chain complex $\tilde{\calC}^{(i)}$ and $\tilde{\partial}_2^A$ and $\tilde{\partial}_1^A$ are the boundary operators of the chain complex $\tilde{\calA}$.

At this point, it might already be clear how the vector-space isomorphisms $\phi$ are related to the algebraic objects $U$ and $V$ that we have obtained from the decoupling process in Appendix~\ref{app:decoupling}. To be concrete, we have the following relations:
\begin{align}
    \phi_2 &= U_Z^{-\dagger}\,:\,C_2 \to \bigoplus_{i=1}^r \tilde{C}^{(i)}_2 \oplus \tilde{A}_2, \\
    \phi_1 &= V_X^{-1}=V_Z^\dagger
    \,:\, C_1 \to \bigoplus_{i=1}^r \tilde{C}^{(i)}_1
    \oplus \tilde{A}_1, \\
    \phi_0 &= U_X\,:\,C_0 \to \bigoplus_{i=1}^r \tilde{C}^{(i)}_0 \oplus \tilde{A}_0,
\end{align}
where $A^{-\dagger}$ denotes the inverse of the adjoint, $A^{-\dagger}=(A^\dagger)^{-1}=(A^{-1})^\dagger$.

The map $\phi_0$ transforms $X$-check syndromes and $\phi_1$
transforms $Z$-error vectors. The map $\phi_2$ transforms
coefficients of $Z$-stabilizer products; it does not transform
measured $Z$-check syndromes. For the dual $X$-error problem,
the error and syndrome maps are $V_Z^{-1}$ and $U_Z$, respectively.

In fact, the vector-space isomorphism $\phi$ can be shown to be a chain isomorphism between the chain complex $\calC$ and the direct sum of the chain complexes $\tilde{\calC}^{(i)}$ and $\tilde{\calA}$ which we show formally in Appendix~\ref{app:layer-decoupling-decoder-lemmas} as Lemma~\ref{lem:chain-isomorphism}. 
Now, we define a set of projections $\pi^{(i)} = (\pi^{(i)}_2, \pi^{(i)}_1, \pi^{(i)}_0)$ for $i \in [r]$ and $\pi^A = (\pi^A_2, \pi^A_1, \pi^A_0)$ as follows:
\begin{widetext}
\[
\begin{tikzcd}[baseline=(current bounding box.center)]
    \tilde{C}_2 \arrow[r, "\tilde{\partial}_2"] \arrow[d, "\pi_2^{(i)}"] & \tilde{C}_1 \arrow[r, "\tilde{\partial}_1"] \arrow[d, "\pi_1^{(i)}"] & \tilde{C}_0 \arrow[d, "\pi_0^{(i)}"]\\
    \tilde{C}^{(i)}_2 \arrow[r, "\tilde{\partial}^{TC}_2"] & \tilde{C}^{(i)}_1 \arrow[r, "\tilde{\partial}^{TC}_1"] & \tilde{C}^{(i)}_0
\end{tikzcd}
\qquad
\begin{tikzcd}[baseline=(current bounding box.center)]
    \tilde{C}_2 \arrow[r, "\tilde{\partial}_2"] \arrow[d, "\pi_2^{A}"] & \tilde{C}_1 \arrow[r, "\tilde{\partial}_1"] \arrow[d, "\pi_1^{A}"] & \tilde{C}_0 \arrow[d, "\pi_0^{A}"]\\
    \tilde{A}_2 \arrow[r, "\tilde{\partial}^A_2"] & \tilde{A}_1 \arrow[r, "\tilde{\partial}^A_1"] & \tilde{A}_0
\end{tikzcd}\]
\end{widetext}

These projection maps effectively project the direct sum of chain complexes into one of the summands. It is straightforward to see that the projections $\pi^{(i)}$ and $\pi^A$ are chain maps (i.e., morphisms of 3-term chain complexes). Again, we formally state this as Lemma~\ref{lem:projection-chain-maps} in Appendix~\ref{app:layer-decoupling-decoder-lemmas}.

By projecting onto one of the TC summands $\tilde{\calC}^{(i)}$, we now obtain a graph chain complex
\begin{equation}
\tilde{\calC}^{(i)}\,:\,
\left(\tilde{C}^{(i)}_2 \xrightarrow{\tilde{\partial}^{(i)}_2}
\tilde{C}^{(i)}_1 \xrightarrow{\tilde{\partial}^{(i)}_1}
\tilde{C}^{(i)}_0\right)
\end{equation}
whose syndrome is decoded by MWPM in our analysis.

\subsection{Mapping the noise model for the TCs}
\label{sec:noise-model-mapping}
Because of the decoupling map, we have the option of reweighting our edges in the decoding graph according to the mapping of the noise to improve the decoder's performance. Let $p_a\in[0,1]$ bound
the marginal probability of a $Z$ error on every translate of qubit
type $a$ in the fixed coarse cell.
\begin{definition}[Noise model mapping]
\label{def:noise-model-mapping-decoupling}
For a Laurent-polynomial matrix $A$, let $\mathsf N(A)_{uv}$ be
the number of nonzero monomials in its entry $A_{uv}$.
Write $J^{(i)}=\pi_1^{(i)}\phi_1$ and $J^A=\pi_1^A\phi_1$.
Define the componentwise marginal bounds
\begin{align}
\tilde p^{(i)}_u
 &=\min\{1,[\mathsf N(J^{(i)})p]_u\},\\
\tilde p^A_u
 &=\min\{1,[\mathsf N(J^A)p]_u\}.
\end{align}
Here the vector $p$ has one entry per physical qubit type, and
$\tilde p^{(i)}$ has two entries per square-toric cell in the
Pauli sector being decoded.
\end{definition}
An error on some qubit in a virtual TC sector requires at least one of its physical input qubits from the 2D TTI code to have an error. Thus, it is possible to apply the union bound to upper bound the probability of an error on a qubit in a virtual TC sector.
These bounds may inform us on what edge weights to use in the practical case but they do not specify a joint distribution or imply independence within a sector.
The theorem below uses unit weights and a joint local-stochastic support bound. The same construction applies to $X$ errors using the dual error map.

\subsection{Mapping the syndrome configuration for the TCs}
\label{sec:syndrome-mapping}
All syndrome vectors below include every check on the finite torus.
The Laurent-polynomial matrices act by replacing each monomial
with the corresponding translation operator on the finite torus. For error components $e_Z,e_X$, write
$s_X=H'_Xe_Z$ and $s_Z=H'_Ze_X$.
\begin{definition}[Syndrome mapping]
\label{def:syndrome-mapping-decoupling}
For $Z$-error decoding, the toric-sector and product-state syndromes are
\begin{equation}
\tilde{s}^{(i)}=\pi_0^{(i)}U_Xs_X,\qquad
\tilde{s}^A=\pi_0^AU_Xs_X.
\end{equation}
Each $\tilde{s}^{(i)}$ has one coordinate per coarse toric vertex.
For $X$-error decoding, use $U_Zs_Z$ and the corresponding
projections onto the dual toric checks and product-state checks.
The full transformed measurement record is
\begin{equation}
(\tilde{s}_X,\tilde{s}_Z)=(U_Xs_X,U_Zs_Z).
\end{equation}
In particular, the measured $Z$-check syndrome transforms by
$U_Z=\phi_2^{-\dagger}$.
\end{definition}

\subsection{Error correction operators and lifting}
\label{sec:lifting-error-correction-operators}
In this subsection, we show how we can obtain the error correction operators for each of the TCs and lift them back to the original \TwoDTLTI code. The lifting process is done by using the chain isomorphisms $\phi$ defined in Section~\ref{sec:morphism-chain-complexes} and the projections $\pi^{(i)}$ defined in Section~\ref{sec:noise-model-mapping}.

Given the syndrome configuration $\tilde{s}^{(i)}$ for each TC sector $i \in [r]$ defined in Definition~\ref{def:syndrome-mapping-decoupling}, we can construct a decoding graph for each TC using the syndrome information.
The decoding graph is constructed by associating vertices with the syndrome bits and edges with the possible error correction operators that can be applied to correct the errors.
For the guarantees in
Theorem~\ref{thm:decoupling-decoder-finite-size}, each elementary
toric edge has unit weight, and MWPM returns a minimum-weight
correction operator $\tilde{e}^{(i)}_1$ with boundary $\tilde{s}^{(i)}$.
Weighted MWPM using the noise model mapping or/and union-find decoding may be used as practical
alternatives.
The product-state syndrome is corrected locally, and the resulting
corrections are lifted by $\phi_1^{-1}$.
The lifting lemma below holds independently of the choice of the matching decoder.

Let us denote the following:
\begin{equation}\tilde{e}_1 = \bigoplus_{i=1}^r \tilde{e}_1^{(i)} \oplus \tilde{e}_1^A \in \bigoplus_{i=1}^r \tilde{C}_1^{(i)} \oplus \tilde{A}_1\end{equation}
as the direct sum of the error correction operators for all TCs and the trivial product state. The lifted error correction operator $e_1$ in the original code is then given by:
\begin{equation}
e_1 = \phi_1^{-1}(\tilde{e}_1) \in C_1
\end{equation}
where $\phi_1^{-1}$ is the inverse of the isomorphism defined in Section~\ref{sec:morphism-chain-complexes}.

This process ensures that the error correction operators, which are valid in the decoupled TCs and trivial product state, are mapped back to valid Pauli operators acting on the qubits of the original code. The correctness of this lifting follows from the fact that $\phi_1$ is an isomorphism, and thus $\phi_1^{-1}$ preserves the structure of the error operators.

\begin{lemma}[Correctness of lifting]
Fix a physical
$Z$ error $z_1$ and write $\tilde{z}_1=\phi_1z_1$.
For a sector correction operator $\tilde{e}_1$, let
$e_1=\phi_1^{-1}\tilde{e}_1$.
Then $e_1$ corrects $z_1$ if and only if the residual in every
TC and product-state summand is a boundary.
\end{lemma}
\begin{proof}
    Let us begin by proving the forward direction. Assume without loss of generality that we are correcting for phase-flip errors. Suppose $e_1 = \phi_1^{-1}(\tilde{e}_1)$ for some $\tilde{e}_1$ is a valid correction for the original code, i.e., for the fixed error $z_1 \in C_1$, there exists some $s_Z$ such that $\partial_2(s_Z) = z_1 + e_1$. Suppose each of the $r$ TC sectors indexed by $i \in [r]$ and the ancilla product state have phase flip errors $\tilde{z}^{(i)}_1 \in \tilde{C}_1^{(i)}$ and $\tilde{z}^A_1 \in \tilde{A}_1$, respectively. We denote $\tilde{z}_1$ as the direct sum of the phase flip errors, i.e., $\tilde{z}_1 = \bigoplus_{i=1}^r \tilde{z}_1^{(i)} \oplus \tilde{z}_1^A = \phi_1(z_1)$. Applying the isomorphism $\phi_1$, we get
    \begin{equation}
    \phi_1(z_1 + e_1) = \phi_1(z_1) + \phi_1(e_1) = \tilde{z}_1 + \tilde{e}_1.
    \end{equation} 
    Our goal is to show that $\tilde{z}_1 + \tilde{e}_1 \in \im \tilde{\partial}_2$, i.e., there exists some $\tilde{s}_Z = \bigoplus_{i = 1}^r \tilde{s}_Z^{(i)} \oplus \tilde{s}_Z^A$ such that $\tilde{\partial}_2(\tilde{s}_Z) = \tilde{z}_1 + \tilde{e}_1$. Setting $\tilde{s}_Z=\phi_2(s_Z)$,
    \begin{align}
    \tilde{z}_1 + \tilde{e}_1 &= \phi_1(z_1 + e_1) \\
    &= \phi_1\left(\partial_2(s_Z)\right) \\
    &= \phi_1\left(\partial_2\left(\phi_2^{-1}(\tilde{s}_Z)\right)\right) \\
    &= \tilde{\partial}_2(\tilde{s}_Z)
    \end{align}
    where the last equality follows from Lemma~\ref{lem:chain-isomorphism} which gives us $\tilde{\partial}_2 \phi_2 = \phi_1\partial_2$. Thus, we have shown that $\tilde{z}_1 + \tilde{e}_1 \in \im \tilde{\partial}_2$, which implies that $\tilde{e}^{(i)}_1$ is a valid correction for the $i$\textsuperscript{th} TC and $\tilde{e}^A_1$ is a valid correction for the trivial product state.

    Now, let us prove the backward direction. The proof is essentially the same but we supply it for the sake of being concrete. Again, assume without loss of generality that we are correcting for phase-flip errors. Suppose each of the $r$ TC sectors indexed by $i \in [r]$ and the ancilla product state have phase flip errors $\tilde{z}^{(i)}_1 \in \tilde{C}_1^{(i)}$ and $\tilde{z}^A_1 \in \tilde{A}_1$, respectively. We denote $\tilde{z}_1$ as the direct sum of the phase flip errors, i.e., $\tilde{z}_1 = \bigoplus_{i=1}^r \tilde{z}_1^{(i)} \oplus \tilde{z}_1^A$. In other words, the original $Z$ error configuration is given by $\phi_1^{-1}(\tilde{z}_1)$.
   Let $\tilde{e}^{(i)}_1 \in \tilde{C}_1^{(i)}$ be a valid $Z$ correction for the $i$\textsuperscript{th} TC and $\tilde{e}^A_1 \in \tilde{A}_1$ be a valid $Z$ correction for the trivial product state for all $i \in [r]$. Our goal is to show that the lifted operator $e_1 = \phi_1^{-1}(\tilde{e}_1)$ is a valid correction for the original code, i.e., $\phi_1^{-1}(\tilde{e}_1 + \tilde{z}_1) \in \im \partial_2$.
    From the validity of the TC and product state corrections, we have $\tilde{z}^{(i)}_1 + \tilde{e}^{(i)}_1 \in \im \tilde{\partial}_2^{TC}$ and $\tilde{z}^A_1 + \tilde{e}^A_1 \in \im \tilde{\partial}_2^A$. Thus, there exists some $\tilde{s}_Z^{(i)} \in \tilde{C}_2^{(i)}$ and $\tilde{s}_Z^A \in \tilde{A}_2$ such that $\tilde{\partial}_2^{TC}(\tilde{s}_Z^{(i)}) = \tilde{z}^{(i)}_1 + \tilde{e}^{(i)}_1$ and $\tilde{\partial}_2^A(\tilde{s}_Z^A) = \tilde{z}^A_1 + \tilde{e}^A_1$. From the invertibility of the isomorphism $\phi_2$, there exists some $s_Z \in C_2$ such that $\phi_2(s_Z) = \left(\bigoplus_{i=1}^r \tilde{s}_Z^{(i)}\right) \oplus \tilde{s}_Z^A$. We now claim that $\phi_1^{-1}(\tilde{e}_1 + \tilde{z}_1) = \partial_2(s_Z)$. To see this, we compute the following:
\begin{widetext}
\begin{align}
    \phi_1^{-1}(\tilde{e}_1 + \tilde{z}_1)
    &= \phi_1^{-1}\left(
    \left(\bigoplus_{i=1}^r \tilde{e}^{(i)}_1 + \tilde{z}^{(i)}_1\right)
    \oplus \left(\tilde{e}^A_1 + \tilde{z}^A_1\right)\right)\\
    &= \phi^{-1}_1\Biggl(
    \bigoplus_{i = 1}^r \tilde{\partial}_2^{TC} \tilde{s}_Z^{(i)}
    \oplus \tilde{\partial}_2^A \tilde{s}_Z^A\Biggr) \\
    &= \phi^{-1}_1\left(
    \tilde{\partial}_2\left(\bigoplus_{i=1}^r \tilde{s}_Z^{(i)}
    \oplus \tilde{s}_Z^A\right)\right)\\
    &= \left(\phi_1^{-1}\tilde{\partial}_2\phi_2\right)(s_Z)\\
    &= \partial_2(s_Z),
\end{align}
\end{widetext}
where the last equality follows from Lemma~\ref{lem:chain-isomorphism} which gives us $\tilde{\partial}_2 \phi_2 = \phi_1\partial_2$.
\end{proof}

\subsection{Layer-decoupling decoding algorithm for finite-size \TwoDTLTI codes}
\label{sec:decoupling-decoding-algorithm}
In this subsection, we detail how we can generalize the decoding algorithm presented for the 2D code in the asymptotic limit to finite-size \TwoDTLTI codes. In addition, there are some edge cases that have to be taken into account when considering the time complexity of the algorithm for a finite-size code.

Before we discuss the subtleties associated with finite-size codes, we summarize the error correction performance of the layer-decoupling decoder for finite-size \TwoDTLTI codes in the following theorems.

\begin{theorem}[Layer-decoupling decoder for \TwoDTLTI codes]
\label{thm:decoupling-decoder-finite-size}
Let $Q$ be a \TwoDTLTI code on an $L_x\times L_y$ torus, with $b\mid L_x,L_y$ and $L_x,L_y\geq3b$. Assume that the fixed decoupling maps induce the chain isomorphism of Lemma~\ref{lem:chain-isomorphism} on the physical chain complex for the Pauli sector being decoded, with bounded support radius and a locally correctable product-state sector. Write
\begin{equation}
J^{(i)}=\pi_1^{(i)}\phi_1,\qquad d_{TC}=\min\{L_x/b,L_y/b\}.
\end{equation}
Fix constants $k_{\mathrm{lay}},m_{\mathrm{lay}}\geq1$, independent of system size, with
\begin{align}
|J^{(i)}a|&\leq k_{\mathrm{lay}},\\
\#\{a:e\in\operatorname{supp}J^{(i)}a\}&\leq m_{\mathrm{lay}},
\end{align}
for every physical basis single-qubit error $a$, toric edge $e$, and sector $i$. With standard MWPM in each sector, the layer-decoupling decoder corrects every error in this Pauli sector of weight at most
\begin{equation}
\left\lfloor\frac{d_{TC}-1}{2k_{\mathrm{lay}}}\right\rfloor.
\end{equation}

Suppose the physical error $E$ satisfies the local-stochastic condition
\begin{equation}
\Pr[A\subseteq\operatorname{supp}E]\leq p^{|A|}
\end{equation}
for some $p\in[0,1]$ and every set $A$ of physical qubits.
Then each induced toric edge support is local-stochastic with parameter
\begin{equation}q_{\mathrm{lay}}=m_{\mathrm{lay}}p^{1/k_{\mathrm{lay}}}.
\end{equation}
Independent physical qubit errors with rates bounded by $p$ are a special case. For the degree-four spatial toric graphs, let $N_V$ be the maximum number of vertices in a sector. Then
\begin{equation}\label{eq:layer-decoupling-peierls-bound}
\Pr[\mathrm{failure}]
\leq c_1rN_V\sum_{\ell\geq d_{TC}}
(6\sqrt{q_{\mathrm{lay}}})^\ell,
\end{equation}
where $c_1$ is an absolute constant. If $6\sqrt{q_{\mathrm{lay}}}<1$ and $\log N_V=o(d_{TC})$, this is exponentially small in $d_{TC}$.
\end{theorem}

\begin{proof}
The constant-depth CNOT circuit on qubits with bounded connectivity gives finite constants $k_{\mathrm{lay}}$ and $m_{\mathrm{lay}}$. For a physical error $z_1=\sum_{a\in\operatorname{supp}z_1}a$, linearity and subadditivity give
\begin{equation}
|\tilde z_1^{(i)}|=|J^{(i)}z_1|
\leq k_{\mathrm{lay}}|z_1|.
\end{equation}
Under the stated adversarial bound this is less than $d_{TC}/2$. A minimum-weight recovery $\tilde e_1^{(i)}$ has weight at most $|\tilde z_1^{(i)}|$, so the residual $\tilde z_1^{(i)}+\tilde e_1^{(i)}$ that has no syndome has weight less than $d_{TC}$ and is necessarily a TC stabilizer. A local correction makes the product-state residual a stabilizer as well. Hence the direct-sum residual lies in $\im\tilde\partial_2$, and the chain isomorphism gives
\begin{equation}
z_1+\phi_1^{-1}(\tilde e_1)\in\im\partial_2.
\end{equation}

For the stochastic argument, fix a set $B$ of induced toric edges. For these edges to be present, it requires a physical basis error on the 2D TTI code for every edge even after mod-two cancellations. There are at most $m_{\mathrm{lay}}^{|B|}$ possible error assignments for the 2D TTI code, each requiring at least $|B|/k_{\mathrm{lay}}$ distinct physical errors. Local stochasticity and a union bound therefore yield
\begin{equation}
\Pr[B\subseteq\operatorname{supp}J^{(i)}z_1]
\leq(m_{\mathrm{lay}}p^{1/k_{\mathrm{lay}}})^{|B|}.
\end{equation}
This argument does not require independent toric edges.

If the correction on the TC fails, the residual contains a simple nontrivial cycle of length $\ell\geq d_{TC}$. Minimality of the correction requires at least $\ell/2$ actual error edges on the cycle. The local-stochastic bound and the at most $2^\ell$ subsets of cycle edges bound this event by $2^\ell q_{\mathrm{lay}}^{\ell/2}$. There are at most $c_1N_V3^\ell$ length-$\ell$ simple cycles in a degree-four graph. Summing over lengths and the $r$ sectors proves    \eqref{eq:layer-decoupling-peierls-bound}.
\end{proof}

For full Pauli noise, write the binary components as $E_X,E_Z$,
with $\operatorname{supp}E=\operatorname{supp}E_X\cup
\operatorname{supp}E_Z$ and physical weight $|E|=|\operatorname{supp}E|$.
The same local-stochastic bound holds for each component because
its support is contained in $\operatorname{supp}E$.
Taking $k_{\mathrm{lay}}$ and $m_{\mathrm{lay}}$ to be the respective
maxima over the two components, the same argument for the number of correctable faults applies
to $|E|$. A union bound over the two decoder failures gives twice
the right-hand side of~\eqref{eq:layer-decoupling-peierls-bound}
as a bound on the total failure probability.
A $Y$ error would contribute to both components and we do not have to assume independence.

For the fixed code family with $r\geq1$, choose
$w_{\mathrm{inv}}\geq1$ bounding the physical weight of the inverse
image of any toric basis edge, taking the maximum over both Pauli
sectors. This constant is independent of system size because the
inverse maps are represented by fixed Laurent-polynomial matrices,
whose entries contain finitely many monomials.
Embed a shortest nontrivial toric loop in the direct sum, with zero
components in the other summands. Its inverse image is a nontrivial
physical logical operator: a physical boundary would map to a
boundary under the chain isomorphism.
Its weight is at most $w_{\mathrm{inv}}d_{TC}$.
Thus the physical code distance $d_L=\min\{d_X,d_Z\}$ satisfies
\begin{equation}\label{eq:layer-physical-distance-comparison}
d_L\leq w_{\mathrm{inv}}d_{TC}.
\end{equation}
Together with Theorem~\ref{thm:decoupling-decoder-finite-size}, this
gives the sufficient condition
\begin{equation}
2k_{\mathrm{lay}}w_{\mathrm{inv}}|E|<d_L
\end{equation}
for correcting an error $E$ in the Pauli sector being decoded.
This establishes a constant-fraction guarantee in the physical
distance, without asserting that the bound is tight.

For a finite-size \TwoDTLTI code defined on a lattice of size $L_x \times L_y$ with periodic boundary conditions, we can directly apply the decoupling decoding algorithm presented without any modifications if $L_x$ and $L_y$ are integer multiples of $b$. On the other hand, if $L_x$ and/or $L_y$ are not integer multiples of $b$, the finite-size code can have very different numbers of logical qubits. In Ref.~\cite{chen2025anyon}, it was shown how one can derive the number of logical qubits for such finite-size BB codes based on their polynomials as well as the lattice dimensions. The gross code is an example of a code that does not satisfy this condition. We show in Section~\ref{sec:gross} how we can work with a different coarse-graining parameter that divides the lattice size and use different short strings to perform matching.

The guaranteed adversarial radius depends on both the coarse toric distance $d_{TC}$ and the forward spread $k_{\mathrm{lay}}$. In particular, a nonzero guaranteed radius requires $d_{TC}\geq2k_{\mathrm{lay}}+1$. These constants depend on the chosen fixed local reduction, and the theorem does not identify the physical distance with the coarse toric distance.

The approximately $10.3\%$ MWPM threshold for independent toric-code bit-flip or phase-flip noise is a useful numerical benchmark~\cite{dennis2002topological,wang2003confinement}.The sufficient criterion proved here is $6\sqrt{m_{\mathrm{lay}}p^{1/k_{\mathrm{lay}}}}<1$. Such estimates are often deliberately conservative. Our choice of the local decoupling map can affect the extent of correlated errors so their practical performance should be assessed using the resulting joint noise model.

We now proceed to state the time complexity of the layer-decoupling decoder for finite-size \TwoDTLTI codes. 

\begin{theorem}[Time complexity of layer-decoupling decoder for finite-size \TwoDTLTI codes]
\label{thm:time-complexity-decoupling-decoder-finite-size}
Under the hypotheses of Theorem~\ref{thm:decoupling-decoder-finite-size},
suppose the fixed local reduction maps and their inverses have been
precomputed. Let $n_L$ be the number of physical qubits and let $N_i$
be the total number of vertices and edges in toric matching graph $i$.
For a fixed code family, $n_L=\Theta(L_xL_y)$, $r=O(1)$, and
$N_i=O(n_L)$. The sequential decoding time for one Pauli component is
\begin{equation}
O\!\left(n_L+\sum_{i=1}^{r}\operatorname{MATCHING}(N_i)\right),
\end{equation}
where $\operatorname{MATCHING}(N_i)$ includes all steps of the chosen
MWPM/UF implementation, including shortest-path computations and
reconstruction of the recovery chain. Decoding both Pauli components
requires at most $2r$ such calls and has the same asymptotic bound.
\end{theorem}

\begin{proof}
The time complexity of the layer-decoupling decoder can be analyzed by considering the different steps involved in the decoding process. Note that $b$ and $r$ are constants independent of $L_x,L_y$ for a fixed translationally-invariant code family. The main steps are as follows:
\begin{enumerate}
    \item \textbf{Decoupling process:} The decoupling process involves performing a series of row and column operations on the parity-check matrix $H'$ to transform it into a block-diagonal form. The reduction maps are precomputed once for the fixed code family. Constructing them is not part of the per-instance decoding time.
    
    \item \textbf{Mapping noise model and syndrome:} Mapping the noise model and syndrome configuration to each TC involves applying the chain isomorphisms and projections defined in Sections~\ref{sec:noise-model-mapping} and~\ref{sec:syndrome-mapping}. This step takes $O(1)$ time per unit cell and there are $\frac{L_x L_y}{b^2}$ unit cells, leading to a total time complexity of $O(L_x L_y)$ for this step. With sufficiently many processors, the local calculations have constant parallel depth, while their total work remains $O(n_L)$.
    
    \item \textbf{Decoding each TC:} Each TC has a lattice size of $(L_x / b) \times (L_y / b)$. The time complexity for decoding one TC instance using the specified MWPM implementation is $\operatorname{MATCHING}(N_i)$, including shortest paths and recovery-chain reconstruction. Summing over the $r$ TC sectors gives $\sum_i\operatorname{MATCHING}(N_i)$.
    
    \item \textbf{Lifting error correction operators:} Lifting the error correction operators from each TC back to the original code involves applying the inverse of the chain isomorphisms. This step, including correction of the product-state sector, takes $O(1)$ time per unit cell, leading to a total time complexity of $O(L_x L_y)=O(n_L)$.
\end{enumerate} 

Combining the time complexities of all steps gives
\begin{equation}
O\!\left(n_L+\sum_i\operatorname{MATCHING}(N_i)\right).
\end{equation}
\end{proof}

Note that many of the subroutines described in the time-complexity analysis can, in fact, be parallelized given sufficiently many processors.

\section{The cell-matching decoder for \TwoDTLTI Codes}
\label{sec:cellular_decoding}

In the previous section, we described how the decoupling formalism can relate a \TwoDTLTI code to multiple copies of the TC. This makes it possible to leverage efficient TC decoders such as MWPM and UF.
The decoupling transformation can map one physical error to several nearby errors in several TC sectors. These correlations can affect the performance of the decoder.

In this section, we present a cell-matching decoding framework that works directly with the intrinsic unit cells of the original code. The decoder first uses fixed local flushing operators to replace the syndrome in each cell by its chosen charge representatives. After this step, there are no more excitations outside of the preassigned basis subcells in the unit cells of the lattice.
Each coefficient of the resulting charge vector specifies an excitation on an auxiliary toric graph. Thus, we can interpret the resulting residual classes as check violation patterns on several TC-like coarse lattices.
These point-like excitations in the basis subcells exist as pairs in the lattice which makes it possible for us to decode all TC copies by matching all of these point-like excitations. The edge generators used in the matchings on the coarse lattice give us the physical Pauli correction for the original code. 

For concreteness, we focus our exposition on decoding $Z$ errors. The bit-flip error correction process is analogous.

\subsection{Unit cell construction}
\label{subsec:unit_cell_mapping}
Use the coarse-graining hypotheses of Proposition~\ref{prop:suitable-coarse-grained-parity-check-matrix}: the charge space is finite dimensional and, for a fixed positive integer $b$,
\begin{equation}(x^b-1)\coker H=(y^b-1)\coker H=0.\end{equation}
Consequently, each charge has a finite-support Pauli transport operator to the corresponding position in an adjacent $b\times b$ unit cell. Fix a finite list of such operators and translate it throughout the lattice. Their weights and support radii are constants. Note that their supports may extend into a bounded neighborhood of the unit cells at their endpoints.

We assume periodic boundary conditions on a torus of linear dimensions $L_x\times L_y$, and we write the physical lattice as
\begin{equation}\Lambda \coloneqq \Z_{L_x}\times \Z_{L_y}.
\end{equation}
From now on, we assume $b\mid L_x$ and $b\mid L_y$, so that the lattice decomposes into disjoint $b\times b$ unit cells. 
The associated coarse lattice is the torus
\begin{equation}\Lambda_b \coloneqq \Z_{L_x/b}\times \Z_{L_y/b},
\end{equation}
whose vertices index unit cells and whose edges connect nearest-neighbor unit cells.

In addition, we can identify a sparse basis for the cokernel of the coarse-grained parity-check matrix $H$ over the base ring $R'$. The dimension of the basis is given by $\dim(\coker\,H)$. By definition of $b$, the basis elements can be chosen to be point-like excitations that are contained in some unit subcell of size $c \times c$. This means that any excitation configuration in the $b \times b$ unit cell can be expressed as a linear combination of these basis excitations. Without loss of generality, we can assume that these basis excitations are located in a single subcell of size $c \times c$ that is positioned in the top-left corner of each $b \times b$ unit cell. From here on, we will refer to this $c \times c$ subcell as the basis subcell. Note that the basis subcell need not be a square and can be any shape that contains the basis excitations, but for simplicity we will refer to it as a $c \times c$ subcell.

In what follows, $c$ denotes the side length of the basis subcell within each unit cell, and we write
\begin{equation}r\coloneqq \dim_{\F_2}(\coker H)\end{equation}
for the number of independent excitation types (basis elements of $\coker H$) represented inside the basis subcell.
Equivalently, the code encodes $k=2r$ logical qubits. There is no requirement that $r=c^2$. Each point-like excitation in the basis subcell need not correspond to a single excitation in a single TC copy. It could be a single violated check in several TC copies. Nonetheless, as we will show in the subsequent sections, as long as we resolve all point-like excitations in all basis subcells, we would have returned the original code to its code space.

\subsection{Flushing of excitations}
\label{subsec:flushing_of_excitations}
We now formalize the local flushing step that converts the raw syndrome inside each $b\times b$ unit cell into a canonical representative supported on the chosen representatives inside its $c\times c$ basis subcell.

The flushing choices are fixed on basis syndrome vectors and then extended linearly. Note that locality in the subsequent discussion means support in a constant-radius neighborhood of a cell. It does not require every transport operator to fit strictly inside that cell.

\begin{definition}[Flushing map]
\label{def:flushing-map}
Let $u\in V(\Lambda_b)$ be a unit cell. Define $\mathsf{Coeff}(s_u)\in\F_2^r$ to be the coordinates of the charge class of $s_u$ in the chosen basis, and let $\mathsf{Flush}(s_u)$ be its chosen representative in the basis subcell of $u$. For each basis syndrome vector $v$ in one unit cell, choose a finite-support Pauli $\Psi(v)$ with
\begin{equation}\partial_1\Psi(v)=v+\mathsf{Flush}(v).\end{equation}
Such a preimage exists by equality of charge classes. Translate these finitely many choices to all cells and extend linearly over $\F_2$. We identify Pauli products with addition in $C_1$ throughout the proofs. The resulting maps $\Psi$ and $\mathsf{Flush}$ have bounded spread.

Let $S^{(i)}:C_0\to C_0^{(i)}$ collect the $i$th coefficient of $\mathsf{Coeff}(s_u)$ over all cells, and let $\iota^{(i)}:C_0^{(i)}\to C_0$ insert its chosen physical syndrome representatives. Then
\begin{equation}\label{eq:cellular-flushing-identity}
\partial_1\Psi(s)=s+\sum_{i=1}^r\iota^{(i)}S^{(i)}s.
\end{equation}
All these maps are linear. The physical flushing operator is the product of the independently selected $\Psi(s_u)$, even when their bounded neighborhoods overlap.
\end{definition}

Intuitively, $\mathsf{Flush}$ annihilates any locally physical syndrome inside $u$ (i.e., any element of $\im H$ supported in $u$), and returns the residual non-physical class (an element of $\coker H$) as an $r$-bit vector.

We now state two lemmas that describe the properties of the flushing procedure.

\begin{lemma}
\label{lem:flushing_procedure_physical}
  If the excitations within the $b \times b$ unit cell lie in $\im H$ (i.e., it is a physical syndrome that can be created by some local Pauli operator), then the flushing procedure reduces it to an empty set of excitations in the basis subcell.
\end{lemma}
\begin{proof}
Let $s_u\in\im H$ be a syndrome supported inside a unit cell $u$.
Then $s_u$ represents the zero class in $\coker\,H$.
By construction, $\mathsf{Flush}(s_u)$ is a representative of this same class supported in the basis subcell, hence must be the empty configuration.
Equivalently, $\mathsf{Coeff}(s_u)=0\in\F_2^{r}$.
\end{proof}

The above lemma implies that performing the flushing procedure across all the unit cells will not introduce any additional excitations in the basis subcells if the original syndrome in each unit cell is physical. This allows us a local way to return to the codespace. 

However, if the original syndrome in a unit cell is not physical (i.e., it does not lie in $\im H$), then the flushing procedure may result in a non-trivial excitation configuration in the basis subcell. This can happen when the physical syndrome resulting from some Pauli error is split across multiple unit cells. In other words, the restriction of the syndrome to a particular unit cell may not correspond to any local Pauli operator within that unit cell. We formalize this observation in the following lemma.

\begin{lemma}
\label{lem:flushing_procedure_nonphysical}
  If the excitations within the $b \times b$ unit cell do not lie in $\im H$ (i.e., it is a non-physical syndrome that cannot be created by any local Pauli operator), then the flushing procedure will map it to a non-trivial excitation configuration in the basis subcell that corresponds to a non-trivial element in $\coker\,H$.
\end{lemma}
\begin{proof}
If $s_u\notin\im H$, then it represents a non-zero class in $\coker H$.
Flushing produces a representative of this class supported in the basis subcell, so the result cannot be empty.
Equivalently, $\mathsf{Coeff}(s_u)\neq 0\in\F_2^{r}$.
\end{proof}

The global parity constraint also holds for errors crossing the boundary between adjacent unit cells. To see this, recall that $\mathcal Q=\operatorname{coker}H
=R_\infty^t/\operatorname{im}H$
is the infinite-lattice charge space, where $t$ is the number of
check types. Two finite-support syndrome patterns represent the
same charge precisely when their difference is the syndrome of
a finite-support physical error. In particular, zero charge does
not mean an empty syndrome.

The identities $(x^b-1)\mathcal Q=(y^b-1)\mathcal Q=0$ therefore say that translation by one coarse cell preserves every charge class.
Since $b\mid L_x,L_y$, translation by either torus period also preserves charge. Equivalently, $I_L\mathcal Q=0$ where $I_L=(x^{L_x}-1,y^{L_y}-1)$.

For a physical syndrome $s=H_LE$, $s$ clearly belongs to the trivial excitation-type since its syndrome belongs to $\operatorname{im}H$. After flushing, the coefficient of $\epsilon^{(i)}$ in the total charge is the parity of the number of type-$i$ excitations.
Since the total charge vanishes and the excitation-type basis is linearly
independent, each of these parities for each excitation-type must be zero. Thus every excitation set $D_i$ has even cardinality, including when the
physical error crosses a boundary between adjacent unit cells.

\subsection{Mapping to TC copies}
\label{subsec:mapping_to_ktc}
Fix a basis of $\coker H$ represented by point-like excitations supported in the basis subcell, and denote the corresponding excitation types by
\begin{equation}\epsilon^{(1)},\dots,\epsilon^{(r)}.
\end{equation}

After flushing, each unit cell $u\in V(\Lambda_b)$ yields a coefficient vector
\begin{equation}a(u)\in\F_2^{r}\end{equation}
obtained by expressing the flushed basis-subcell configuration in the fixed basis of $\coker H$.
We interpret the $i$th component $a_i(u)\in\F_2$ as the presence/absence of an excitation of \emph{excitation type} $\epsilon^{(i)}$ in unit cell $u$.
Therefore, for each $i\in[r]$ we obtain an excitation set
\begin{equation}D_i\coloneqq \{u\in V(\Lambda_b): a_i(u)=1\}\subseteq V(\Lambda_b),
\end{equation}
which we decode on the coarse torus $\Lambda_b$ using a TC matching routine. Again, we emphasize that this coarse torus need not actually correspond to a single TC copy. It could be a combination of TC copies. To avoid any potential confusion, we shall refer to the matching graphs for each coarse torus as coarse toric graphs.
Note that $\Psi$ from Definition~\ref{def:flushing-map} composes with the error $E$ on the code to produce these excitation sets. Thus, one can denote the edge-error pattern on the coarse lattice $\Lambda_b$ for excitation type $\epsilon^{(i)}$ as $\gamma^{(i)}$.

To lift a coarse correction back to a physical Pauli correction, we assume access to local edge generators that create pairs of a given excitation type across a unit-cell boundary.

\begin{definition}[Edge generators]
\label{def:edge-generators}
For each excitation type $\epsilon^{(i)}$ and each coarse edge $e=(u,v)$, an \emph{edge generator} is a $Z$-type Pauli operator $P_e^{(i)}$ supported in a bounded neighborhood of $u$ and $v$, with the exact physical syndrome
\begin{equation}\label{eq:cellular-edge-boundary}
\partial_1P_e^{(i)}=\iota^{(i)}(u+v).
\end{equation}
Let $s_i$ be the chosen infinite-lattice representative of $\epsilon^{(i)}$. The annihilator condition gives $(1+x^b)s_i,(1+y^b)s_i\in\im H$. Choose a finite size physical Pauli operator for each of these $2r$ syndromes and translate these choices to form all possible coarse edges. The weight of these Pauli edge generators must be constant, and translational invariance preserves the exact syndrome equations. If an initial Pauli operator $Q_e^{(i)}$ transports the correct charge type but induces some unwanted syndrome pattern in the process, replace it by $Q_e^{(i)}+\Psi(\partial_1Q_e^{(i)})$. \eqref{eq:cellular-flushing-identity} shows how the unwanted syndrome patterns can be canceled, leaving us with the chosen type-$i$ excitation type at each endpoint. This process enlarges the support of the Pauli operator by only a bounded amount.

We assume there is a uniform constant $w_{\mathrm{edge}}\in\N$ that is independent of $L_x,L_y$ such that $\mathrm{wt}(P_e^{(i)})\le w_{\mathrm{edge}}$ for all $i\in[r]$ and all coarse edges $e$.
\end{definition}

We can now state the cell-matching decoding procedure as an explicit algorithm.

\begin{algorithm}[H]
\caption{Cell-matching decoder for $Z$-type errors}
\label{alg:cellular-decoder}
\begin{algorithmic}[1]
  \Require{$X$-check syndrome $s$ of a $Z$ error, converted to the selected check presentation if necessary; coarse-graining parameter $b$ and basis subcell; flushing map $\mathsf{Flush}$ and coefficient extraction $\mathsf{Coeff}$ (Definition~\ref{def:flushing-map}); edge generators $\{P_e^{(i)}\}$ (Definition~\ref{def:edge-generators}); unweighted MWPM.}
  \Ensure{$Z$-type Pauli correction $\widehat{E}$.}
  \ForAll{unit cells $u\in V(\Lambda_b)$}
    \State Restrict the syndrome to $u$ to obtain $s_u$.
    \State Flush with $\Psi(s_u)$ and extract coefficients: $a(u)\coloneqq \mathsf{Coeff}(s_u)\in\F_2^{r}$.
  \EndFor
  \For{$i\gets 1$ \textbf{to} $r$}
    \State Form excitations $D_i\coloneqq\{u\in V(\Lambda_b): a_i(u)=1\}$.
    \State Compute a minimum-weight 1-chain $\widehat\gamma^{(i)}\in \F_2^{E(\Lambda_b)}$ with boundary $\partial_1^{(i)}\widehat\gamma^{(i)}=\mathbf{1}_{D_i}$ using MWPM with unit weight per coarse edge.
  \EndFor
  \State Lift and combine: $\widehat{M}\coloneqq \prod_{i=1}^{r}\ \prod_{e\in E(\Lambda_b)} (P_e^{(i)})^{\widehat\gamma^{(i)}_e}$.
  \State \Return $\widehat{E} \coloneqq \widehat{M} \cdot \prod_{u\in V(\Lambda_b)} \Psi(s_u)$.
\end{algorithmic}
\end{algorithm}

The toric graphs used by Algorithm~\ref{alg:cellular-decoder} are extracted from the unit-cell structure and need not coincide with globally disentangled TC copies.
Nevertheless, the only Pauli operators the decoder applies are products of the edge generators $P_e^{(i)}$ and the flushing operators $\Psi(s_u)$.
The next lemma formalizes the fact that any decoder-induced logical operator must come from a non-trivial cycle on at least one of the coarse toric graphs.

\begin{definition}[Coarse toric graph complex]
\label{def:coarse-ktc-graph-complex}
For each $i\in[r]$, let
\begin{equation}\calG^{(i)}:\quad C^{(i)}_2 \xrightarrow{\partial^{(i)}_2} C^{(i)}_1 \xrightarrow{\partial^{(i)}_1} C^{(i)}_0\end{equation}
be the graph chain complex of the coarse torus $\Lambda_b$ over $\F_2$, where $C^{(i)}_2\cong \F_2^{F(\Lambda_b)}$, $C^{(i)}_1\cong \F_2^{E(\Lambda_b)}$, $C^{(i)}_0\cong \F_2^{V(\Lambda_b)}$, and $\partial^{(i)}_j$ is the incidence map for $j \in \{1, 2\}$.
\end{definition}

\begin{definition}[Lifting of coarse 1-chains]
\label{def:coarse-chain-lifting}
Fix $i\in[r]$.
For any 1-chain $\gamma\in C^{(i)}_1\cong\F_2^{E(\Lambda_b)}$, define the associated $Z$-type Pauli operator
\begin{equation}\mathsf{P}^{(i)}(\gamma) \coloneqq \prod_{e\in E(\Lambda_b)} (P^{(i)}_e)^{\gamma_e}.
\end{equation}
\end{definition}

\begin{definition}[Contractible-cycle stabilizer property]
\label{def:contractible-cycle-stabilizer-property}
We say the cellular construction satisfies the \emph{contractible-cycle stabilizer property} if, for every $i\in[r]$ and every homologically trivial 1-cycle $\gamma\in\im(\partial^{(i)}_2)$ on the coarse torus $\Lambda_b$, the lifted operator $\mathsf{P}^{(i)}(\gamma)$ lies in $\im(\partial_2)$. Here the property concerns arbitrary sums of plaquette boundaries, not only simple contractible loops.
\end{definition}

\begin{lemma}[Homologically-trivial coarse cycles lift to $Z$-stabilizers]
\label{lem:contractible-cycle-stabilizer-property}
Let $d_Z(\calC)$ denote the $Z$-distance of $\calC$, i.e., the minimum weight of a $Z$-type Pauli operator in $\ker(\partial_1)\setminus \im(\partial_2)$.
Assume that each edge generator $P_e^{(i)}$ has support size at most $w_{\mathrm{edge}}$ (independent of $L_x,L_y$).
Then for every $i\in[r]$ and every homologically trivial 1-cycle $\gamma\in\im(\partial^{(i)}_2)$ on $\Lambda_b$, we have
\begin{equation}\mathsf{P}^{(i)}(\gamma)\in\im(\partial_2)\end{equation}
provided $d_Z(\calC) > 4w_{\mathrm{edge}}$.
In particular, for fixed $b$ and fixed local choices of the generators, this holds for all sufficiently large system sizes whenever $d_Z(\calC)$ grows with $\min\{L_x,L_y\}$.
\end{lemma}

\begin{proof}
Fix $i\in[r]$.
Let $\gamma\in\im(\partial^{(i)}_2)$ be a homologically-trivial coarse 1-cycle.
View $\Lambda_b$ as a square cell complex with face set $F(\Lambda_b)$.
By definition, there exists a 2-chain $\sigma\in \F_2^{F(\Lambda_b)}$ such that
\begin{equation}\gamma = \partial_2^{(i)}\sigma = \sum_{p\in F(\Lambda_b)} \sigma_p\,\partial p,\end{equation}
where $\partial p\in C_1^{(i)}$ denotes the 4 edges that form the boundary cycle of a coarse plaquette $p \in F(\Lambda_b)$.

By Definition~\ref{def:coarse-chain-lifting},
\begin{equation}\mathsf{P}^{(i)}(\gamma) = \prod_{p\in F(\Lambda_b)} \bigl(\mathsf{P}^{(i)}(\partial p)\bigr)^{\sigma_p}.
\end{equation}
Thus it suffices to show that $\mathsf{P}^{(i)}(\partial p)\in\im(\partial_2)$ for every coarse plaquette $p$.

Fix a coarse plaquette $p$.
The operator $\mathsf{P}^{(i)}(\partial p)$ is a product of four edge generators, hence has weight at most $4w_{\mathrm{edge}}$.
Moreover, $\partial p\in\ker(\partial_1^{(i)})$ implies that each coarse vertex of $p$ is incident to an even number of selected edges, so the type-$\epsilon^{(i)}$ excitations cancel.
Equation~\eqref{eq:cellular-edge-boundary} makes this cancellation exact for the physical syndrome, so $\mathsf{P}^{(i)}(\partial p)\in\ker(\partial_1)$.

If $\mathsf{P}^{(i)}(\partial p)\notin\im(\partial_2)$, then it would represent a non-trivial physical $Z$ logical operator and hence have weight at least $d_Z(\calC)$.
But its weight is at most $4w_{\mathrm{edge}}$, which contradicts the assumption $d_Z(\calC)>4w_{\mathrm{edge}}$.
Therefore $\mathsf{P}^{(i)}(\partial p)\in\im(\partial_2)$.

Taking the product over all plaquettes with $\sigma_p=1$ shows $\mathsf{P}^{(i)}(\gamma)\in\im(\partial_2)$.
\end{proof}

\begin{lemma}[Decoder-induced logical operators of 2D TTI code comes from some coarse toric sector]
  \label{lem:logical_relationship}
Assume $b\mid L_x,L_y$, and let
\begin{equation}\calC:\quad C_2 \xrightarrow{\partial_2} C_1 \xrightarrow{\partial_1} C_0\end{equation}
be the chain complex for the \TwoDTLTI code $Q$ so that the $Z$ logical operators are represented by
\begin{equation}H_1(\calC)=\ker(\partial_1)/\im(\partial_2).\end{equation}

Fix edge generators $\{P_e^{(i)}\}$ as in Definition~\ref{def:edge-generators} and coarse graph complexes $\calG^{(i)}$ as in Definition~\ref{def:coarse-ktc-graph-complex}.
Then the contractible-cycle stabilizer property (Definition~\ref{def:contractible-cycle-stabilizer-property}) holds whenever $d_Z(\calC)>4w_{\mathrm{edge}}$, by Lemma~\ref{lem:contractible-cycle-stabilizer-property}.

Then for each $i\in[r]$ the lifting map (Definition~\ref{def:coarse-chain-lifting}) induces a well-defined homomorphism on homology
\begin{equation}j^{(i)}:H_1(\calG^{(i)})\to H_1(\calC),\qquad j^{(i)}([\gamma])\coloneqq [\mathsf{P}^{(i)}(\gamma)].\end{equation}

In addition, assume the local reconstruction hypotheses of Lemma~\ref{lem:cellular-error-reconstruction}, including $d_Z(\calC)>w_{\mathrm{loc}}$. For a physical error $E$, set $\gamma^{(i)}=J^{(i)}E$ using Lemma~\ref{lem:constant-spread-cellular}. If Algorithm~\ref{alg:cellular-decoder} returns corrections for which every residual $\gamma^{(i)}+\widehat\gamma^{(i)}$ is homologically trivial, then the physical residual is a $Z$ stabilizer. Thus physical logical failure requires a nontrivial residual homology class in at least one coarse sector.
\end{lemma}

\begin{proof}
Fix $i\in[r]$.
Let $\gamma\in\ker(\partial^{(i)}_1)$.
At each coarse vertex, an even number of incident edges satisfy $\gamma_e=1$.
By Eq.~\eqref{eq:cellular-edge-boundary}, the total product $\mathsf{P}^{(i)}(\gamma)$ has trivial physical syndrome.
Hence $\mathsf{P}^{(i)}(\gamma)\in\ker(\partial_1)$.

If $\gamma$ and $\gamma'$ represent the same class in $H_1(\calG^{(i)})$, then $\gamma+\gamma'$ is a sum of contractible cycles.
By Lemma~\ref{lem:contractible-cycle-stabilizer-property}, $\mathsf{P}^{(i)}(\gamma+\gamma')\in\im(\partial_2)$.
Since $\mathsf{P}^{(i)}(\gamma+\gamma')=\mathsf{P}^{(i)}(\gamma)\,\mathsf{P}^{(i)}(\gamma')$, we conclude $[\mathsf{P}^{(i)}(\gamma)]=[\mathsf{P}^{(i)}(\gamma')]$ in $H_1(\calC)$.
Therefore $j^{(i)}$ is well-defined.

For the decoder statement, write its correction additively as
\begin{equation}\widehat E=\Psi(\partial_1E)+\sum_i\mathsf{P}^{(i)}(\widehat\gamma^{(i)}).
\end{equation}
Lemma~\ref{lem:cellular-error-reconstruction} implies that $E+\widehat E$ differs by a stabilizer from $\sum_i\mathsf{P}^{(i)}(J^{(i)}E+\widehat\gamma^{(i)})$. Each summand is a stabilizer when its coarse residual is homologically trivial. Taking the contrapositive yields the final statement.
\end{proof}

\subsection{Correctness and complexity analysis of the cell-matching decoder}
\label{subsec:cellular-decoder-analysis}

In this subsection, we analyze the correctness and complexity of the cell-matching decoding framework. We show that the cell-matching decoder correctly recovers the logical information when the error rate is below a certain threshold, and we analyze its computational complexity.

We analyze Algorithm~\ref{alg:cellular-decoder}: for each unit cell we flush the restricted syndrome into the basis subcell (Definition~\ref{def:flushing-map}), interpret the resulting coefficients as excitation patterns $D_i\subseteq V(\Lambda_b)$ for each excitation type $\epsilon^{(i)}$, decode each excitation pattern independently by matching on $\Lambda_b$, and lift the coarse corrections using the edge generators $P_e^{(i)}$.

We refine the correctness argument into several lemmas corresponding to the two principal ways a logical error may be introduced during decoding:
(i)~the \emph{flushing stage} may choose a correction representative that, when combined with the pre-existing physical error, has non-trivial homology, and
(ii)~the \emph{matching stage} on one (or more) of the induced coarse toric instances may return a correction in the wrong homology class.

We now isolate the only geometric input needed to relate physical error weight to the induced coarse edge-error patterns.

\begin{lemma}[Canonical coarse error maps and bounded spread]
\label{lem:constant-spread-cellular}
Fix the cell construction and the linear maps in Definition~\ref{def:flushing-map}. There are linear maps
\begin{equation}
J^{(i)}:C_1\longrightarrow C_1^{(i)},\qquad
\partial_1^{(i)}J^{(i)}=S^{(i)}\partial_1,
\label{eq:cellular-chain-commutation}
\end{equation}
with bounded support radius and constants $k_{\mathrm{cell}},m_{\mathrm{cell}}\geq1$, independent of $L_x,L_y$, such that
\begin{align}
|J^{(i)}E|&\leq k_{\mathrm{cell}}|E|,\label{eq:cellular-spread}\\
\#\{a:e\in\operatorname{supp}J^{(i)}a\}&\leq m_{\mathrm{cell}}
\end{align}
for every physical error $E$, every coarse edge $e$, and every $i$. Here $a$ ranges over basis single-qubit errors.
\end{lemma}

\begin{proof}
We first define $J^{(i)}$ for each single-qubit $Z$ error $a$
in an arbitrary unit cell. Its syndrome occupies a bounded neighborhood of that cell. After flushing, the syndrome is
represented by the chosen charge patterns in the basis subcells of nearby unit cells. The coefficient pattern
$S^{(i)}\partial_1a$ records which of these unit cells contain
a type-$i$ excitation. Since a physical error has zero total charge, there are an even number of such excitations.

Fix a pairing of these type-$i$  excitations and choose paths on the coarse lattice joining their unit cells within a
bounded neighborhood. Each coarse edge along a path corresponds to an edge generator $P_e^{(i)}$, whose physical syndrome is the chosen type-$i$ representative in the basis subcells of
its two endpoint unit cells. Along a path, these representatives cancel at intermediate unit cells, leaving only the endpoint patterns.

Define $J^{(i)}a$ to be the sum modulo two of the chosen coarse paths. Its boundary is precisely the type-$i$ coefficient pattern $\partial_1^{(i)}J^{(i)}a=S^{(i)}\partial_1a$.
Fix these choices once for the finitely many physical qubit types and charge types in a unit cell. Translate
them to all unit cells while accounting for the periodic boundary conditions and extend linearly $J^{(i)}E
=\sum_{a\in\operatorname{supp}E}J^{(i)}a$.
The boundary identity on single-qubit errors therefore gives $\partial_1^{(i)}J^{(i)}E=S^{(i)}\partial_1E$
for every physical error $E$. This holds even after accounting for the periodic boundary condition even across boundaries between adjacent unit cells. The map is
thus fixed on single-qubit errors, rather than defined by a new matching calculation for each syndrome.

There are finitely many chosen path templates, each containing
a bounded number of coarse edges and lying within a bounded
distance of the unit cell containing $a$. Let
$k_{\mathrm{cell}}\geq1$ bound the number of coarse edges in
any template. Then
\[
|J^{(i)}E|
\leq\sum_{a\in\operatorname{supp}E}|J^{(i)}a|
\leq k_{\mathrm{cell}}|E|,
\]
since cancellations modulo two can only reduce the support.

Conversely, a fixed coarse edge can belong to $J^{(i)}a$
only when $a$ lies in one of a bounded number of nearby
unit cells. Each unit cell contains a fixed number of
physical qubits, so there is a system-size-independent
constant $m_{\mathrm{cell}}$ satisfying
\[
\#\{a:e\in\operatorname{supp}J^{(i)}a\}
\leq m_{\mathrm{cell}}.
\]
This proves the stated locality and spread bounds.
\end{proof}

\begin{lemma}[Local reconstruction of physical errors]
\label{lem:cellular-error-reconstruction}
Use the maps of Lemma~\ref{lem:constant-spread-cellular} and the normalized edge generators of Definition~\ref{def:edge-generators}. There is a constant $w_{\mathrm{loc}}$, independent of system size, such that, whenever $d_Z(\calC)>w_{\mathrm{loc}}$, every $E\in C_1$ satisfies
\begin{equation}\label{eq:cellular-error-reconstruction}
E+\Psi(\partial_1E)+\sum_{i=1}^r\mathsf{P}^{(i)}(J^{(i)}E)
\in\im\partial_2.
\end{equation}
Here $\mathsf{P}^{(i)}$ is regarded as a linear map into the binary space of $Z$ Paulis.
\end{lemma}

\begin{proof}
For a basis error $a$, denote the left-hand side of~\eqref{eq:cellular-error-reconstruction} by $B_a$. \eqref{eq:cellular-flushing-identity}, \eqref{eq:cellular-edge-boundary}, and \eqref{eq:cellular-chain-commutation} give
\begin{align}
\partial_1B_a
&=\partial_1a+\partial_1a+\sum_i\iota^{(i)}S^{(i)}\partial_1a\nonumber\\
&\quad+\sum_i\iota^{(i)}\partial_1^{(i)}J^{(i)}a\nonumber\\
&=0,
\end{align}
implying that $B_a$ is a zero-syndrome operator, i.e., either a logical operator or a stabilizer.
Let $w_H$ bound the syndrome weight of a basis error and let $w_\Psi$ bound the weight of the flushing image of a basis syndrome vector. These bounds are independent of the lattice size. Set
\begin{equation}
w_{\mathrm{loc}}\coloneqq 1+w_Hw_\Psi+r w_{\mathrm{edge}}k_{\mathrm{cell}}.
\end{equation}
Subadditivity gives $|B_a|\leq w_{\mathrm{loc}}$. A zero-syndrome Pauli of weight below $d_Z(\calC)$ is a stabilizer. Hence every $B_a$ is a stabilizer, and linearity gives the claim for arbitrary $E$.
\end{proof}

In particular, for the actual physical error $E$ the appropriate coarse error is the fixed linear representative $\gamma^{(i)}=J^{(i)}E$. It need not minimize weight. The decoder only needs $S^{(i)}\partial_1E$ and computes its own minimum-weight recovery. The reconstruction lemma accounts for the local physical residual that the coarse syndrome alone does not determine.

\begin{lemma}[Matching-decoder correctness on each coarse toric graph instance]
\label{lem:cellular-ktc-matching-correctness}
Fix $i\in[r]$. Let $\Lambda_b$ have size $(L_x/b)\times(L_y/b)$, and let
\begin{equation}d_{TC} \coloneqq \min\{L_x/b,\,L_y/b\}
\end{equation}
be the distance of the corresponding coarse toric decoding graph.
Let $\gamma\in C^{(i)}_1$ be an adversarial edge-error pattern with $|\gamma|\le \left\lfloor \frac{d_{TC}-1}{2}\right\rfloor$, and let $\hat\gamma$ be any minimum-weight 1-chain satisfying $\partial^{(i)}_1\hat\gamma = \partial^{(i)}_1\gamma$, as produced by MWPM with unit weight per coarse edge.
Then $\gamma+\hat\gamma$ is homologically trivial on the torus, i.e., it differs from $0$ by a sum of contractible cycles.
\end{lemma}

\begin{proof}
The 1-chain $\gamma+\hat\gamma$ has zero boundary, and thus decomposes into a disjoint union of cycles on $\Lambda_b$. If any component were non-contractible, its length would be at least $d_{TC}$. On the other hand, by minimality of $\hat\gamma$ we have $|\hat\gamma|\le |\gamma|$, so
\begin{equation}|\gamma+\hat\gamma| \le |\gamma|+|\hat\gamma| \le 2\left\lfloor \frac{d_{TC}-1}{2}\right\rfloor < d_{TC},\end{equation}
which rules out any non-contractible component. Hence all cycles are contractible.
\end{proof}

\begin{theorem}[Cell-matching decoder for \TwoDTLTI codes]
\label{thm:cellular-decoder-finite-size}
Assume the fixed local cell construction above, $b\mid L_x,L_y$, $L_x,L_y\ge3b$, and
\begin{equation}
d_Z(\calC)>\max\{4w_{\mathrm{edge}},w_{\mathrm{loc}}\}.
\end{equation}
Let $d_{TC}=\min\{L_x/b,L_y/b\}$. With unweighted MWPM, Algorithm~\ref{alg:cellular-decoder} corrects every $Z$ error of weight at most
\begin{equation}
t_{\max}=\left\lfloor\frac{d_{TC}-1}{2k_{\mathrm{cell}}}\right\rfloor.
\end{equation}

Suppose $p\in[0,1]$ and $\Pr[A\subseteq\operatorname{supp}E]\leq p^{|A|}$ for every set $A$ of physical qubits. Then the induced support in each coarse graph is local-stochastic with parameter
\begin{equation}q_{\mathrm{cell}}=m_{\mathrm{cell}}p^{1/k_{\mathrm{cell}}}.
\end{equation}
For $6\sqrt{q_{\mathrm{cell}}}<1$, the failure probability obeys
\begin{equation}\label{eq:cellular-peierls-bound}
\Pr[\mathrm{failure}]\leq c_0r|V(\Lambda_b)|
\sum_{\ell\geq d_{TC}}(6\sqrt{q_{\mathrm{cell}}})^\ell,
\end{equation}
where $c_0$ is an absolute constant. This is exponentially small in $d_{TC}$ when $\log|V(\Lambda_b)|=o(d_{TC})$, in particular for bounded aspect ratio. The analogous conclusions hold for the $X$ sector with its corresponding constants and distance.
\end{theorem}

\begin{proof}
For the number of correctable errors, Lemma~\ref{lem:constant-spread-cellular} gives $|J^{(i)}E|\leq k_{\mathrm{cell}}t_{\max}<d_{TC}/2$. Lemma~\ref{lem:cellular-ktc-matching-correctness} makes every residual $J^{(i)}E+\widehat\gamma^{(i)}$ homologically trivial. Each lifted residual is a stabilizer by Lemma~\ref{lem:contractible-cycle-stabilizer-property}. Equation~\eqref{eq:cellular-error-reconstruction} then shows that the physical residual after flushing and lifting is a stabilizer.

For the stochastic argument, let $R$ be a set of coarse edges present in $J^{(i)}E$. Each edge in $R$ must arise from at least one physical error from the 2D TTI code even after mod-two cancellations. Assign such a physical error to every edge. There are at most $m_{\mathrm{cell}}^{|R|}$ assignments. Each consistent assignment needs at least $|R|/k_{\mathrm{cell}}$ distinct physical errors, whose joint occurrence has probability at most $p^{|R|/k_{\mathrm{cell}}}$. A union bound gives
\begin{equation}
\Pr[R\subseteq\operatorname{supp}J^{(i)}E]
\leq(m_{\mathrm{cell}}p^{1/k_{\mathrm{cell}}})^{|R|}.
\end{equation}
This establishes local stochasticity without assuming independence of the coarse edges.

If MWPM fails, some simple nontrivial cycle $\gamma$ in the residual has length $\ell\geq d_{TC}$. Minimality of the correction operator implies that at least $\ell/2$ edges of $\gamma$ are in $J^{(i)}E$. There are at most $2^\ell$ candidate subsets of these edges, and local stochasticity bounds their total probability by $2^\ell q_{\mathrm{cell}}^{\ell/2}$. On the degree-four coarse graph, the number of length-$\ell$ simple cycles is at most $c_0|V(\Lambda_b)|3^\ell$. Summing over lengths and the $r$ sectors proves~\eqref{eq:cellular-peierls-bound}.
\end{proof}

For full Pauli noise, the same argument presented for the \emph{layer-decoupling} decoder holds for the \emph{cell-matching} decoder.

To compare with the physical distance, choose a common coarse-graining scale $b$ for the cell-matching construction and the layer-decoupling approach of Section~\ref{sec:decoding}, with at least one toric summand and $b\mid L_x,L_y$.
The same argument as in~\eqref{eq:layer-physical-distance-comparison} gives $d_L\leq w_{\mathrm{inv}}d_{TC}$, where $w_{\mathrm{inv}}$ bounds the weight of the physical errors that lie in the inverse image of a TC edge.
Under the hypotheses of Theorem~\ref{thm:cellular-decoder-finite-size}, the condition $2k_{\mathrm{cell}}w_{\mathrm{inv}}|E|<d_L$ therefore implies $2k_{\mathrm{cell}}|E|<d_{TC}$ and guarantees successful decoding.
For full Pauli errors, use the common constants and the distance conditions for both components specified above.

The guaranteed number of correctable physical errors is
$\lfloor(d_{TC}-1)/(2k_{\mathrm{cell}})\rfloor$.
Both the coarse toric distance and the error-spread constant
depend on the chosen coarse-graining parameter.
This worst-case guarantee does not determine the decoder's
performance under random noise because errors beyond the guaranteed
radius may still be corrected successfully. In addition, excitation-types with shorter edge generators can permit matching that utilizes less Pauli operators. Section~\ref{sec:gross} investigates this possibility
for finite BB codes.

We now proceed to state the time complexity of the cell-matching decoder for finite-size \TwoDTLTI codes. Note that $r$ and $b$ are constants.

\begin{theorem}[Time complexity of the cell-matching decoder]
\label{thm:time-complexity-cellular-decoder-finite-size}
Under the hypotheses of Theorem~\ref{thm:cellular-decoder-finite-size}, suppose the fixed flushing maps and edge generators have been precomputed. Let $n_L$ be the number of physical qubits and $N_i$ the total number of vertices and edges of coarse graph $i$. For a fixed code family, $n_L=\Theta(L_xL_y)$, $r=O(1)$, and $N_i=O(n_L)$. The sequential decoding time for one Pauli component is
\begin{equation}
O\!\left(n_L+\sum_{i=1}^r\operatorname{MATCHING}(N_i)\right),
\end{equation}
where $\operatorname{MATCHING}(N_i)$ includes the shortest-path computations, MWPM, and reconstruction of the recovery chain. Decoding both Pauli components uses at most $2r$ calls and has the same asymptotic bound.
\end{theorem}

\begin{proof}
The time complexity of the cell-matching decoder can be analyzed by considering the different steps involved in the decoding process:
\begin{enumerate}
  \item \textbf{Flushing within each unit cell:} Flushing uses only transport operators supported within a constant-radius neighborhood of a $b\times b$ unit cell. Hence it takes $O(1)$ time per unit cell. There are $\frac{L_xL_y}{b^2}$ unit cells, so the total flushing time is $O\left(L_xL_y\right)$. The local calculations have constant parallel depth with sufficiently many processors, while their total work remains $O(n_L)$.

  \item \textbf{Excitation calculation for each excitation type:} For each $i\in[r]$, calculating the induced syndrome/excitation pattern on $\Lambda_b$ from the flushed basis-subcell data takes $O(1)$ per unit cell and hence $O\left(L_xL_y\right)$ overall. Again, the local calculations can be parallelized without changing the total work.

  \item \textbf{Decoding each cell-matching TC instance:} For each $i\in[r]$, the induced decoding graph has $O\left(L_xL_y\right)$ vertices/edges, and the specified MWPM runtime is $\operatorname{MATCHING}(N_i)$, including shortest paths and recovery-chain reconstruction. Summed over $r$ instances, this step costs $\sum_i\operatorname{MATCHING}(N_i)$.
  \item \textbf{Lifting and applying corrections:} After summing paths modulo two, each selected elementary coarse edge corresponds to applying a fixed local operator $P^{(i)}_e$ supported near a unit-cell boundary. Thus lifting/applying corrections costs $O(1)$ per selected elementary coarse edge, and hence at most $O\left(L_xL_y\right)$ per instance in the worst case.
\end{enumerate}

Combining these costs yields the stated time complexity. The local flushing and edge-generator templates are precomputed for the fixed code family.
\end{proof}

The same accounting applies to a different toric decoding subroutine with its own running time. The correctness and threshold guarantees above are proved here for MWPM. Weighted MWPM and union-find may require modifications to those arguments.

\subsection{Cell-matching decoder threshold for the color code}
As proof of principle, we numerically estimate a threshold for our cell-matching decoder on the 6.6.6 hexagonal color code~\cite{bombin2006topological,bombin2007exact},
which is described by the polynomial $(1 + x + xy, 1 + y + xy)$.
We only consider bit-flip $X$ errors in the code capacity setting and flush on $3 \times 1$ unit cells containing
a $Z$ check of each color: red, green, and blue.
The local flushing on each blue-type check applies up to two nearby $X$ operators,
which are chosen to guarantee that non-trivial syndromes caused by any one- or two-qubit error would be removed.
The matching graphs for the two remaining color classes correspond to TCs on triangular lattices.
These are decoded independently by MWPM and lifted back to the color code to yield two-qubit short strings connecting checks of the same color.
As shown in Fig.~\ref{fig:color_threshold}, the resulting threshold is $p_{\mathrm{th}}=7.31\pm 0.05\%$,
which is lower than but comparable to the 8-9\% thresholds of preexisting optimized matching-based decoders~\cite{lee2025color,delfosse2014decoding,sahay2022decoder} or the $10\%$ thresholds of matching-based decoders on 4.8.8 square-octagonal color codes~\cite{kubica2023efficient,benhemou2025minimising,liu2025correlated} which share the same optimal threshold of 10.9\%~\cite{katzgraber2009error}.
The good performance of the basic cell-matching decoder on this simple example encourages an optimized adaptation for BB codes in the next section.
\begin{figure*}[t]
    \centering
    \includegraphics[width=0.49\linewidth]{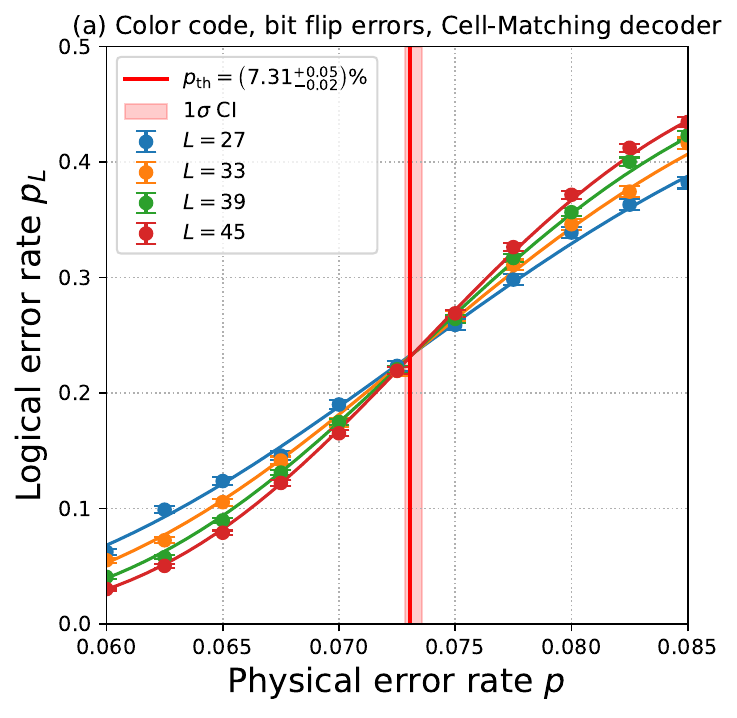}
    \includegraphics[width=0.49\linewidth]{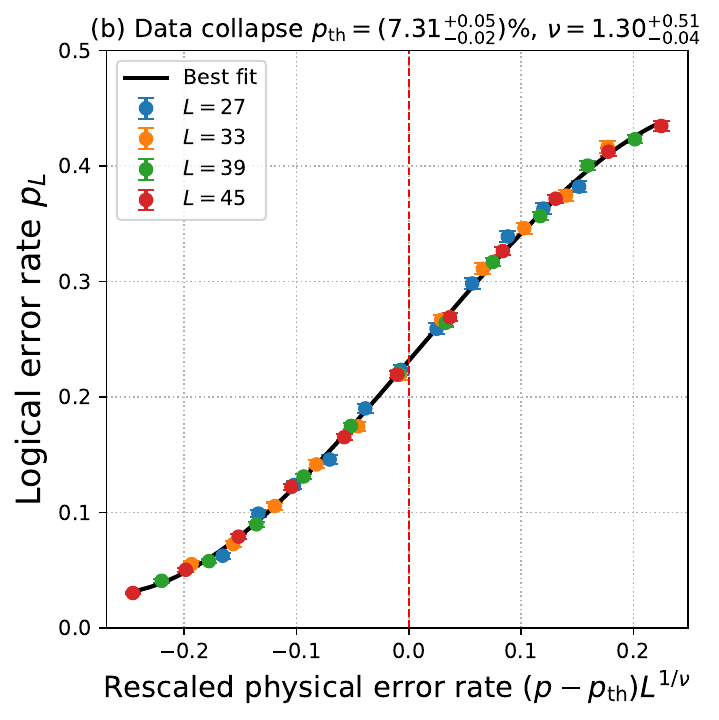}
    \caption{(a) The threshold of the hexagonal color code under i.i.d. bit flip errors using the cell-matching decoder is $p_{\mathrm{th}}=7.31\pm 0.05\%$. (b) Data collapse plot confirming the logical error rate fits the finite-size scaling ansatz $p_L(p, L) = F[(p - p_{\mathrm{th}}) L^{1/\nu}]$ for a cubic $F(x) = a_0 + a_1 x + a_2 x^2 + a_3 x^3$, which was used to estimate the threshold $p_{\mathrm{th}}$ with bootstrapped $1\sigma$ credible intervals.}
    \label{fig:color_threshold}
\end{figure*}

\makeatletter
\newenvironment{breakablealgorithm}
  {%
    \begin{center}
    \refstepcounter{algorithm}%
    \hrule height .8pt depth 0pt \kern 2pt%
    \renewcommand{\caption}[2][\relax]{%
      {\raggedright\textbf{\ALG@name~\thealgorithm} ##2\par}%
      \ifx\relax##1\relax
        \addcontentsline{loa}{algorithm}{\protect\numberline{\thealgorithm}##2}%
      \else
        \addcontentsline{loa}{algorithm}{\protect\numberline{\thealgorithm}##1}%
      \fi
      \kern 2pt\hrule\kern 2pt
    }%
  }
  {%
    \kern 2pt\hrule\relax%
    \end{center}
  }
\makeatother

\section{Adapting the cell-matching decoder for small and intermediate code sizes}
\label{sec:gross}
In the previous section, we developed the theory behind the cell-matching decoder. In this section we develop and implement matching decoders using the same principles which are specialized to small/intermediate code sizes. We first provide exposition for why such specialized decoders are necessary using BB codes as an example. As discussed in Ref. \cite{chen2025anyon}, the family of BB codes with the gross code polynomial has coarse-graining parameter equal to 12. This causes two problems:

\begin{itemize}
    \item \emph{Problem A (small lattice).} For the standard gross code, the lattice is not meaningfully larger than the coarse-graining scale (one dimension equals $b$), so the unit-cell construction underlying cell-matching degenerates.  \\ 
    \item \emph{Problem B (intermediate lattice).} For the $24\times 24$ gross code built from the same polynomial, coarse-graining yields an auxiliary toric-code matching instance with distance $d_{\mathrm{TC}}=\frac{L_x}{b}=\frac{L_y}{b}=2$, where $d_{\mathrm{TC}}$ is the distance of the induced toric-code matching graph; this is insufficient to correct single-qubit errors.
\end{itemize}

To address these problems, we develop two related adaptations of the cell-matching framework, one designed for small lattices (gross/two-gross) and one for intermediate lattices ($24\times 24$ gross). Although we present them in the context of these BB code examples, the construction applies to general small and intermediate \TwoDTLTI codes.  The BB codes we consider and their associated polynomials are given in Tab. \ref{tab:Bu_codes_considered}. \\ \indent We do not provide rigorous guarantees about the performance of these algorithms, as many elements in their construction are based on heuristics. The goal of this section is to provide evidence it is possible to adapt the cell-matching framework for practically relevant TTI codes rather than to present a provably optimized method of performing this adaptation. In the following analysis, we will assume that all errors are purely $X$ type and all checks are purely $Z$ type, but the exact same structure works for $Z$ type errors and $X$ type checks. The syndrome-dependent path costs, adaptive reweighting, and cluster selection below differ from the decoder analyzed in Theorem~\ref{thm:cellular-decoder-finite-size}. The correctness and threshold guarantees are not asserted for these variants.

\begin{table*}[t]
\centering
\begin{tabular}{lccc}
\hline
Code & Parameters & $(l,m)$ & Polynomial \\
\hline
Standard Gross Code
& $[[144,12,12]]$
& $(12,6)$
& $\left(x^3+y+y^2,\; y^3+x+x^2\right)$ \\

Two-Gross Code
& $[[288,12,18]]$
& $(12,12)$
& $\left(x^3+y^2+y^7,\; y^3+x+x^2\right)$ \\

$24\times 24$ Gross Code
& $[[1152,16,24]]$
& $(24,24)$
& $\left(x^3+y+y^2,\; y^3+x+x^2\right)$ \\
\hline
\end{tabular}
\caption{The code parameters, polynomials, and lattice sizes of the BB codes considered in this section. The corresponding BB code polynomials originate from \cite{bravyi2024high}.}
\label{tab:Bu_codes_considered}
\end{table*}

\subsection{Shared machinery for small and intermediate matching decoders }
We first provide the shared machinery common to both decoders presented.  
\subsubsection{Identifying equivalence class of syndromes}
In this section $H=H_{\ell,m}$ is the finite binary syndrome map for $X$ errors and $Z$ checks, represented over $R_{\ell,m}=\F_2[x^{\pm1},y^{\pm1}]/(x^\ell-1,y^m-1)$. Its cokernel is the quotient of finite physical syndrome vectors by $\im H_{\ell,m}$. This quotient need not coincide with the infinite-lattice charge space used in Section~\ref{sec:cellular_decoding}, particularly when the torus dimensions are not divisible by its coarse-graining scale. We choose a basis of this finite quotient and compute the coordinates of each individual check. The infinite-lattice coarse-graining scale $b$ preserves the induced finite classes, but smaller translations may preserve some or all of them. For a syndrome representative $s_i$ with class $u_i$, one can therefore seek $b'_i<b$ such that $\overline{x^{b'_i}s_i}=\overline{s_i}=u_i$ and $\overline{y^{b'_i}s_i}=\overline{s_i}=u_i$, where the overline denotes quotient by $\im H$. These equalities mean that some physical errors produce two copies of the pattern $s_i$ separated by $b'_i$ horizontally or vertically.
\begin{definition}
For each basis class $u_i$, choose a positive integer $b'_i$ such that
\begin{equation}
\overline{x^{b'_i}s_i}=u_i,
\qquad \overline{y^{b'_i}s_i}=u_i,
\label{eq:shortstringequiv}
\end{equation}
where $\overline{s_i}=u_i$, and write $\lVert u_i\rVert_T=b'_i$ for this chosen transport scale. We do not require the choice to be minimal. Horizontal and vertical \emph{short strings} are fixed error patterns satisfying the exact finite physical-syndrome equations
\begin{align}
H(e_{u_i})_H&=(1+x^{b'_i})s_i,\label{eq:bb-horizontal-short-string}\\
H(e_{u_i})_V&=(1+y^{b'_i})s_i.\label{eq:bb-vertical-short-string}
\end{align}
The quotient equalities ensure solvability on the finite torus; bounded physical support is a separate property of a chosen solution.
\label{def:translation_size}
\end{definition}
\indent We will want to choose the basis for  $\coker{H}$ such that as many basis elements as possible satisfy $||u_i||_T < b$, as this will lead to a higher effective distance for the matching graphs which we will ultimately construct. The canonical syndrome representation of the basis vectors used for the decoders presented in this section and their corresponding $||u_i||_T$ is given in Tab. \ref{tab:bb_basis_annihilators}. In the following discussion, we require $\lVert u_i\rVert_T\mid\ell,m$ for every $i$, and assume that the smallest chosen scale divides every other chosen scale. These conditions make all the square-cell partitions used below well-defined. This holds for the codes under consideration as seen by observing values in the column labeled $||u_i||_T$ in Tab.~\ref{tab:bb_basis_annihilators}. For the weight-3 basis elements of $\coker{H}$ given in Tab. \ref{tab:bb_basis_annihilators}, explicit short strings that transport these syndrome patterns by $x^3$ are shown in Fig. \ref{fig:3x3_shortstrings}. A procedure for choosing a desirable basis and for finding associated short strings is given in Appendix~\ref{app:computing-short-strings}.

With the basis now fixed, the following lemma tells us that one can compute the equivalence class of each single check on the lattice.

\begin{lemma}[Basis decomposition]
Let $r\coloneqq \dim(\coker\,H)$ and $n_Z$ be the number of $Z$ checks. Define $N$ as an $r \times n_Z$ matrix with entries in $\mathbb{F}_2$ and rows equal to a given basis for $\ker(H^T)=\im(H)^\perp$ and $B$ as the $n_Z \times r $ matrix which has columns of our chosen basis for $\coker\,H$. Finding the equivalence class $u_s \in \mathbb{F}^r_2 \simeq \coker{H}$ of a given syndrome $s \in \mathbb{F}^{n_Z}_2$ in our basis is equivalent to solving the linear system 
\begin{equation}
    NB(u_s)=Ns,
\end{equation}
and such a solution $u_s$ always exists and is unique. 
\end{lemma}
\begin{proof}
First we show why an  r-dimensional basis in $F_2^{n_Z}$ exists for $\ker(H^T)=\im(H)^\perp$. This follows from $H$ being a $n_Z  \times n$ matrix such that
\begin{align}
    \dim(\ker(H^T))
    &= \dim(\mathbb{F}_2^{n_Z})-\operatorname{rank}(H^T)\nonumber\\
    &= \dim(\mathbb{F}_2^{n_Z})-\operatorname{rank}(H)\nonumber\\
    &= \dim(\coker\,H)=r.
\end{align}
Now we prove the actual statement. Assume that $Ns_1=Ns_2$ but $\overline{s_1}\neq \overline{s_2}$. then $N(s_1-s_2)=0$. This implies that $s_1-s_2$ is orthogonal to the rows of $N$ which is the basis for $\im(H)^\perp$. Thus, $s_1-s_2\in \text{im}(H)$, so there must exist some error $e$ such that $He=(s_1-s_2)$, so $\overline{s_1-s_2}=\overline{s_1}-\overline{s_2}=0$, which is a contradiction. Thus, if $Ns=NB(u_s)$, $\overline{s}=\overline{B(u_s)}=u_s$  by construction of $B$. Furthermore, $B$ is full rank, and the kernel of $N$ is disjoint from the image of $B$, so $NB$ is invertible and the linear system is solvable. Uniqueness is manifest.
\end{proof}
 
\begin{table*}[t]
\centering
\renewcommand{\arraystretch}{1.2}
\begin{tabular}{|c|c|c|c|c|}
\hline
Code & Basis label $i$ & Basis element & $\lVert u_i\rVert_T$ & Annihilating elements \\ \hline

\multirow{8}{*}{$[[1152,16,24]]$}
& $1$ & $x$ & $12$ & $x^{12}-1,\; y^{12}-1$ \\ \cline{2-5}
& $2$ & $x^2$ & $12$ & $x^{12}-1,\; y^{12}-1$ \\ \cline{2-5}
& $3$ & $x + x^2y$ & $6$ & $x^{6}-1,\; y^{6}-1$ \\ \cline{2-5}
& $4$ & $xy + x^2y^2$ & $6$ & $x^{6}-1,\; y^{6}-1$ \\ \cline{2-5}
& $5$ & $x + y + xy$ & $3$ & $x^{3}-1,\; y^{3}-1$ \\ \cline{2-5}
& $6$ & $x^2 + xy + x^2y$ & $3$ & $x^{3}-1,\; y^{3}-1$ \\ \cline{2-5}
& $7$ & $y^2 + xy + xy^2$ & $3$ & $x^{3}-1,\; y^{3}-1$ \\ \cline{2-5}
& $8$ & $x^2y + xy^2 + x^2y^2$ & $3$ & $x^{3}-1,\; y^{3}-1$ \\ \hline

\multirow{6}{*}{$[[144,12,12]], [[288,12,18]]$}
& $1$ & $x + x^2y$ & $6$ & $x^{6}-1,\; y^{6}-1$ \\ \cline{2-5}
& $2$ & $x^2y + x^2y^2$ & $6$ & $x^{6}-1,\; y^{6}-1$ \\ \cline{2-5}
& $3$ & $xy + x^2y^2$ & $6$ & $x^{6}-1,\; y^{6}-1$ \\ \cline{2-5}
& $4$ & $xy + xy^2$ & $6$ & $x^{6}-1,\; y^{6}-1$ \\ \cline{2-5}
& $5$ & $x + y + xy$ & $3$ & $x^{3}-1,\; y^{3}-1$ \\ \cline{2-5}
& $6$ & $x^2y + xy^2 + x^2y^2$ & $3$ & $x^{3}-1,\; y^{3}-1$ \\ \hline

\end{tabular}
\caption{Basis representatives, chosen transport scales $\lVert u_i\rVert_T$, and annihilating elements in $\operatorname{coker}(H)$ for  BB codes under consideration. While the gross and two-gross code have the same basis for $\coker{H}$, the errors which realize the short strings are different.}
\label{tab:bb_basis_annihilators}
\end{table*}

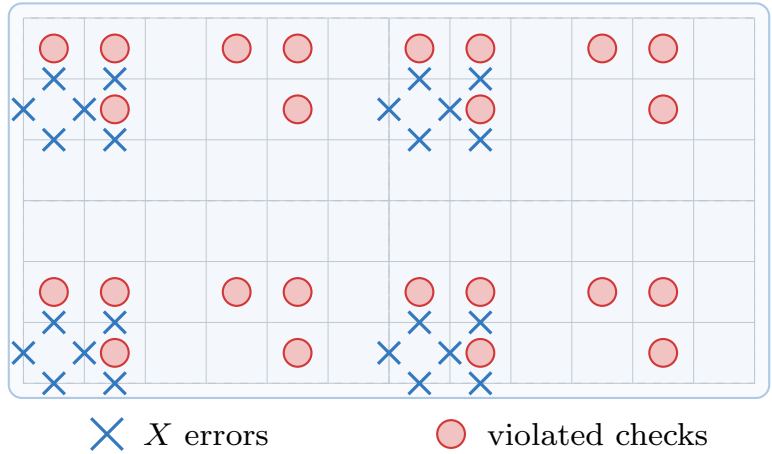
\begin{figure*}[t]
\centering
\definecolor{siteblue}{RGB}{52,120,190}
\definecolor{stencilred}{RGB}{210,55,55}
\definecolor{stencilgreen}{RGB}{40,160,80}
\definecolor{stencilpurple}{RGB}{130,60,180}
\definecolor{coarsegold}{RGB}{200,150,20}
\definecolor{gridgray}{RGB}{190,200,210}
\definecolor{panelbg}{RGB}{248,250,253}
\definecolor{headerbg}{RGB}{230,238,250}
\definecolor{cellbg}{RGB}{235,245,255}

\tikzset{
  panelbox/.style={rectangle, rounded corners=5pt, draw=siteblue!40, fill=panelbg, thick},
  headerbox/.style={rectangle, rounded corners=3pt, fill=headerbg, draw=siteblue!60, thick, inner sep=4pt},
  syndromeq/.style={circle, draw=stencilred, fill=stencilred!30, thick, minimum size=7pt, inner sep=0pt},
  qerror/.style={siteblue, line width=1.2pt},
}

\providecommand{\Graydt}[2]{}
\providecommand{\qHxAt}[2]{%
  \draw[qerror] ({#1-0.18*\s},{#2-0.5*\s-0.18*\s}) -- ({#1+0.18*\s},{#2-0.5*\s+0.18*\s});%
  \draw[qerror] ({#1-0.18*\s},{#2-0.5*\s+0.18*\s}) -- ({#1+0.18*\s},{#2-0.5*\s-0.18*\s});%
}
\providecommand{\qVxAt}[2]{%
  \draw[qerror] ({#1-0.5*\s-0.18*\s},{#2-0.18*\s}) -- ({#1-0.5*\s+0.18*\s},{#2+0.18*\s});%
  \draw[qerror] ({#1-0.5*\s-0.18*\s},{#2+0.18*\s}) -- ({#1-0.5*\s+0.18*\s},{#2-0.18*\s});%
}
\providecommand{\ERR}[2]{\qHxAt{#1}{#2}}
\providecommand{\BOTH}[2]{\qHxAt{#1}{#2} \qVxAt{#1}{#2}}
\providecommand{\SYN}[2]{\node[syndromeq] at (#1,#2) {};}
\providecommand{\SYNERR}[2]{\SYN{#1}{#2} \qHxAt{#1}{#2}}
\providecommand{\SYNBOTH}[2]{\SYN{#1}{#2} \qHxAt{#1}{#2} \qVxAt{#1}{#2}}

\begin{tikzpicture}[font=\small,scale=1.5, transform shape]

\def\s{0.54}      
\def\hg{\s}      
\def\vg{\s}       
\def\lm{2.05}     
\def\bm{0.82}     
\def\pd{0.00}     

\pgfmathsetmacro\PW{2*\s + 2*\pd}
\pgfmathsetmacro\PH{2*\s + 2*\pd}
\pgfmathsetmacro\colstep{\PW + \hg}
\pgfmathsetmacro\rowstep{\PH + \vg}
\pgfmathsetmacro\gridW{4*\colstep - \hg}
\pgfmathsetmacro\gridH{2*\rowstep - \vg}

\pgfmathsetmacro\padX{0.4}
\pgfmathsetmacro\padY{0.4}

\pgfmathsetmacro\gridCx{\lm + 0.5*\gridW}
\pgfmathsetmacro\gridCy{\bm + 0.5*\gridH}

\pgfmathsetlengthmacro\panelWlen{(\gridW + 2*\padX)*1cm}
\pgfmathsetlengthmacro\panelHlen{(\gridH + 2*\padY)*1cm}

\node[panelbox, minimum width=\panelWlen, minimum height=\panelHlen]
  at (\gridCx,\gridCy) {};

\pgfmathsetmacro\xL{\lm - 0.5*\s}
\pgfmathsetmacro\yB{\bm - 0.5*\s}
\pgfmathsetmacro\xR{\xL + 12*\s}
\pgfmathsetmacro\yT{\yB + 6*\s}

\fill[cellbg, opacity=0.30] (\xL,\yB) rectangle (\xR,\yT);

\foreach \i in {0,...,12} {
  \draw[gridgray, very thin] ({\xL+\i*\s},\yB) -- ({\xL+\i*\s},\yT);
}
\foreach \j in {0,...,6} {
  \draw[gridgray, very thin] (\xL,{\yB+\j*\s}) -- (\xR,{\yB+\j*\s});
}

\foreach \cp in {0,1} {
  \foreach \rp in {0,1} {
    \pgfmathsetmacro\bx{\lm + 2*\cp*\colstep - 0.5*\s}
    \pgfmathsetmacro\by{\bm + \rp*\rowstep - 0.5*\s}
    \pgfmathsetmacro\bxr{\bx + 6*\s}
    \pgfmathsetmacro\byt{\by + 3*\s}
    \draw[gridgray, dashed] (\bx, \by) rectangle (\bxr, \byt);
  }
}


\pgfmathsetmacro\sxAA{\lm + 0*\colstep + \pd + 0*\s}
\pgfmathsetmacro\sxAB{\lm + 0*\colstep + \pd + 1*\s}
\pgfmathsetmacro\sxAC{\lm + 0*\colstep + \pd + 2*\s}
\pgfmathsetmacro\sxBA{\lm + 1*\colstep + \pd + 0*\s}
\pgfmathsetmacro\sxBB{\lm + 1*\colstep + \pd + 1*\s}
\pgfmathsetmacro\sxBC{\lm + 1*\colstep + \pd + 2*\s}
\pgfmathsetmacro\sxCA{\lm + 2*\colstep + \pd + 0*\s}
\pgfmathsetmacro\sxCB{\lm + 2*\colstep + \pd + 1*\s}
\pgfmathsetmacro\sxCC{\lm + 2*\colstep + \pd + 2*\s}
\pgfmathsetmacro\sxDA{\lm + 3*\colstep + \pd + 0*\s}
\pgfmathsetmacro\sxDB{\lm + 3*\colstep + \pd + 1*\s}
\pgfmathsetmacro\sxDC{\lm + 3*\colstep + \pd + 2*\s}

\pgfmathsetmacro\syT{\bm + 1*\rowstep + \pd + 2*\s}
\pgfmathsetmacro\syM{\bm + 1*\rowstep + \pd + 1*\s}
\pgfmathsetmacro\syB{\bm + 1*\rowstep + \pd + 0*\s}
\pgfmathsetmacro\syLT{\bm + 0*\rowstep + \pd + 2*\s}
\pgfmathsetmacro\syLM{\bm + 0*\rowstep + \pd + 1*\s}
\pgfmathsetmacro\syLB{\bm + 0*\rowstep + \pd + 0*\s}


\SYNERR{\sxAA}{\syT}    \SYNERR{\sxAB}{\syT}      \Graydt{\sxAC}{\syT}
\BOTH{\sxAA}{\syM}      \SYNBOTH{\sxAB}{\syM}     \Graydt{\sxAC}{\syM}
\Graydt{\sxAA}{\syB}    \Graydt{\sxAB}{\syB}      \Graydt{\sxAC}{\syB}

\SYN{\sxBA}{\syT}    \SYN{\sxBB}{\syT}         \Graydt{\sxBC}{\syT}
\Graydt{\sxBA}{\syM}    \SYN{\sxBB}{\syM}      \Graydt{\sxBC}{\syM}
\Graydt{\sxBA}{\syB}    \Graydt{\sxBB}{\syB}      \Graydt{\sxBC}{\syB}

\begin{scope}[xshift={-\s cm}]

\Graydt{\sxCA}{\syT}    \SYNERR{\sxCB}{\syT}      \SYNERR{\sxCC}{\syT}
\Graydt{\sxCA}{\syM}    \BOTH{\sxCB}{\syM}        \SYNBOTH{\sxCC}{\syM}
\Graydt{\sxCA}{\syB}    \Graydt{\sxCB}{\syB}      \Graydt{\sxCC}{\syB}

\Graydt{\sxDA}{\syT}    \SYN{\sxDB}{\syT}         \SYN{\sxDC}{\syT}
\Graydt{\sxDA}{\syM}    \Graydt{\sxDB}{\syM}      \SYN{\sxDC}{\syM}
\Graydt{\sxDA}{\syB}    \Graydt{\sxDB}{\syB}      \Graydt{\sxDC}{\syB}

\end{scope}


\Graydt{\sxAA}{\syLT}   \Graydt{\sxAB}{\syLT}     \Graydt{\sxAC}{\syLT}
\SYNERR{\sxAA}{\syLM}   \SYNERR{\sxAB}{\syLM}     \Graydt{\sxAC}{\syLM}
\BOTH{\sxAA}{\syLB}     \SYNBOTH{\sxAB}{\syLB}    \Graydt{\sxAC}{\syLB}

\Graydt{\sxBA}{\syLT}   \Graydt{\sxBB}{\syLT}     \Graydt{\sxBC}{\syLT}
\SYN{\sxBA}{\syLM}      \SYN{\sxBB}{\syLM}        \Graydt{\sxBC}{\syLM}
\Graydt{\sxBA}{\syLB}   \SYN{\sxBB}{\syLB}        \Graydt{\sxBC}{\syLB}

\begin{scope}[xshift={-\s cm}]

\Graydt{\sxCA}{\syLT}   \Graydt{\sxCB}{\syLT}     \Graydt{\sxCC}{\syLT}
\Graydt{\sxCA}{\syLM}   \SYNERR{\sxCB}{\syLM}     \SYNERR{\sxCC}{\syLM}
\Graydt{\sxCA}{\syLB}     \BOTH{\sxCB}{\syLB}    \SYNBOTH{\sxCC}{\syLB}

\Graydt{\sxDA}{\syLT}   \Graydt{\sxDB}{\syLT}     \Graydt{\sxDC}{\syLT}
\Graydt{\sxDA}{\syLM}      \SYN{\sxDB}{\syLM}        \SYN{\sxDC}{\syLM}
\Graydt{\sxDA}{\syLB}   \Graydt{\sxDB}{\syLB}        \SYN{\sxDC}{\syLB}

\end{scope}

\pgfmathsetmacro\ly{\yB - 0.45}
\pgfmathsetmacro\lsep{2.5}

\pgfmathsetmacro\mxL{\gridCx - \lsep}
\draw[qerror] ({\mxL-0.14},{\ly-0.14}) -- ({\mxL+0.14},{\ly+0.14});
\draw[qerror] ({\mxL-0.14},{\ly+0.14}) -- ({\mxL+0.14},{\ly-0.14});

\pgfmathsetmacro\mxR{\gridCx + 0.55}
\node[syndromeq] at (\mxR,\ly) {};

\begin{scope}[font=\scriptsize]
\node[anchor=base west, inner sep=0pt, outer sep=0pt, text depth=0pt] 
  at ({\mxL+0.32},\ly - 0.1) {$X$ errors};
\node[anchor=base west, inner sep=0pt, outer sep=0pt, text depth=0pt] 
  at ({\mxR+0.32},\ly - 0.1) {violated checks};
\end{scope}

\pgfresetboundingbox
\path[use as bounding box] (\xL,\ly-0.18) rectangle (\xR,\yT);

\end{tikzpicture}
\caption{Examples of four short string  $X$ error patterns (blue X markers) for the gross code polynomial which respectively translate each of the weight-3 check-violation type patterns (red circles) presented in Tab.~\ref{tab:bb_basis_annihilators} horizontally by distance three, shown on a $12 \times 6$ lattice with origin in the bottom left corner. Each error pattern occupies one $6\times 3$ quadrant of the $12 \times 6$ lattice presented.}
\label{fig:3x3_shortstrings}
\end{figure*}

As a consequence of the lemma, for each check site $c$ on the lattice, we can assign a vector $u_c$ in $\mathbb{F}^r_2$ corresponding to the equivalence class of that check in $\coker\,H$. To reiterate the meaning of $u_s$,  for arbitrary syndrome $s$ understood as an element of $R$, letting $(u_s)_i$ give the $i$th component of $u_s$, we have that in $\coker{H}$:
\begin{equation}
    \overline{s}+\sum_{i=1}^r (u_s)_i u_i=0.
    \label{eq:equiv_coker}
\end{equation}
 Before any decoding begins, for each check $c$ on the lattice, we compute a dictionary which gives $u_c$ for each check. Next, we partition the checks in the lattice into unit cells of size $\text{min}_i\{||u_i||_T \}\times \text{min}_i\{||u_i||_T\}$, which induces a partitioning of the lattice into cells of size $||u_j||_T\times ||u_j||_T$ for $j \in \{1,\dots,r\}$  by the division property. 

\begin{definition}
Let $a \in \{0,\dots,\frac{\ell}{||u_i||_T}-1\}$ and $d \in \{0,\dots,\frac{m}{||u_i||_T}-1\}$ be the coarse grained coordinates.  For a given basis element $u_i$ and single check $c$ on the lattice understood as a polynomial, let  $C_{||u_i||_T}(c)=x^{a||u_i||_T}y^{d||u_i||_T}$ be the monomial which specifies the coordinates of the corner closest to the origin of the $||u_i||_T \times ||u_i||_T$ unit cell that this check occupies.   
\end{definition}
Using Definition \ref{def:translation_size} and Eq. \ref{eq:equiv_coker}, we have that
\begin{align}
\overline{C_{\lVert u_i\rVert_T}(c)s_i}&=u_i,\nonumber\\
\overline c+\sum_{i=1}^r(u_c)_i
\overline{C_{\lVert u_i\rVert_T}(c)s_i}&=0.
\label{eq:decomp}
\end{align}
These equalities hold in $\coker H$. Consequently, there exists an error $e_c$ such that
\begin{equation}
    He_c=c+\sum_{i=1}^r (u_c)_i C_{||u_i||_T}(c)s_i, 
    \label{eq:rewrite}
\end{equation}
where $s_i$ is the canonical syndrome pattern of $u_i$. This is an exact equation in the finite physical syndrome space. It guarantees that a rewrite operator exists, but does not by itself bound its support or fix its logical class. Unlike the fixed infinite-lattice templates in Section~\ref{sec:cellular_decoding}, a solution obtained independently on each finite torus need not have uniformly bounded spread. As a result of this equation, we can interpret $e_c$ to be a ``rewrite" operator which decomposes a check $c$ into $\coker{H}$ basis syndrome representative which are most local to $c$.   An example of a single check being decomposed into its translated canonical syndrome representatives of the basis elements of $\coker{H}$ is given in Fig.~\ref{fig:basis_decomp} for the $24 \times 24$ gross code.

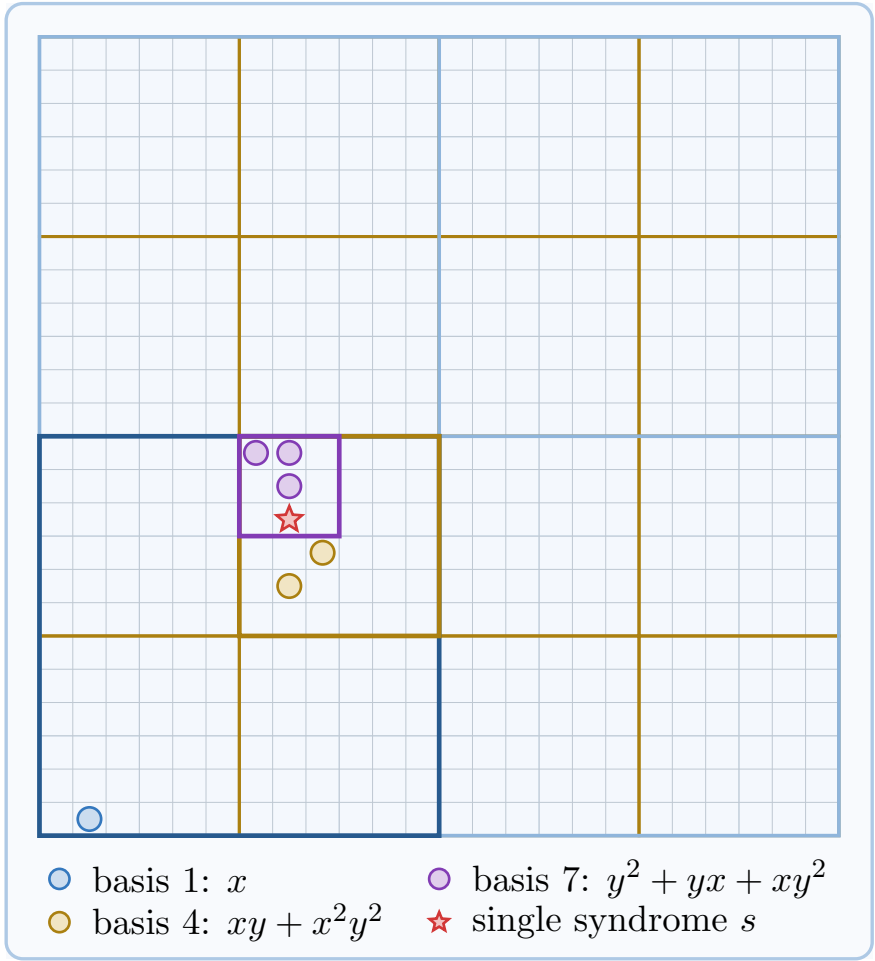
\begin{figure*}[t]

\centering
\definecolor{siteblue}{RGB}{52,120,190}
\definecolor{stencilred}{RGB}{210,55,55}
\definecolor{stencilgreen}{RGB}{130,60,180}
\definecolor{stencilpurple}{RGB}{40,160,80}
\definecolor{coarsegold}{RGB}{200,150,20}
\definecolor{gridgray}{RGB}{190,200,210}
\definecolor{panelbg}{RGB}{248,250,253}
\definecolor{headerbg}{RGB}{230,238,250}
\definecolor{cellbg}{RGB}{235,245,255}

\tikzset{
  panelbox/.style={rectangle, rounded corners=5pt, draw=siteblue!40, fill=panelbg, thick},
  syndromeq/.style={star, star points=5, star point ratio=2.25, draw=stencilred, fill=stencilred!30, thick, minimum size=8pt, inner sep=0pt},
  basis1q/.style={circle, draw=siteblue, fill=siteblue!25, thick, minimum size=7pt, inner sep=0pt},
  basis4q/.style={circle, draw=coarsegold!85!black, fill=coarsegold!25, thick, minimum size=7pt, inner sep=0pt},
  basis7q/.style={circle, draw=stencilgreen, fill=stencilgreen!25, thick, minimum size=7pt, inner sep=0pt},
  qerror/.style={siteblue, line width=1.2pt},
}

\newcommand{\qHx}[2]{%
  \draw[qerror] (#1, {#2+0.5}) -- ({#1+1}, {#2+0.5});%
  \draw[qerror] ({#1+0.5-0.18}, {#2+0.5-0.18}) -- ({#1+0.5+0.18}, {#2+0.5+0.18});%
  \draw[qerror] ({#1+0.5-0.18}, {#2+0.5+0.18}) -- ({#1+0.5+0.18}, {#2+0.5-0.18});%
}
\newcommand{\qVx}[2]{%
  \draw[qerror] ({#1+0.5}, #2) -- ({#1+0.5}, {#2+1});%
  \draw[qerror] ({#1+0.5-0.18}, {#2+0.5-0.18}) -- ({#1+0.5+0.18}, {#2+0.5+0.18});%
  \draw[qerror] ({#1+0.5-0.18}, {#2+0.5+0.18}) -- ({#1+0.5+0.18}, {#2+0.5-0.18});%
}

\resizebox{0.70\linewidth}{!}{%
\begin{tikzpicture}[x=0.35cm, y=0.35cm, yscale=-1, font=\small]


  \fill[panelbg] (-1.0,-1.0) rectangle (25.0,27.7);
  \draw[siteblue!40, line width=1.0pt, rounded corners=5pt] (-1.0,-1.0) rectangle (25.0,27.7);

  \foreach \gc in {0,...,7} {
    \foreach \gr in {0,...,7} {
      \fill[cellbg, opacity=0.30] (3*\gc,3*\gr) rectangle (3*\gc+3,3*\gr+3);
      \draw[siteblue!18, line width=0.35pt] (3*\gc,3*\gr) rectangle (3*\gc+3,3*\gr+3);
    }
  }

  \foreach \i in {0,...,24} {
    \draw[gridgray, very thin] (\i,0) -- (\i,24);
    \draw[gridgray, very thin] (0,\i) -- (24,\i);
  }

  \foreach \GC in {0,...,3} {
    \foreach \GR in {0,...,3} {
      \draw[coarsegold!85!black, line width=0.9pt]
        (6*\GC,6*\GR) rectangle (6*\GC+6,6*\GR+6);
    }
  }


  \draw[siteblue!55, line width=1.1pt] (0,0) rectangle (24,24);
  \draw[siteblue!55, line width=1.1pt] (12,0) -- (12,24);
  \draw[siteblue!55, line width=1.1pt] (0,12) -- (24,12);
  \draw[siteblue!75!black, line width=1.4pt] (0,12) rectangle (12,24);

  \draw[coarsegold!85!black, line width=1.4pt] (6,12) rectangle (12,18);

  \draw[stencilgreen, line width=1.4pt] (6,12) rectangle (9,15);
  \node[basis1q] at (1.5,23.5) {};

  \node[basis4q] at (7.5,16.5) {};
  \node[basis4q] at (8.5,15.5) {};

  \node[basis7q] at (7.5,13.5) {};
  \node[basis7q] at (6.5,12.5) {};
  \node[basis7q] at (7.5,12.5) {};

  \node[syndromeq] at (7.5,14.5) {};

  \node[basis1q, minimum size=6pt] at (0.6,25.3) {};
  \node[anchor=west, font=\normalsize] at (1.2,25.3) {basis 1: $x$};

  \node[basis7q, minimum size=6pt] at (12.0,25.3) {};
  \node[anchor=west, font=\normalsize] at (12.6,25.3) {basis 7: $y^2+yx+xy^2$};

  \node[basis4q, minimum size=6pt] at (0.6,26.6) {};
  \node[anchor=west, font=\normalsize] at (1.2,26.6) {basis 4: $xy+x^2y^2$};

  \node[syndromeq, minimum size=6pt] at (12.0,26.6) {};
  \node[anchor=west, font=\normalsize] at (12.6,26.6) {single syndrome $s$};

\end{tikzpicture}%
}
\caption{%
  Basis decomposition of a single violated check $s$ (red star).
  The full $24\times 24$ grid of check sites is shown; orange-brown borders
  outline $6\times 6$ coarse cells.
  The colored circles indicate the translated canonical syndrome
  representatives used to decompose~$s$: blue (basis~1~$=x$),
  golden (basis~4~$=xy+x^2y^2$), and purple (basis~7~$=y^2+yx+xy^2$).
  The bold golden and blue rectangles highlight
  the $6\!\times\!6$ and $12\!\times\!12$  coarse cells respectively containing~$s$.%
}
\label{fig:basis_decomp}
\end{figure*}

\subsubsection{Defining the matching graphs}
Once we have constructed the dictionary which maps a single check $c$ to the basis decomposition vector $u_c$, we can now define the graphs we will use for the matching decoding process given an observed syndrome $s_{obs}$. For $i \in \{1,\dots, r\}$, define the graph $\mathcal{G}_i=(V_i,E_i)$ to be the complete graph with vertex set $V_i=\{c \ | c \in \text{supp}(s_{obs}) \text{ and } (u_c)_i=1 \}$. An illustration of the matching graphs is given in Fig. \ref{fig:matching_graphs}. 
\begin{lemma}
$|V_i| \text{ mod } 2=0$ for all $i \in \{1,\dots, r\}$  meaning it is possible to obtain a perfect matching on each $\mathcal{G}_i$. 
\end{lemma}
\begin{proof}
    First, observe that  $\overline{s_{obs}}=0$ because $s_{obs}$ is a physically achievable syndrome. Using this fact and Eq. \ref{eq:decomp} we have that 
\begin{equation}
\begin{aligned}
&\sum_{i=1}^r\left(\sum_{c \in \supp(s_{obs})} (u_c)_i\right) u_i\\
&\quad=\sum_{c \in \supp(s_{obs})}\sum_{i=1}^r (u_c)_i u_i\\
&\quad=\sum_{c \in \supp(s_{obs})} \overline{c}\\
&\quad=\overline{\sum_{c \in \supp(s_{obs})} c}\\
&\quad=\overline{s_{obs}}=0,
\end{aligned}
\end{equation}
which by the independence of the $u_i$ implies that for all $i \in \{1,\dots, r\}$, we have that \begin{equation} 
\sum_{c \in \text{supp}(s_{obs})} (u_c)_i=0.
\end{equation}
This completes the proof. 
\end{proof}
Now, we must define the weights for the edges on each matching graph. To understand the intuition behind the construction of the weights, we need the following definition,
\begin{definition}
Assume that the observed syndrome $s_{obs}$ was produced by an error pattern $e$, and let
$c_1, c_2 \in \operatorname{supp}(s_{obs}).$ We say that $c_1$ and $c_2$ are \emph{error-cluster-equivalent}, and write
$c_1 \sim_e c_2$,
if there does not exist an error pattern $e'$ with
\begin{equation}
\operatorname{supp}(e') \subset \operatorname{supp}(e)
\end{equation}
such that
\begin{equation}
\operatorname{supp}(He') \subseteq \operatorname{supp}(s_{obs}),
\end{equation}
and exactly one of the two statements
\begin{equation}
c_1 \in \operatorname{supp}(He'), \qquad c_2 \in \operatorname{supp}(He')
\end{equation}
holds. Equivalently, $c_1 \sim_e c_2$ means that no sub-error of $e$, whose syndrome remains entirely inside the observed syndrome, can separate $c_1$ from $c_2$. This defines an equivalence relation on the individual checks in $\operatorname{supp}(s_{obs})$. Now suppose that $c_1 \sim_e c_2$. An \emph{error cluster} associated with $c_1$ and $c_2$ is an error pattern $e'$ with $\operatorname{supp}(e')\subseteq\operatorname{supp}(e)$ such that
\begin{equation}
\begin{gathered}
c_1, c_2 \in \operatorname{supp}(He'),\\
c,c' \in \operatorname{supp}(He') \implies c \sim_e c',
\end{gathered}
\end{equation}
and $e'$ contains no proper nonempty subset whose syndrome is zero. Since error-cluster-equivalence is defined only for checks in $\operatorname{supp}(s_{obs})$, every error cluster satisfies
\begin{equation}
\operatorname{supp}(He') \subseteq \operatorname{supp}(s_{obs}).
\end{equation}
\end{definition}

\begin{figure*}[t]
\centering
\definecolor{siteblue}{RGB}{52,120,190}
\definecolor{stencilred}{RGB}{210,55,55}
\definecolor{coarsegold}{RGB}{200,150,20}
\definecolor{gridgray}{RGB}{190,200,210}
\definecolor{panelbg}{RGB}{248,250,253}
\definecolor{cellbg}{RGB}{235,245,255}
\definecolor{blackdot}{RGB}{25,25,25}

\tikzset{
  checkobs/.style={circle, draw=blackdot, fill=blackdot!20, line width=0.9pt, minimum size=7.5pt, inner sep=0pt},
  checkactive/.style={circle, draw=stencilred, fill=stencilred!25, line width=0.9pt, minimum size=7.0pt, inner sep=0pt},
  checkinactive/.style={circle, draw=gridgray!90!black, fill=gridgray!25, line width=0.9pt, minimum size=6.2pt, inner sep=0pt},
}

\newcommand{\DrawLatticeTwelveBySix}{%
  \fill[panelbg] (-0.85,-0.55) rectangle (12.85,6.55);
  \draw[siteblue!40, line width=1.0pt, rounded corners=4pt] (-0.85,-0.55) rectangle (12.85,6.55);

  \fill[cellbg, opacity=0.30] (0,0) rectangle (12,6);

  \foreach \i in {0,...,12} {
    \draw[gridgray, very thin] (\i,0) -- (\i,6);
  }
  \foreach \j in {0,...,6} {
    \draw[gridgray, very thin] (0,\j) -- (12,\j);
  }

  \foreach \i in {0,...,12} {
    \foreach \j in {0,...,6} {
      \fill[gridgray!65] (\i,\j) circle (0.55pt);
    }
  }

  \draw[coarsegold!85!black, line width=0.9pt] (0,0) rectangle (6,6);
  \draw[coarsegold!85!black, line width=0.9pt] (6,0) rectangle (12,6);
}

\newcommand{\Rdot}[2]{\node[checkactive] at (#1,#2) {};}
\newcommand{\Gdot}[2]{\node[checkinactive] at (#1,#2) {};}
\newcommand{\Bdot}[2]{\node[checkobs] at (#1,#2) {};}
\newcommand{\Edge}[4]{\draw[siteblue, line width=1.2pt] (#1,#2) -- (#3,#4);}

\resizebox{\linewidth}{!}{%
\begin{tikzpicture}[x=0.40cm,y=0.40cm,font=\small]

\begin{scope}[shift={(0,8)}]
  \DrawLatticeTwelveBySix
  \node[font=\footnotesize] at (6.0,6.95) {{Observed Syndrome}};
  \begin{scope}[shift={(0.5,0.5)}]
    \Bdot{3}{1}
    \Bdot{6}{1}
    \Bdot{4}{2}
    \Bdot{5}{3}
    \Bdot{4}{5}
  \end{scope}
\end{scope}

\begin{scope}[shift={(15,8)}]
  \DrawLatticeTwelveBySix
  \node[font=\footnotesize] at (6.0,6.95) {{Basis $k=1$}};
  \begin{scope}[shift={(0.5,0.5)}]
    \Edge{3}{1}{6}{1}
    \Edge{3}{1}{4}{2}
    \Edge{3}{1}{5}{3}
    \Edge{6}{1}{4}{2}
    \Edge{6}{1}{5}{3}
    \Edge{4}{2}{5}{3}
    \Gdot{4}{5}
    \Rdot{3}{1}
    \Rdot{6}{1}
    \Rdot{4}{2}
    \Rdot{5}{3}
  \end{scope}
\end{scope}

\begin{scope}[shift={(0,0)}]
  \DrawLatticeTwelveBySix
  \node[font=\footnotesize] at (6.0,6.95) {{Basis $k=2$}};
  \begin{scope}[shift={(0.5,0.5)}]
    \Edge{6}{1}{5}{3}
    \Gdot{3}{1}
    \Gdot{4}{2}
    \Gdot{4}{5}
    \Rdot{6}{1}
    \Rdot{5}{3}
  \end{scope}
\end{scope}

\begin{scope}[shift={(15,0)}]
  \DrawLatticeTwelveBySix
  \node[font=\footnotesize] at (6.0,6.95) {{Basis $k=3$}};
  \begin{scope}[shift={(0.5,0.5)}]
    \Edge{6}{1}{5}{3}
    \Gdot{3}{1}
    \Gdot{4}{2}
    \Gdot{4}{5}
    \Rdot{6}{1}
    \Rdot{5}{3}
  \end{scope}
\end{scope}

\Bdot{2.0}{-2.0}
\node[anchor=west, font=\footnotesize] at (2.7,-2.0) {violated check};

\Rdot{14.0}{-2.0}
\node[anchor=west, font=\footnotesize] at (14.7,-2.0) {violated check in $V_k$};

\Gdot{2.0}{-3.5}
\node[anchor=west, font=\footnotesize] at (2.7,-3.5) {violated check not in $V_k$};

\Edge{13.6}{-3.5}{14.6}{-3.5}
\node[anchor=west, font=\footnotesize] at (14.7,-3.5) {edge in complete graph on $V_k$};

\end{tikzpicture}%
}
\caption{Complete matching graphs $\mathcal{G}_k$ for $k=1,2,3$, shown on the $12 \times 6$ gross code lattice, 
  constructed from an observed syndrome consisting of five violated checks (black circles,
  top left panel). Note that 3 other complete matching graphs exist for this code which are not shown. 
  Red circles are violated checks containing basis~$k$ in the basis decomposition of their equivalence class in $\coker{H}$; gray circles are violated checks
  assigned to other bases.
  Cyan edges indicate all possible pairings in the complete graph on the red vertices; the minimum-weight perfect matching selects the subset of edges
  used for decoding.}
\label{fig:matching_graphs}
\end{figure*}
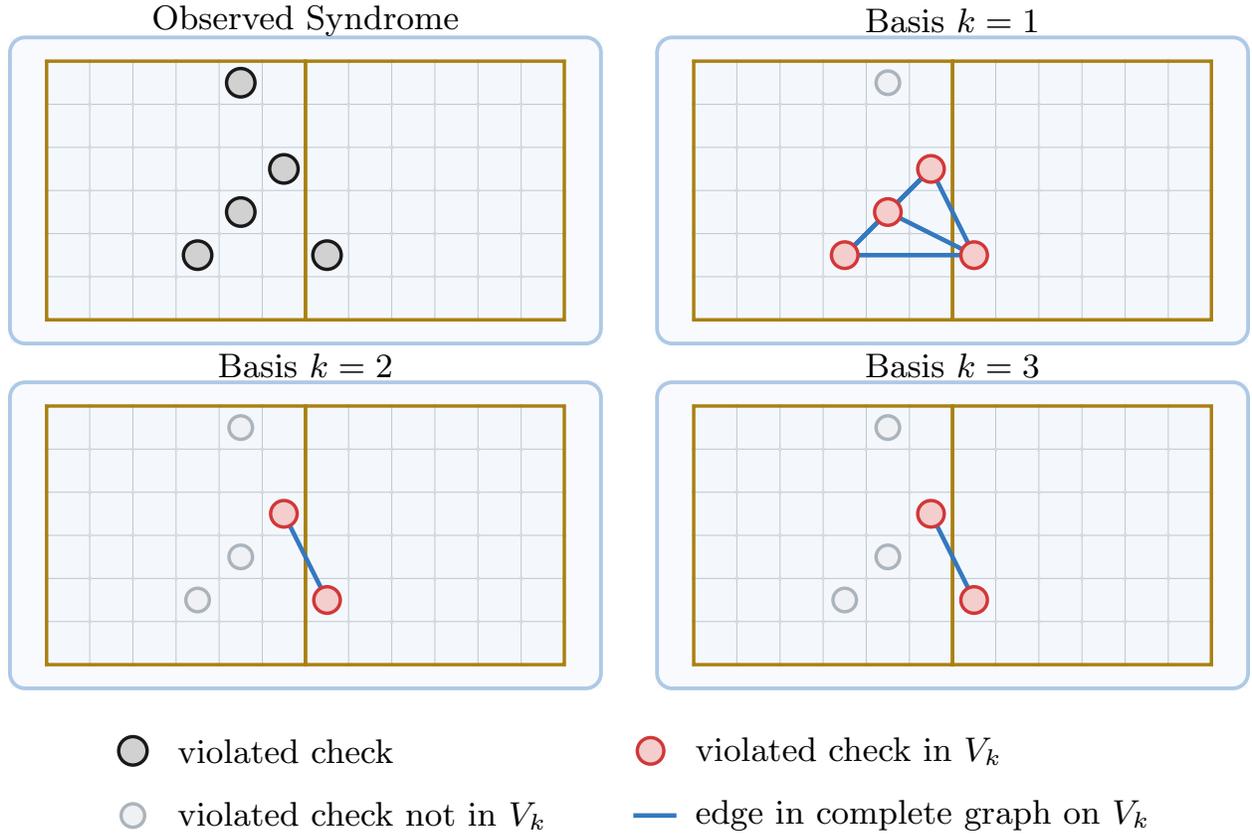

The objective of the matching procedure is for the pairs of vertices returned by the matching on each graph to lie in the same error cluster. The reason for this will relate to how the final correction is obtained from the perfect matching on $\mathcal{G}_i$, which has not yet been discussed as the method differs between the small and intermediate decoders. However, it should appear at least reasonable as an objective since it reduces to the standard one on the TC. Accepting this as our objective, the idea is then for the weight of the edge between $c_1,c_2$ to estimate the probability that $c_1 \sim_e c_2$. A proxy for this probability is the minimum possible weight of the error cluster associated to $c_1,c_2$. However, this metric is itself not easily computable in general. In the TC, it is computable by determining the minimum-length path of errors connecting two endpoints. The approach we employ to estimate the minimum possible weight of the error cluster associated to $c_1,c_2$, uses the following generalization of the familiar TC error path to the case where errors can activate more than two checks, 
\begin{definition}
 A \emph{path} $P$ of length $N$ between two checks $c_1,c_2 \in s_{obs}$ is defined as an ordered sequence of single qubit errors $P=(e_1,\dots,e_N)$ such that for each $j \in \{2,\dots, N\}$, we have that 
\begin{equation}
    \text{supp}\left(He_{j-1}\right) \:\bigcap\:  \text{supp} (He_j) \neq \emptyset,
    \label{eq:path_condition_1}
\end{equation}
and 
\begin{equation}
\begin{gathered}
    c_1,c_2 \in \supp\left(H\left(\sum_{i=1}^{N} e_i \right)\right),\\
    c_1 \in \supp(He_1),\qquad c_2 \in \supp(He_N).
\end{gathered}
    \label{eq:path_condition_2}
\end{equation}
\end{definition}
\begin{lemma}
    If the shortest path $P$ between activated checks $c_1,c_2$ is of length $N$ and $c_1,c_2$ are error-cluster-equivalent, then $|\operatorname{supp}(e')| \geq N$, where $e'$ is any such error cluster containing $c_1,c_2$.
\end{lemma}
\begin{proof}
If $N=1$ the claim is trivial, and if $N=2$, then a cluster of size 1 implies that there exists an error which has $c_1,c_2$ in its support, which will also define a path, so the claim holds. Assume that the error cluster corresponding to $c_1,c_2$ has size $N'<N$. For the above reason we consider $N' \geq 2$. Writing $e'=\sum_{i=1}^{N'} e_i$ where the $e_i$ are single qubit errors, we claim that for any two errors $e_{j_1},e_{j_2}$ in $\{e_1,\dots,e_{N'}\}$, there exists a sequence of errors from this subset such that $e_{j_1},e_{j_2}$ are the endpoints and the sequence satisfies Eq. \ref{eq:path_condition_1}. Assume otherwise. Since the condition just stated is transitive, this implies we can partition $\{e_1,\dots,e_{N'}\}$ into two nonempty subsets $A,B$ such that 

\begin{equation}
    \text{supp}(Ha) \cap \text{supp}(Hb)=\emptyset \quad \text{ for all } a \in A, b \in B,
\end{equation}
and thus 
\begin{equation}
\text{supp}\left(H\left(\sum_{a \in A} a\right)\right) \cap \text{supp}\left(H\left(\sum_{b \in B}\right) b\right)=\emptyset.
\end{equation}

Since both $\text{supp}(H\left(\sum_{b \in B} b\right) )$ and $\text{supp}(H\left(\sum_{a \in A} a\right))$ are nonempty by the error cluster definition, we have that the disjointness condition implies that the checks in $\text{supp}(H\left(\sum_{a \in A} a\right))$  are not error-cluster-equivalent to the checks in $\text{supp}(H\left(\sum_{b \in B}\right) b)$, which is a contradiction.  Thus, we can consider the shortest chain of errors which have pairwise overlap $(\hat{e}_1,\dots,\hat{e}_k)$, where $k \leq N'$ such that $c_1 \in \text{supp}(H\hat{e_1})$ and $c_2 \in \text{supp}(H\hat{e_k})$.  Assume that $c_1 \notin \text{supp}(H\left( \sum_{i=1}^k \hat{e}_i \right))$. Then it must appear in the support of the syndrome of at least one other $\hat{e}_j \in (\hat{e_2},\dots,\hat{e_k})$. But then $(\hat{e_j},\dots,\hat{e_k})$ is a strictly shorter chain of errors which have pairwise overlap and have $c_1$ supported in the first error and $c_2$ supported in the last error, which is a contradiction of our initial choice being minimal. The same logic shows that $c_2 \in \text{supp}(H\left( \sum_{i=1}^k \hat{e}_i \right))$. This proposed shortest chain of errors satisfies both conditions of a path and has length smaller than $N$, which is a contradiction, completing the proof. 

\end{proof}
While the shortest path between two activated checks lower bounds the weight of the associated error cluster when the checks are error-cluster-equivalent, this lower bound can be a significant underestimate. For a path $P=(e_1,\dots,e_N)$ to define a valid error cluster by taking a sum of the $e_i$, a necessary condition is \begin{equation} \text{supp}\left(H\left(\sum_{i=1}^{N} e_i \right) \right) \subseteq \text{supp}(s_{obs}).
\end{equation}
This condition is not in general satisfied by a given path. In particular, defining 
\begin{equation}
    \text{supp}(s_{remain})=\text{supp}\left(H\left(\sum_{i=1}^{N} e_i \right)\right)\setminus \text{supp}(s_{obs}),
\end{equation} 
an extension of $P$ to a potential error cluster would require at least $|P|+ \left \lceil \frac{|s_{remain}|}{w} \right \rceil$ qubits, where $w$ is the maximum number of stabilizers that can be activated by a single qubit $X$ error. The quantity we would want to minimize across paths is then $|P|+ \left \lceil \frac{|s_{remain}|}{w} \right \rceil$ rather than $|P|$. However, this quantity can only be evaluated after the full path is constructed, making it impractical for efficient shortest-path algorithms. Instead, we compute a quantity which takes $|P|$ and $|s_{\text{remain}}|$ into account while each move along the path has fixed cost independent of previous or future steps. To do this, we need to consider the set $M$ of polynomial indices of checks which can be activated by single qubit errors which also activates the check indexed at 1. If we consider an arbitrary check $c$, using translation invariance all checks for which there exists a length-one path with one endpoint at $c$ can be written as $cm$ for $m \in M$. 

Thus, we can think of the set $M$ as enumerating the ``moves" one can chain together to construct a path between two checks. In the BB codes under consideration, there is at most one length one path between two specified checks, so each move between two checks uniquely specifies an error and thus a third activated check $c_{other}$. In light of this discussion, we present an algorithm to determine a particular choice of weight for each edge in $E_i$ which accounts for the fact that we seek paths with smaller $|s_{\text{remain}}|$.
\begin{breakablealgorithm}
\caption{Weight determination for edges in $E_i$}
\label{alg:weight_determination_Ei_plain}
\begin{algorithmic}[1]
  \Require{Observed syndrome $s_{\mathrm{obs}}$; check lattice $\{x^iy^j | 0 \leq i < \ell, 0 \leq j < m \} $; fixed basis index $i$; vertex set
  $V_i=\{c\in \text{supp}(s_{\mathrm{obs}}):(u_c)_i=1\}$; background set
  $B_i=\{c\in \text{supp}(s_{\mathrm{obs}}):(u_c)_i=0\}$; move set $M$, weighting parameter $\lambda$.}
  \Ensure{A weight $w_i(c_a,c_b)$ for every edge $\{c_a,c_b\}\in E_i$ of the complete graph on $V_i$.}

  \ForAll{checks $c$ and moves $(c\to cm)$ for $m \in M$}
    \State Let $c_{\mathrm{other}}$ be the other check associated to that move.
    \If{$c_{\mathrm{other}}\in B_i$}
      \State $\mathrm{penalty}(c,cm)\gets 0$
    \Else
      \State $\mathrm{penalty}(c,cm)\gets 1$
    \EndIf
    \State Set move cost $\mathrm{cost}(c,cm)\gets 1+\lambda\cdot \mathrm{penalty}(c,cm)$.
  \EndFor

  \ForAll{start points $u\in V_i$}
    \State Initialize unresolved targets $T_u \gets V_i\setminus\{u\}$.
    \State Run a shortest-path search from $u$ over the check lattice using $\mathrm{cost}(\cdot,\cdot)$ with Dijkstra's algorithm (general nonnegative weights) or Dial's bucketed variant (nonnegative integer weights).
    \State Update the best known route to each reached check.
    \State If two routes to a given check have the same total cost, keep the one with the shorter path length.
    \State Whenever a vertex $v\in V_i$ is settled with final best value, remove $v$ from $T_u$.
    \If{$T_u=\emptyset$}
      \State Stop the shortest path search.
    \EndIf
  \EndFor

  \ForAll{unordered pairs $\{c_a,c_b\}\subset V_i$}
    \State Set $w_i(c_a,c_b)$ to the best route value from $c_a$ to $c_b$.
  \EndFor

  \State \Return $w_i$.
\end{algorithmic}
\end{breakablealgorithm}

For a fixed nonnegative integer $\lambda$, all move costs are nonnegative integers bounded independently of lattice size. Dial's algorithm then performs the shortest-path search in $O(\ell m)$ time per start point. We also note that the presented algorithm associates each edge on $\mathcal{G}_i$ with a path. While the weight of the edge as determined here will not exactly lower bound the size of the error cluster  containing the endpoints of a given edge, it does account for both $|P|$ and $|s_{remain}|$ and as such serves as a reasonable computable heuristic for this lower bound. The tunable parameter $\lambda$ allows one to control the relative importance of minimizing  $|P|$ and minimizing $|s_{remain}|$. A visualization of the move set $M$, a shortest path between two checks, and a minimal weight path between two checks for $\lambda=1$ is given in Fig. ~\ref{fig:move_set_paths}.

\begin{figure*}[t]
\centering
\centering

\definecolor{siteblue}{RGB}{52,120,190}
\definecolor{stencilred}{RGB}{210,55,55}
\definecolor{coarsegold}{RGB}{200,150,20}
\definecolor{gridgray}{RGB}{190,200,210}
\definecolor{panelbg}{RGB}{248,250,253}
\definecolor{cellbg}{RGB}{235,245,255}
\definecolor{blackdot}{RGB}{25,25,25}
\definecolor{stencilpurple}{RGB}{130,60,180}

\definecolor{msa_basisorange}{RGB}{200,150,20}
\definecolor{msa_latticegray}{RGB}{190,200,210}
\definecolor{msa_endpoint}{RGB}{210,55,55}
\definecolor{msa_active}{RGB}{200,150,20}
\definecolor{msa_remain}{RGB}{52,120,190}
\definecolor{msa_cancel}{RGB}{130,60,180}
\definecolor{msa_path}{RGB}{185,30,30}
\definecolor{msa_origin}{RGB}{25,25,25}
\definecolor{msa_helper}{RGB}{130,140,155}

\tikzset{
  msaendpoint/.style={circle, draw=msa_endpoint, fill=msa_endpoint!25, line width=0.9pt, minimum size=15.0pt, inner sep=0pt},
  msaactive/.style={circle, draw=coarsegold!85!black, fill=coarsegold!25, line width=0.9pt, minimum size=15.0pt, inner sep=0pt},
  msaremain/.style={circle, draw=siteblue, fill=siteblue!22, line width=0.9pt, minimum size=15.0pt, inner sep=0pt},
  msaorigin/.style={circle, draw=blackdot, fill=blackdot!20, line width=0.9pt, minimum size=15.0pt, inner sep=0pt},
}

\newcommand{\MSAMapX}[1]{2*((#1)-3)}
\newcommand{\MSAMapY}[1]{2*((#1)-3)}
\newcommand{\MSAPt}[2]{({\MSAMapX{#1}},{\MSAMapY{#2}})}

\newcommand{\MSALattice}{%
  \fill[panelbg] (-1.85,-1.55) rectangle (13.85,14.55);
  \draw[siteblue!40, line width=1.0pt, rounded corners=4pt] (-1.85,-1.55) rectangle (13.85,14.55);


  \fill[cellbg, opacity=0.30] (-1,-1) rectangle (13,13);

  \foreach \i in {-1,1,...,13} {
    \draw[gridgray, very thin] (\i,-1) -- (\i,13);
  }
  \foreach \j in {-1,1,...,13} {
    \draw[gridgray, very thin] (-1,\j) -- (13,\j);
  }

  \foreach \i in {-1,1,...,13} {
    \foreach \j in {-1,1,...,13} {
      \fill[gridgray!65] (\i,\j) circle (0.55pt);
    }
  }

  \draw[coarsegold!85!black, line width=0.9pt] (-1,-1) rectangle (13,13);
}

\newcommand{\MSAEnd}[2]{\node[msaendpoint] at \MSAPt{#1}{#2} {};}
\newcommand{\MSAActive}[2]{\node[msaactive] at \MSAPt{#1}{#2} {};}
\newcommand{\MSARemain}[2]{\node[msaremain] at \MSAPt{#1}{#2} {};}
\newcommand{\MSACancel}[2]{%
  \draw[msa_cancel, line width=2.7pt]
    ({\MSAMapX{#1}-0.55},{\MSAMapY{#2}-0.55}) -- ({\MSAMapX{#1}+0.55},{\MSAMapY{#2}+0.55});
  \draw[msa_cancel, line width=2.7pt]
    ({\MSAMapX{#1}+0.55},{\MSAMapY{#2}-0.55}) -- ({\MSAMapX{#1}-0.55},{\MSAMapY{#2}+0.55});
}
\newcommand{\MSAPath}[4]{%
  \draw[msa_path, line width=2.2pt, -{Stealth[length=14pt,width=12pt]}]
    \MSAPt{#1}{#2} -- \MSAPt{#3}{#4};
}

\newcommand{\MSAHelper}[6]{%
  \draw[gridgray!55!black, line width=1.4pt, densely dashed] \MSAPt{#1}{#2} -- \MSAPt{#3}{#4};
  \draw[gridgray!55!black, line width=1.4pt, densely dashed] \MSAPt{#3}{#4} -- \MSAPt{#5}{#6};
}

\newcommand{\MSALEnd}[2]{\node[msaendpoint] at (#1,#2) {};}
\newcommand{\MSALActive}[2]{\node[msaactive] at (#1,#2) {};}
\newcommand{\MSALRemain}[2]{\node[msaremain] at (#1,#2) {};}
\newcommand{\MSALCancel}[2]{%
  \draw[msa_cancel, line width=2.7pt] ({#1-0.52},{#2-0.52}) -- ({#1+0.52},{#2+0.52});
  \draw[msa_cancel, line width=2.7pt] ({#1+0.52},{#2-0.52}) -- ({#1-0.52},{#2+0.52});
}

\resizebox{0.65\linewidth}{!}{%
\begin{tikzpicture}[x=0.46cm,y=0.46cm,font=\small]

\begin{scope}[local bounding box=panelblock]

\begin{scope}[shift={(8.5,16.2)}]
  \MSALattice
  \node[font=\Large] at (6,13.6) {{A.\ Move Set $M$}};

  \foreach \dx/\dy in {
    0/-1, 0/1, 1/0, -1/0,
    3/-1, 3/-2, 2/-3, 1/-3,
    -3/1, -3/2, -2/3, -1/3} {
    \pgfmathsetmacro{\mx}{6+\dx}
    \pgfmathsetmacro{\my}{6+\dy}
    \draw[-{Stealth[length=8pt,width=7pt]}, msa_path, line width=1.9pt]
      \MSAPt{6}{6} -- \MSAPt{\mx}{\my};
  }
  \node[msaorigin] at \MSAPt{6}{6} {};
\end{scope}

\begin{scope}[shift={(0,0)}]
  \MSALattice
  \node[font=\Large] at (6,13.6) {{B.\ Shortest-Path Criterion}};

  \MSAPath{5}{5}{8}{3}
  \MSAPath{8}{3}{7}{6}
  \MSAHelper{5}{5}{5}{4}{8}{3}
  \MSAHelper{8}{3}{9}{3}{7}{6}

  \node[font=\LARGE\bfseries, text=msa_path] at \MSAPt{6.20}{4.95} {1};
  \node[font=\LARGE\bfseries, text=msa_path] at \MSAPt{8.15}{5.20} {2};

  \MSAActive{4}{8}
  \MSAActive{6}{4}
  \MSAActive{8}{3}

  \MSARemain{5}{4}
  \MSARemain{9}{3}

  \MSAEnd{5}{5}
  \MSAEnd{7}{6}
\end{scope}

\begin{scope}[shift={(17,0)}]
  \MSALattice
  \node[font=\Large] at (6,13.6) {{C.\ $\lambda$-Weighted Criterion}};

  \MSAPath{5}{5}{6}{5}
  \MSAPath{6}{5}{9}{3}
  \MSAPath{9}{3}{7}{6}

  \MSAHelper{5}{5}{4}{8}{6}{5} 
  \MSAHelper{6}{5}{6}{4}{9}{3} 
  \MSAHelper{9}{3}{8}{3}{7}{6} 

  \node[font=\LARGE\bfseries, text=msa_path] at \MSAPt{5.10}{4.20} {1};
  \node[font=\LARGE\bfseries, text=msa_path] at \MSAPt{7.10}{3.10} {2};
  \node[font=\LARGE\bfseries, text=msa_path] at \MSAPt{8.80}{5.20} {3};

  \MSACancel{4}{8}
  \MSACancel{6}{4}
  \MSACancel{8}{3}

  \MSAEnd{5}{5}
  \MSAEnd{7}{6}
\end{scope}
\end{scope}

\begin{scope}[overlay]
  \MSALEnd{-2.5}{22.9}
  \node[anchor=west, font=\Large] at (-1.5,22.9) {path endpoints};

  \MSALCancel{-2.5}{17.9}
  \node[anchor=west, font=\Large] at (-1.5,17.9) {cancellation};

  \MSALActive{24.8}{22.9}
  \node[anchor=west, font=\Large] at (26.2,22.9) {violated checks};

  \MSALRemain{24.8}{20.4}
  \node[anchor=west, align=left, font=\Large] at (26.2,20.4) {remaining checks\\from path};

  \draw[msa_path, line width=2.6pt, -{Stealth[length=14pt,width=12pt]}] (24.0,17.9) -- (26.8,17.9);
  \node[anchor=west, font=\Large] at (27.3,17.9) {move from $M$};
\end{scope}

\end{tikzpicture}%
}
\caption{Illustrations of the move set $M$ used to construct paths on the gross code (A), an example of a shortest path between two endpoints (B), an example of a minimal path according to the $\lambda=1$ criterion (C). Endpoints of arrows in (A) represent checks which can be violated by errors which also activate the check at the origin. The monomial representations of these checks collectively form the move set $M$. The elements of this set can be thought of as the building blocks of a path. Triangles in (B) and (C) represent single qubit errors which form the path between the two endpoints. Note that the path in panel $B$ is shorter than the path in panel $C$, but leaves a greater amount of remaining syndrome.  }
\label{fig:move_set_paths}
\end{figure*}

\subsubsection{Re-matching and syndrome shifts}
Consider an error cluster which contains multiple pairs of activated checks on the same matching graph $\mathcal{G}_i$. The desirable outcome is for the matching procedure to pair up these checks internally within the cluster. The issue is that the edge weights are pairwise path costs: each matched pair is scored independently. As a result, when several internal pairs are chosen, the corresponding paths can overlap substantially, and the same physical error support is effectively counted multiple times.
So even if each individual edge weight is close to a lower bound on cluster size, the sum used by MWPM can overestimate the true joint cost of explaining all checks in that cluster. The optimizations discussed in this section are aimed to overcome this issue. In what follows we assume that MWPM has been run on the graphs $\mathcal{G}_i$. 
\begin{definition}
For each basis index $i$, let $E_{\mathrm{MWPM},i}$ be the edges returned by MWPM on $\mathcal{G}_i$, and then define
$E_{\mathrm{MWPM}}=\bigcup_{i=1}^r E_{\mathrm{MWPM},i}.$
\end{definition}
For each edge $(c_a,c_b) \in E_{\mathrm{MWPM}}$, we want to find a candidate error cluster which contains these endpoints. This is accomplished by a depth-limited depth-first-search (DFS) over candidate error cluster supports up to some fixed weight $w_{\text{fixed}}$, where $w_{\text{fixed}}$ is a tunable parameter of the decoding procedure. In the BB codes we consider, there are 6 qubits which touch each check. A node at depth $t$ of the DFS is an error pattern $e$ such that $|e|=t$ and $c_a \in \text{supp}(He)$. A node is expanded by adding a single qubit error which activates a check in $\text{supp}(H(e))\setminus \text{supp}(s_{obs})$. The DFS terminates when $c_a,c_b \in \text{supp}(H(e))$ and $\text{supp}(H(e)) \subseteq \text{supp}(s_{obs}).$ or when $t$ reaches a fixed weight $w_{\text{fixed}}$. For fixed $w_{\mathrm{fixed}}$ and bounded check and qubit degree, this procedure costs $O(|E_{MWPM}|)$, with a constant depending on the depth and the branching rule. This procedure then produces a set of subsets of checks contained in $\text{supp}(s_{obs})$, each corresponding to a potential error cluster containing two checks which were matched on some graph $\mathcal{G}_i$.  For each subset $S$, we can then compute
\begin{equation}
    F_S=\frac{|\{(c_a,c_b) \in E_{MWPM} | c_a, c_b \in S \}|}{|\{(c_a,c_b) \in E_{MWPM} | c_a \text{ or }c_b \in S \}|}
\end{equation}
if $F_S$ is close to one, then the matching across all the subsets are in agreement that $S$ is likely a true error cluster, while if it is close to zero, only a few matchings are consistent with this being a true error cluster. The idea then is if $F_S>F_{limit}$, where $F_{limit}$ is some tunable parameter, we alter all edge weights on all matching graphs for edges between checks in $S$ by some negative tunable parameter $\delta_-$ and alter all edge weights on all matching graphs for edges between checks with exactly one endpoint in $S$ by some positive tunable parameter $\delta_+$, and then repeat the MWPM procedure on all graphs. This process can then be iterated some set number of times or until the matching outcomes stop changing. \\ 
\indent The next tool we can use to our advantage is that we are free to shift the observed syndrome by any monomial, obtain a correction for the shifted syndrome, and then shift this correction back. This is notable because the equivalence class decomposition of each syndrome depends on its exact location, so different shifts will produce different matching graphs and thus different matching outcomes.  The exact utilization of this shifting freedom will be different for the small and intermediate code size decoders. Note that for a given code the number of unique shifts which have to be considered is at most $b\times b$, where $b$ is the coarse-graining parameter of the code. To see why this is the case, fixing one $b\times b$ cell, note that any check outside the cell has its equivalence class identified with a check inside a cell via a string operator with point-like endpoints by definition of $b$. We now present pseudocode which synthesizes the procedures described in this section to run the matching procedure for a given syndrome shift. 

\begin{breakablealgorithm}
\caption{Matching procedure for a fixed shift}
\label{alg:matching_pipeline_fixed_shift}
\begin{algorithmic}[1]
\Require{Shifted syndrome $s^{(\tau)}$; basis data $\{u_i,\|u_i\|_T,s_i\}_{i=1}^r$; dictionary $c\mapsto u_c$; weight inputs $(M,\lambda)$; DFS depth $w_{\text{fixed}}$; threshold $F_{\mathrm{limit}}$; weight updates $(\delta_-,\delta_+)$; maximum re-matching rounds $R_{\mathrm{match}}$.}
\Ensure{Matched-edge set $E_{\mathrm{MWPM}}^{(\tau)}$, stored paths $P_i^{(\tau)}$, and candidate syndrome clusters $S$ and errors $e_S$ associated to each edge.}

\For{$i\gets 1$ \textbf{to} $r$}
  \State Build $V_i^{(\tau)}=\{c\in \mathrm{supp}(s^{(\tau)}):(u_c)_i=1\}$ and $B_i^{(\tau)}=\{c\in \mathrm{supp}(s^{(\tau)}):(u_c)_i=0\}$.
  \State Compute all edge weights on $\mathcal{G}_i^{(\tau)}$ using Algorithm~\ref{alg:weight_determination_Ei_plain}.
  \State Run MWPM on $\mathcal{G}_i^{(\tau)}$ to obtain $E_{\mathrm{MWPM},i}^{(\tau)}$ and store selected paths $P_i^{(\tau)}$.
\EndFor
\State Set $E_{\mathrm{MWPM}}^{(\tau)}\gets\bigcup_{i=1}^r E_{\mathrm{MWPM},i}^{(\tau)}$.
\For{$t\gets 1$ \textbf{to} $R_{\mathrm{match}}$}
  \State For each $(c_a,c_b)\in E_{\mathrm{MWPM}}^{(\tau)}$, run depth-limited DFS to depth $w_{\text{fixed}}$ to obtain a candidate error cluster $e_S$ containing this edge and corresponding subset $S\subseteq \mathrm{supp}(s^{(\tau)})$.
  \State If a candidate error cluster is found for this edge,  compute $F_S=\frac{\#\{\text{matched edges internal to }S\}}{\#\{\text{matched edges touching }S\}}$.
  \State For all $S$ with $F_S>F_{\mathrm{limit}}$, add $\delta_-$ to edges internal to $S$ and add $\delta_+$ to edges with exactly one endpoint in $S$.
  \State Re-run MWPM on all $\mathcal{G}_i^{(\tau)}$ with updated weights and update $E_{\mathrm{MWPM}}^{(\tau)}$ and $P_i^{(\tau)}$.
\EndFor
\State Compute candidate syndrome/error clusters associated to each edge in the final matching along with $ F_S$ for each subset $S$.
\end{algorithmic}
\end{breakablealgorithm}

\subsection{Intermediate code size machinery}
\begin{definition}
Let $V_C=\{c \mid c\in \text{supp}(s_{\mathrm{obs}})\}$ 
For each $(c_a,c_d)\in E_{\mathrm{MWPM},i}$, let $P_i(c_a,c_d)$ be the path found using the edge weight finding algorithm in the previous section corresponding to the edge $(c_a,c_d)$ on $\mathcal{G}_i$. 
Define
\begin{equation}
  \begin{aligned}
  O_i(c_a,c_d)=\{c\in V_C \mid\;& c \text{ appears as } c_{\mathrm{other}}\\
  &\text{along } P_i(c_a,c_d),\ c_{\mathrm{other}}\in B_i\}.
  \end{aligned}
\end{equation}
This can be understood as the set of originally violated checks other than the endpoints which cancel with the syndrome produced by $P_i(c_a,c_d)$. Now define
\begin{align}
E_i
={}&\{(c_a,c_b)\mid \exists c_d,\ (c_a,c_d)\in E_{\mathrm{MWPM},i},\nonumber\\
&\hspace{7.5em} c_b\in O_i(c_a,c_d)\}\nonumber\\
&\cup
\{(c_a,c_b)\mid \exists c_d,\ (c_b,c_d)\in E_{\mathrm{MWPM},i},\nonumber\\
&\hspace{7.5em} c_a\in O_i(c_b,c_d)\}.
\end{align}
The $E_i$ can be understood as an additional edge set between the starting and ending points of selected paths from $\mathcal{G}_i$ and the additional canceled checks along those paths. Finally set
\begin{equation}
  E_C=E_{\mathrm{MWPM}}\cup\bigcup_{i=1}^r E_{i},
\end{equation}
and define the compatibility graph as $
G_C=(V_C,E_C).$
\end{definition}
This graph can be thought of as follows: if every matched edge pairs checks within the same error cluster, and if every “background” syndrome check used in weighting that edge’s path also lies in that same cluster, then every connected component of the graph is contained in one true error cluster. Equality with the true clusters additionally requires these edges to connect all syndrome checks within each cluster.  As we expect that for reasonable physical error rates the true error clusters will not be very large, a compatibility graph produced from MWPM on each $\mathcal{G}_i$ whose largest connected component contains many checks is likely to have resulted from matchings which are not internal to true error clusters.\\
\indent Consequently, we determine the final applied correction for the shifted syndrome for which the matching procedure minimized the largest connected component of the compatibility graph across all shifts. A further optimization which can be employed is to feed the observed syndrome restricted only to the largest connected component of the compatibility graph in the best shift back into the matching pipeline. This refining procedure disentangles the part of the syndrome that we expect the matching algorithm struggled with from the component of the syndrome we believe matching handled correctly. As such, this procedure can break up this connected component of the compatibility graph into smaller clusters, which we expect to produce a more accurate matching. \\ 
\indent Once all matchings have been finalized for each activated check in the observed syndrome, the error $e_c$ assured to exist from Eq. \ref{eq:rewrite} is applied to rewrite each single syndrome in terms of its associated basis patterns. Note that $e_c$ is not unique and there are in principle multiple ways of choosing such an error. In the current implementation, BP-OSD is used once prior to observing any syndrome and the resulting $e_c$ is stored in a dictionary for all subsequent decoding runs. The chosen solution is part of the finite-size decoder: two solutions of Eq.~\eqref{eq:rewrite} differ by an element of $\ker H$ and may differ by a logical operator. Geometric locality of the prescribed syndrome alone does not rule out this ambiguity.\\ \indent To furnish a final correction, one must then resolve all the syndrome corresponding to the active basis elements for each observed check. To accomplish this, for each matched pair on each $\mathcal{G}_i$, if the corresponding basis pattern occupies the same $||u_i||_T \times ||u_i||_T $ unit cell, then nothing needs to be done as the rewrite operator $e_{c_a},e_{c_b}$ associated to each check in the matched pair ensures that these basis patterns cancel.  If they do not occupy the same $||u_i||_T \times ||u_i||_T $ unit cell, then a chain of short strings for this basis element is applied between these two cells to annihilate both patterns, and the winding directionality of this chain is determined by the winding directionality of the path $P_i(c_a,c_b)$ obtained from Dial's algorithm during the matching stage. An example of using a chain of short strings to annihilate a pair of weight 3 basis syndrome patterns in different cells is shown in Fig.  \ref{fig:short_string_chain_full_lattice}. The reason why we obtain a final correction using short strings rather than the candidate error clusters obtained from the DFS is that for the intermediate code size, there may be a large number of candidate error clusters, and in general choosing one incorrectly will result in a failed correction. The short string method negates the requirement to obtain a covering set of candidate error clusters.  \\
The short-string equations and Eq.~\eqref{eq:rewrite} ensure exact syndrome cancellation: summing the rewrite operators produces $s_{\mathrm{obs}}$ plus all inserted basis patterns, while each matched short-string chain produces its two endpoint patterns. Hence $H\widehat e=s_{\mathrm{obs}}$. This establishes return to the code space; the logical success probability of these finite-size choices is assessed numerically.

\indent A full algorithmic summary of the intermediate code size decoding process is given below.  

\begin{breakablealgorithm}
\caption{Intermediate-size matching decoder}
\label{alg:intermediate_decoder_summary}
\begin{algorithmic}[1]
\Require{Observed syndrome $s_{\mathrm{obs}}$; basis data $\{u_i,\|u_i\|_T,s_i\}_{i=1}^r$; dictionaries $c\mapsto u_c$ and $c\mapsto e_c$; shift set $\mathcal{T}=\{x^a y^d:0\le a<\text{min}\{L_x,b\},0\le d<\text{min}\{L_y,b\}\}$; weight inputs $(M,\lambda)$; DFS depth $w_{\text{fixed}}$; threshold $F_{\mathrm{limit}}$; weight updates $(\delta_-,\delta_+)$; maximum re-matching rounds $R_{\mathrm{match}}$; maximum refinement rounds $R_{\mathrm{refine}}$.}
\Ensure{Correction $\widehat e$ for the unshifted syndrome.}

\ForAll{shifts $\tau\in\mathcal{T}$}  \Comment{Beginning of matching pipeline}
  \State Set $s^{(\tau)}\gets \tau s_{\mathrm{obs}}$.
  \State Compute matched-edge set $E_{\mathrm{MWPM}}^{(\tau)}$ and stored paths $P_i^{(\tau)}$ using Algorithm \ref{alg:matching_pipeline_fixed_shift}.
  \State Build $G_C^{(\tau)}$ from $E_{\mathrm{MWPM}}^{(\tau)}$ and from checks that appear as $c_{\mathrm{other}}$ along stored paths.
  \State Record score $m_{\max}^{(\tau)}\gets$ largest connected-component size of $G_C^{(\tau)}$. 
\EndFor

\State Choose $\tau^\star\in\arg\min_{\tau\in\mathcal{T}} m_{\max}^{(\tau)}$. \Comment{End of matching pipeline}
\For{$t\gets 1$ \textbf{to} $R_{\mathrm{refine}}$}
  \State Restrict to the largest connected component(s) of $G_C^{(\tau^\star)}$, rerun the matching pipeline on the syndrome restricted to this component, and update matchings on this component.
  \State Stop if the largest-component score does not decrease.
\EndFor

\State Initialize shifted correction $\widehat e^{(\tau^\star)}\gets 0$. 
\ForAll{checks $c\in \mathrm{supp}(s^{(\tau^\star)})$}
  \State Add $e_c$ to total shifted correction $\hat{e}^{(\tau^\star)}$,  rewriting $c$ into translated canonical basis representatives.
\EndFor
\For{$i\gets 1$ \textbf{to} $r$}
  \ForAll{$(c_a,c_b)\in E_{\mathrm{MWPM},i}^{(\tau^\star)}$}
    \If{the corresponding basis patterns are in different $\|u_i\|_T\times\|u_i\|_T$ cells}
      \State Add a chain of short strings for basis $i$ between these cells, with winding direction taken from $P_i^{(\tau^\star)}(c_a,c_b)$ to the total shifted correction $\hat{e}^{(\tau^\star)}$.
    \EndIf
  \EndFor
\EndFor

\State Return $\widehat e\gets (\tau^\star)^{-1}\widehat e^{(\tau^\star)}$.
\end{algorithmic}
\end{breakablealgorithm}
\begin{figure*}[t]
\centering
\definecolor{siteblue}{RGB}{52,120,190}
\definecolor{stencilred}{RGB}{210,55,55}
\definecolor{coarsegold}{RGB}{200,150,20}
\definecolor{gridgray}{RGB}{190,200,210}
\definecolor{panelbg}{RGB}{248,250,253}
\definecolor{cellbg}{RGB}{235,245,255}
\definecolor{redoutline}{RGB}{250,0,0}

\tikzset{
  syndromeq/.style={circle, draw=stencilred, fill=stencilred!30, line width=1.0pt, minimum size=5.6pt, inner sep=0pt},
  qerror/.style={siteblue, line width=1.2pt},
}

\newcommand{\CTqH}[2]{%
  \draw[qerror] ({#1-0.18},{#2+0.5-0.18}) -- ({#1+0.18},{#2+0.5+0.18});%
  \draw[qerror] ({#1-0.18},{#2+0.5+0.18}) -- ({#1+0.18},{#2+0.5-0.18});%
}
\newcommand{\CTqV}[2]{%
  \draw[qerror] ({#1-0.5-0.18},{#2-0.18}) -- ({#1-0.5+0.18},{#2+0.18});%
  \draw[qerror] ({#1-0.5-0.18},{#2+0.18}) -- ({#1-0.5+0.18},{#2-0.18});%
}
\newcommand{\CTqHV}[2]{\CTqH{#1}{#2}\CTqV{#1}{#2}}
\newcommand{\CTEnd}[2]{\node[syndromeq] at (#1,#2) {};}

\newcommand{\CTHstep}[2]{%
  \draw[redoutline!55, dashed, line width=1.20pt] ({#1-0.5},{#2-0.5}) rectangle ({#1+5.5},{#2+2.5});
  \CTqH{#1+0}{#2+0}
  \CTqH{#1+1}{#2+0}
  \CTqHV{#1+0}{#2+1}
  \CTqHV{#1+1}{#2+1}
}
\newcommand{\CTVstep}[2]{%
  \draw[redoutline!55, dashed, line width=1.20pt]
    ({#1-0.5},{#2-3.5}) rectangle ({#1+2.5},{#2+2.5});
  \CTqHV{#1}{#2}
  \CTqHV{#1}{#2+1}
  \CTqV{#1+1}{#2}
  \CTqV{#1+1}{#2+1}
}

\resizebox{0.40\linewidth}{!}{%
\begin{tikzpicture}[x=0.35cm,y=0.35cm,yscale=-1,font=\small]
  \path[use as bounding box] (-0.95,-0.95) rectangle (11.95,11.95);

  \fill[panelbg] (-0.90,-0.90) rectangle (11.90,11.90);
  \draw[siteblue!40, line width=1.0pt, rounded corners=4pt] (-0.90,-0.90) rectangle (11.90,11.90);

  \foreach \gc in {0,...,3} {
    \foreach \gr in {0,...,3} {
      \fill[cellbg, opacity=0.30] ({3*\gc-0.5},{3*\gr-0.5}) rectangle ({3*\gc+2.5},{3*\gr+2.5});
      \draw[siteblue!18, line width=0.35pt] ({3*\gc-0.5},{3*\gr-0.5}) rectangle ({3*\gc+2.5},{3*\gr+2.5});
    }
  }

  \foreach \i in {0,...,12} {
    \draw[gridgray, very thin] ({\i-0.5},-0.5) -- ({\i-0.5},11.5);
    \draw[gridgray, very thin] (-0.5,{\i-0.5}) -- (11.5,{\i-0.5});
  }
  
  \draw[coarsegold!85!black, line width=0.9pt] (5.5,-0.5) -- (5.5,11.5);
  \draw[coarsegold!85!black, line width=0.9pt] (-0.5,5.5) -- (11.5,5.5);
  \fill[cellbg, opacity=0.30] (-0.5,-0.5) rectangle (11.5,11.5);

  \foreach \i in {0,...,12} {
    \foreach \j in {0,...,12} {
      \fill[gridgray!65] ({\i-0.5},{\j-0.5}) circle (0.55pt);
    }
  }

  \draw[siteblue!55, line width=1.1pt] (-0.5,-0.5) rectangle (11.5,11.5);

  \CTHstep{0}{6}
  \CTHstep{3}{6}
  \CTHstep{6}{6}
  \CTVstep{9}{3}
  \CTVstep{9}{6}

  \CTEnd{0}{6}
  \CTEnd{1}{6}
  \CTEnd{1}{7}
  \CTEnd{9}{0}
  \CTEnd{10}{0}
  \CTEnd{10}{1}

  \begin{scope}[font=\scriptsize]
    \draw[qerror, line width=1.5pt] (0.2,12.58) -- (0.7,13.08);
    \draw[qerror, line width=1.5pt] (0.2,13.08) -- (0.7,12.58);
    \node[anchor=west] at (1.0,12.83) {$X$ errors};

    \node[syndromeq] at (6.05,12.83) {};
    \node[anchor=west] at (6.4,12.83) {violated checks};
    \path[use as bounding box] (-0.95,-0.95) rectangle (11.95,13.15);
  \end{scope}
\end{tikzpicture}%
}
\vspace{10 pt}
\caption{Short string chain of errors (blue crosses) for the gross code polynomial. Two identical weight-3 copies of basis-7 $(y^2+yx+xy^2)$ for the $24\times24$ gross code are produced as violated checks (red circles) in distinct $3 \times 3$ cells at the endpoint of this short string chain. Chains such as these are implemented based on matching outcomes on each basis graph to arrive at a total correction for the intermediate-size decoder.}
\label{fig:short_string_chain_full_lattice}
\end{figure*}
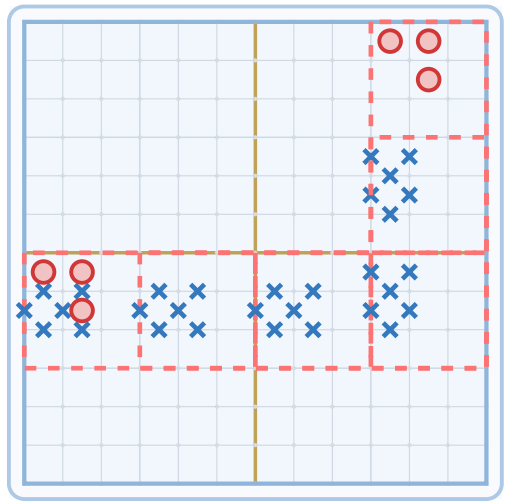

\subsection{Small code size machinery}
The key reason why the method discussed above to obtain a final correction from a matching does not work in the small code size regime, is that the maximum lattice distance between errors in the short strings for certain basis elements is comparable to the dimensions of the polynomial lattice. In other words, the short strings may span the full dimensions of the small code size qubit lattice, and therefore may have unintended, nontrivial logical action on the code when applied. Thus, instead of using the short strings to obtain a proposed correction, we use the subsets obtained by the depth-limited DFS. In particular, after running the matching procedure on all shifts of the observed syndrome, including the within-shift redecodes with the applied $\delta_+$ and $\delta_-$ modifiers, we collect all distinct subsets $S$ obtained from the DFS such that $F_S=1$. This condition means that all matchings are internal to that subset. After collecting these subsets, we order them by their multiplicity across shifts. In descending order, we then select disjoint subsets $S$ from this list. If the union of the subsets covers the observed syndrome and the combined weight of errors associated to each subset is less than some input value $e_{\text{limit}}$, the union of the errors associated to each subset is returned as a correction. This process is stage 1. \\ 
\indent Otherwise, if stage 1 does not produce a covering correction of weight below $e_{\text{limit}}$,  the observed syndrome in the complement of the union of the selected subsets is fed back into the decoding pipeline, producing a new collection of subsets of the observed syndrome and their associated errors. The same multiplicity selection logic is used to obtain a disjoint collection of these subsets. If combining the subsets from stage 2 with the subsets selected in stage 1 still does not produce a covering correction with weight below $e_{\text{limit}}$, we rerun the ``subset-redecode" procedure. This involves rerunning stage 2 with one stage 1 subset removed from the collection of accepted error cluster at a time and reintroduced into stage 2. In this stage 2 rerun, internal matching within the reintroduced subset is  forbidden, matching must pair these checks externally with the rest of the syndrome. If no covering correction which is below the weight limit is obtained from a union of the subsets for any single subset removal procedure as above, we attempt re-injecting all possible \emph{pairs} of selected subsets with the associated internal weight penalty into stage 2. \\
\indent The intuition behind this step is that if the matching process is unable to produce a consistent disjoint subset covering of the stage 2 inputs, it is likely that first stage  selected an error cluster which is not a true error cluster. As such, running stage 2 again with these ``bad" error clusters forbidden from being produced by the matching has the potential to improve the final output.  If no covering set of disjoint subsets is found after the subset removal procedure, we declare a decoder failure, and if no below weight limit error is found, the lowest weight error syndrome covering error found is applied. \\
\indent We note that what makes this procedure successful is exactly what makes the short string approach fail: due to the small code size, there is a small number of total clusters, so successfully determining all of them through this procedure is reasonable. This method would fail badly in the intermediate case because the number of error clusters is sufficiently large that the matching very rarely gets them all correct. However, this is not a problem because a successful correction can still in some cases be obtained via short string pairing even if a small fraction of matchings occur between rather than within true error clusters. \\  
\indent A full algorithmic summary of the small code size decoding process is given below.  
\begin{breakablealgorithm}
\caption{Small-size matching decoder}
\label{alg:small_decoder_summary}
\begin{algorithmic}[1]
\Require{Observed syndrome $s_{\mathrm{obs}}$; basis data $\{u_i,\|u_i\|_T,s_i\}_{i=1}^r$; dictionary $c\mapsto u_c$; shift set $\mathcal{T}=\{x^a y^d:0\le a<\text{min}\{L_x,b\},0\le d<\text{min}\{L_y,b\}\}$; weight inputs $(M,\lambda)$; DFS depth $w_{\text{fixed}}$; Threshold $F_{\text{limit}}$; weight updates $(\delta_-,\delta_+)$; maximum re-matching rounds $R_{\mathrm{match}}$; weight limit $e_{\mathrm{limit}}$.}
\Ensure{Correction $\widehat e$ for the unshifted syndrome.}

\ForAll{shifts $\tau\in\mathcal{T}$} \Comment{Beginning of stage 1 matching pipeline}
  \State Set $s^{(\tau)}\gets \tau s_{\mathrm{obs}}$.
    \State Compute matched-edge set $E_{\mathrm{MWPM}}^{(\tau)}$ and stored paths $P_i^{(\tau)}$ using Algorithm \ref{alg:matching_pipeline_fixed_shift}
  \State Store all subsets $S^{\tau}$ with $F_{S^{\tau}}=1$ found for shift $\tau$, with their associated candidate errors $e_S^{\tau}$.
\EndFor

\State Unshift each stored subset $S^{\tau}$ and associated error $e_S^{\tau}$ to obtain $S$ and $e_S$, grouping identical subsets/error across shifts.
\State Record multiplicity $\mu(S)$ for each distinct subset $S$.
\State Sort subsets by decreasing multiplicity and select a disjoint family $\mathcal{S}_1$. 
\State Let $\Omega_1=\bigcup_{S\in\mathcal{S}_1} S$ and $\widehat e_1=\bigoplus_{S\in\mathcal{S}_1} e_S$. 
\If{$\Omega_1=\mathrm{supp}(s_{\mathrm{obs}})$ and $|\widehat e_1|\le e_{\mathrm{limit}}$}
  \State \Return $\widehat e_1$.
\EndIf \Comment{End of stage 1 matching pipeline}

\State Define residual syndrome support $\Omega_{\mathrm{res}}=\mathrm{supp}(s_{\mathrm{obs}})\setminus \Omega_1$. \Comment{Beginning of stage 2 matching pipeline}
\State Re-run the full  stage 1 matching pipeline above on the residual syndrome only, producing a second disjoint family $\mathcal{S}_2$ of complete subsets.
\State Form $\widehat e_{12}=\left(\bigoplus_{S\in\mathcal{S}_1} e_S\right)\oplus\left(\bigoplus_{S\in\mathcal{S}_2} e_S\right)$.
\If{$\left(\bigcup_{S\in\mathcal{S}_1\cup\mathcal{S}_2} S\right)=\mathrm{supp}(s_{\mathrm{obs}})$ and $|\widehat e_{12}|\le e_{\mathrm{limit}}$}
  \State \Return $\widehat e_{12}$.
\EndIf \Comment{End of stage 2 matching pipeline}

\ForAll{single subsets $S\in\mathcal{S}_1$} \Comment{Beginning of subset-redecode procedure}
  \State Temporarily remove $S$ from $\mathcal{S}_1$. Re-inject $S$ into stage-2 input and rerun the  matching pipeline, while forbidding internal matching inside $S$ via a large penalty.
  \State Form the resulting candidate correction $\hat{e}_{12}$ resulting from the altered stage 1 output and stage 2 input, and accept it if it covers all of $\mathrm{supp}(s_{\mathrm{obs}})$ with weight at most $e_{\mathrm{limit}}$. 
\EndFor

\ForAll{pairs of subsets $\{S,S'\}\subset\mathcal{S}_1$}
  \State temporarily remove $S, S'$ from $\mathcal{S}_1$. Re-inject $S\cup S'$ into stage-2 input and rerun the matching pipeline with the same internal-match penalty on $S$ and $S'$.
  \State Form the resulting covering candidate correction $\hat{e}_{12}$ resulting from the altered stage 1 output and stage 2 input, and accept it if it covers all of $\mathrm{supp}(s_{\mathrm{obs}})$ with weight at most $e_{\mathrm{limit}}$. \Comment{End of subset-redecode procedure}
\EndFor

\State If no covering candidate has weight $\le e_{\mathrm{limit}}$, return the minimum-weight covering candidate found; if no covering candidate is found, declare decoder failure.
\end{algorithmic}
\end{breakablealgorithm}

\subsection{Numerical results for BB codes}
 The logical failure rate as a function of physical error rate when decoding the gross code under i.i.d. bit-flip noise using BP, BP-OSD, and the presented matching based decoder is shown in Fig.~\ref{fig:d12decodercomparison}. The same plot for the two-gross code and $24 \times 24$ gross code is illustrated in Fig.~\ref{fig:matching_bp_comparison}. The gross and two-gross code use the small-size matching decoder approach while the $24 \times 24$ gross code uses the intermediate-size approach. Logical failure rate is estimated using a Monte Carlo simulation. For a data point with logical error rate of order $10^{-k}$,  $10^{k+2}$ trials were run.  The parameters used for the intermediate matching decoder are $\lambda=1,R_{\text{match}}=5,R_{\text{refine}}=2,w_{\text{fixed}}=6,F_{\text{limit}}=0.5, $ and $(\delta_-,\delta_+)=(-1,1)$, and the parameters used for the small matching decoder are $\lambda=1,R_{\text{match}}=3,w_{\text{fixed}}=6,F_{\text{limit}}=0.5, e_{\text{limit}}=10$ and $(\delta_-,\delta_+)=(-1,1)$. $R_{\text{match}}$, $R_{\text{refine}}$ and $w_{\text{fixed}}$ were chosen so that further increase did not noticeably change results on the codes considered, and the remaining parameters were chosen by iterated estimation and performance testing. Note that these choices are not necessarily optimal since total parameter space is large. BP is run using min-sum with 1000 maximum iterations, and the OSD step is order 10. \\ 
 \indent We estimate the pseudothreshold achieved for each decoder on each code by considering the two data points where the maximum (or minimum) value in the uncertainty window are on different sides of the $p=p_L$ curve, and linearly interpolate between the maximum (or minimum) values in the uncertainty window. The average crossing location with $p=p_L$ between the minimum and maximum case is taken as the pseudothreshold and the standard deviation is taken as the uncertainty. The obtained values are shown in Tab.~\ref{tab:pseudothresholds}. Based on these pseudothresholds and the data in Fig.~\ref{fig:matching_bp_comparison}, we observe that matching seems to perform best for the gross code, achieving logical failure rates relatively close to BP-OSD. The fact that the small-size oriented matching decoder performs better for the standard gross code than for the two-gross code is not unexpected; the key mechanism used in this procedure of removing up to two selected error clusters and re-decoding will become less effective as the expected number of clusters increases. Furthermore, the relative improvement, at least compared to BP, that we see on the $24 \times 24$ lattice is also not unexpected. This is because the intermediate size decoder uses the application of short strings as the mechanism for obtaining a final correction, side-stepping the above problem since all explicit error clusters need not be found. \\

\begin{figure*}[t]
    \centering
    \includegraphics[width=\linewidth]{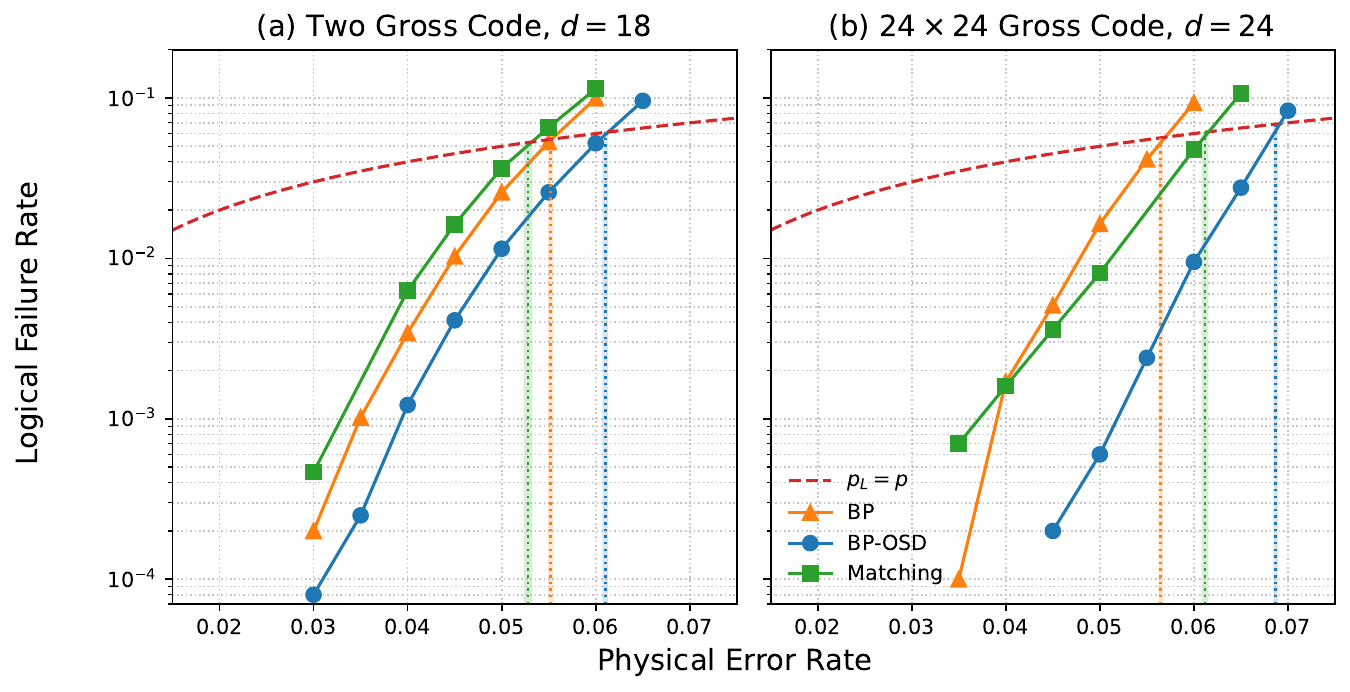}
    \caption{Decoding the two-gross code and the $24 \times 24$ gross code under i.i.d bit-flip noise using BP, BP-OSD, and the novel matching based decoder. The small size matching decoder is used for the two-gross code, while the intermediate size matching decoder is used for the  $24 \times 24$ gross code. Logical error rates are determined through Monte Carlo sampling of error configurations at the associated physical error rate followed by decoding. For a data point with logical error rate of order $10^{-k}$,  $10^{k+2}$ trials were run.  Statistical uncertainties are determined by binomial counting statistics and are contained within marker size. Vertical lines indicate pseudothresholds for each decoding method, and shaded regions indicate pseudothreshold uncertainty estimated from logical error rate uncertainty.}
    \label{fig:matching_bp_comparison}
\end{figure*}

\begin{table*}[t]
\centering
\begin{tabular}{lccc}
\hline
Decoder & [[144,12,12]] & [[288,12,18]]  & [[1152,16,24]] \\
\hline
BP-OSD   & $5.47 \pm 0.06\%$ & $6.10 \pm 0.03\%$ & $6.87 \pm 0.02\%$ \\
BP       & $4.95 \pm 0.06\%$ & $5.52 \pm 0.03\%$ & $5.65 \pm 0.02\%$ \\
Matching & $5.19 \pm 0.07\%$ & $5.28 \pm 0.05\%$ & $6.12 \pm 0.04\%$ \\
\hline
\end{tabular}
\caption{Pseudothresholds under bit-flip, code capacity noise for the codes considered. Pseudothreshold is estimated by linear interpolation between data points lying above and below the $p=p_L$ curve, and uncertainties are estimated by considering the interpolation between the maximum and minimum values respectively of the uncertainty window around each data point.  }
\label{tab:pseudothresholds}
\end{table*}
\indent We now briefly discuss runtime. While the MWPM procedure applied in the matching decoder is fast, there is a significant constant overhead in the current implementation of both decoders. Repeated decoding of the syndrome in all shifts of the $b\times b$ unit cell as well as repeated decoding of the subset-removal stage 2 syndromes in the small code size adaptation are costs that are very easily parallelized but can result in slowdown with a limited quantity of compute. Additionally, the potentially large constant associated with the algorithm used to compute potential error clusters has a nontrivial effect on runtime. Finally, the code used to run the matching decoder was implemented in Python, while the computational load for the BP and BP-OSD decoders is handled through generally faster C++. As a result of these factors, our matching implementation currently lags BP-OSD in speed. Our decoders serve as a proof of concept of the application of matching to TTI codes, and further optimizations needed to improve the runtime are left to future work.

\section{Circuit-level noise and toric spacetime matching}
\label{sec:circuit_noise}

The reductions in Sections~\ref{sec:decoding} and \ref{sec:cellular_decoding} extend to circuit-level Pauli noise in the memory experiment specified below. Every elementary fault is represented by a bounded spacetime chain since it can affect several data qubits and measurement outcomes. One can map the spacetime chain into TC spacetime chains before considering their boundaries for error correction. Using the decomposition and lifting identities from the \emph{cell-matching} decoder, this reduces decoding to a constant number of ordinary \emph{3D toric spacetime matching graph} instances. We prove a threshold for MWPM under local-stochastic circuit noise,.

Fix a compatible finite torus and the physical complex for one Pauli sector,
\begin{equation}
\calC:\quad C_2\xrightarrow{\partial_2}C_1\xrightarrow{\partial_1}C_0.
\end{equation}
Here $C_1$ records $Z$ errors and $C_0$ records their $X$-check syndromes; the dual construction treats $X$ errors. All chains are over $\F_2$. The code family and its coarse-graining are fixed, and every local map used below has support radius bounded independently of $L_x,L_y$ and the number of rounds. The number $r$ counts matching sectors for the Pauli type being decoded. We state distance bounds for $r\ge1$; when $r=0$, the product-state sector is corrected locally and carries no logical information.

We assume that the extraction circuit directly measures the local stabilizer generators described in $H$. Thus each measurement record $s^{(t)}$ belongs to the syndrome space $C_0$, and detectors are formed from consecutive records in these same coordinates.

\subsection{Extraction circuit and memory experiment}
\label{subsec:circuit-noise-extraction}
\label{subsec:circuit-noise-detection-events}

Fix one ancilla for every measured stabilizer generator. In the $Z$-check subround, prepare each ancilla in $\lvert0\rangle$, apply CNOTs from its data qubits to the ancilla, and measure the ancilla in $Z$. In the subsequent $X$-check subround, prepare each ancilla in $\lvert+\rangle$, apply CNOTs from the ancilla to its data qubits, and measure it in $X$. Use a fixed edge-coloured schedule within each subround. The schedule is geometrically local and has constant depth because check weights, qubit degrees, and interaction ranges are bounded. Ancillas are freshly prepared each round, with no adaptive gates. Separating the check types according to the colors avoids ordering conflicts between their CNOTs.

\begin{definition}[Circuit memory experiment and local-stochastic noise]
\label{def:phenomenological-memory}
\label{def:circuit-local-stochastic}
Start in the code space with known syndrome $s^{(0)}=0$ and perform $T\ge1$ noisy extraction rounds of the fixed circuit. Elementary fault locations are space-time locations in a circuit including state initializations, measurements, one-qubit gates, two-qubit gates, and idles where Pauli errors can occur. Preparation faults are placed after state preparation and measurement faults before measurement as they flip the measurement outcome bit.

Let $F$ be the random set of faulty locations, and write $f_\lambda$ for the type of fault that happened at $\lambda\in F$. For $p\in[0,1]$, assume for every finite set $U$ of locations that
\begin{equation}
\Pr[U\subseteq F]\le p^{|U|}.
\label{eq:circuit-local-stochastic}
\end{equation}
This allows for correlations across fault types and locations.

A perfect terminal syndrome measurement at $t=T+1$ introduces no further data noise. We can perform recovery with a Pauli-frame update. The perfect initial and final syndromes help to impose constraints on what the detectors can be. Note that a known nonzero initial syndrome can be removed by a fixed reference frame and the logical state need not be known.
\end{definition}

For each elementary fault $f$, track its effect through the subsequent Clifford gates. Conjugation by a Clifford gate maps
a Pauli operator to another Pauli operator, while a Pauli that
anticommutes with an ancilla's measured observable flips that
measurement outcome. We track Pauli operators modulo phase
and use the end of each extraction round as the reference time.

Let $a_f^{(t)}\in C_1$ denote the data-error pattern caused by
$f$ at the end of round $t$, in the Pauli sector being decoded.
Let $s_f^{(t)}\in C_0$ denote the changes to that round's
measurement record caused by $f$, relative to the fault-free
experiment. Define
\begin{equation}
e_f^{(t)}=a_f^{(t)}+a_f^{(t-1)},\qquad
\mu_f^{(t)}=s_f^{(t)}+\partial_1a_f^{(t)}.
\label{eq:circuit-round-frames}
\end{equation}
All additions are over $\F_2$. The first quantity records the
change in the data error during round $t$: an error persisting
unchanged from the previous round contributes zero to
$e_f^{(t)}$. The second records the discrepancy between the
measured syndrome and the syndrome of the data at the end
of the round.

In particular, $\mu_f^{(t)}$ need not arise from a faulty measurement. A data error occurring after a check's relevant
interaction may change the final data syndrome without changing that round's check outcome. Such a discrepancy is
also included in $\mu_f^{(t)}$.

For a realized fault history, sum these contributions over
the faults $f_\lambda$, $\lambda\in F$, and omit the fault
subscript for the resulting quantities. Because clifford propagation
and measurement-outcome flips are linear in the binary Pauli
representation, 
\begin{equation}
a^{(t)}=\sum_{u=1}^{t}e^{(u)},\qquad
s^{(t)}=\partial_1a^{(t)}+\mu^{(t)}.
\end{equation}
Thus the data error is the sum of its round-by-round increments,
and the measurement record is its end-of-round syndrome plus
the record discrepancy.

The ideal initial condition gives $a^{(0)}=\mu^{(0)}=0$.
The perfect terminal measurement introduces no new data error
and reports the final syndrome exactly, so we set
$e^{(T+1)}=\mu^{(T+1)}=0$. Existing data errors may remain:
$a^{(T+1)}=a^{(T)}$, and the terminal record is
$s^{(T+1)}=\partial_1a^{(T)}$.

\begin{lemma}[Bounded circuit-fault spread]
\label{lem:circuit-bounded-spread}
For the specified circuit, an isolated fault in round $u$ has $e_f^{(t)}=\mu_f^{(t)}=0$ for $t\ne u$. The supports of $e_f^{(u)}$ and $\mu_f^{(u)}$ lie in a bounded spatial neighbourhood of the fault, with bounds independent of $L_x,L_y,T$.
\end{lemma}
\begin{proof}
Within the faulty round, a Pauli can spread only through its forward light cone. The locality constraints from the 2D TTI code and its associated syndrome extraction circuit bound its support and the affected measurement records.

We next show that the resulting data error does not spread
further during subsequent fault-free rounds. In a $Z$-check
subround, each data qubit is a CNOT control and the ancilla
is the target. A data $X$ error can copy an $X$ onto the
ancilla, but this target $X$ cannot propagate to other data
controls. A data $Z$ error remains on its original qubit.
The same happens for an $X$-check subround. In addition, the same conclusion holds for a general Pauli error by
combining its $X$ and $Z$ components. Fresh ancillas ensure
that no ancilla error persists between rounds.

Consequently, the data-error pattern caused by $f$ remains
constant after round $u$: $a_f^{(t)}=a_f^{(u)}$ for $t>u$.
Each subsequent ideal round measures the commuting stabilizers
without changing their eigenvalues, and therefore reports
the syndrome of this persistent data error: $s_f^{(t)}=\partial_1a_f^{(t)}$.
It follows that, for $t>u$, $e_f^{(t)}=a_f^{(t)}+a_f^{(t-1)}=0$, $\mu_f^{(t)}=s_f^{(t)}+\partial_1a_f^{(t)}=0$.
Thus the data error may persist, but the only time the values change is in the faulty round.
\end{proof}

\begin{definition}[Syndrome differences and spacetime complex]
\label{def:syndrome-history-detection-events}
Let $T_\ast$ be the chain complex of the interval with vertices $1,\ldots,T+1$ and edges $[t,t+1]$ for $1\le t\le T$, with $\partial_T[t,t+1]=t+(t+1)$. Define the detector history
\begin{equation}
D=\sum_{t=1}^{T+1}(s^{(t)}+s^{(t-1)})\otimes t.
\label{eq:circuit-detector-history}
\end{equation}
The relevant degrees of $\calC\otimes T_\ast$ are
\begin{align}
C_{0,\mathrm{st}}&=C_0\otimes T_0,\\
C_{1,\mathrm{st}}&=(C_1\otimes T_0)\oplus(C_0\otimes T_1),\\
C_{2,\mathrm{st}}&=(C_2\otimes T_0)\oplus(C_1\otimes T_1),
\end{align}
with boundaries
\begin{align}
\partial_{1,\mathrm{st}}(h,v)
&=(\partial_1\otimes1)h+(1\otimes\partial_T)v,
\label{eq:physical-spacetime-boundary}\\
\partial_{2,\mathrm{st}}(u,w)
&=\bigl((\partial_2\otimes1)u+(1\otimes\partial_T)w,
\nonumber\\[-2pt]
&\hspace{4em}(\partial_1\otimes1)w\bigr).
\end{align}
Their composition is zero. For an elementary fault $f$, define the 1-chain
\begin{equation}
A(f)=\sum_{t=1}^{T}e_f^{(t)}\otimes t
+\sum_{t=1}^{T}\mu_f^{(t)}\otimes[t,t+1],
\label{eq:circuit-fault-representative}
\end{equation}
and set $A=\bigoplus_{\lambda\in F}A(f_\lambda)$. Horizontal edges represent data-frame changes and vertical edges represent measurement outcome discrepancies. In general $|A|\ne|F|$.
\end{definition}

This tensor complex retains the useful representation of phenomenological noise decoding~\cite{dennis2002topological,fowler2012surface}, but its distribution here is induced by circuit faults. For a general CSS code, a physical horizontal edge may have more than two boundary checks so this physical complex need not be a graph complex in general.

\begin{lemma}[Exact circuit detector identity]
\label{lem:phenomenological-history-boundary}
\label{lem:circuit-fault-chain}
Let $D_{\mathrm{circ}}(f)$ be the detector history caused by an isolated typed fault. Then
\begin{equation}
D_{\mathrm{circ}}(f)=\partial_{1,\mathrm{st}}A(f),\qquad
D=\partial_{1,\mathrm{st}}A.
\label{eq:exact-circuit-detector-identity}
\end{equation}
Each $A(f)$ has uniformly bounded spacetime support, although $D_{\mathrm{circ}}(f)$ can contain more than two detectors.
\end{lemma}
\begin{proof}
Fix an elementary fault $f$. Its spacetime chain is $A(f)=\sum_{t=1}^{T}e_f^{(t)}\otimes t
+\sum_{t=1}^{T}\mu_f^{(t)}\otimes[t,t+1].$
We compute its boundary at each time $t$.
The horizontal term contributes the syndrome
$\partial_1e_f^{(t)}$.
A temporal term $\mu_f^{(u)}\otimes[u,u+1]$
contributes $\mu_f^{(u)}$ at both endpoints.
Thus the coefficient of time $t$ in
$\partial_{1,\mathrm{st}}A(f)$ is $\partial_1e_f^{(t)}
+\mu_f^{(t-1)}+\mu_f^{(t)}$.
The two $\mu$ factors correspond to the fact that
a discrepancy in the recorded measurement at time $u$ affects both
consecutive comparisons involving that record.

Using the definitions of the $e_f$ and
$\mu_f$, we obtain
\begin{align}
&\partial_1e_f^{(t)}
+\mu_f^{(t-1)}+\mu_f^{(t)}
\\
&=\partial_1\bigl(a_f^{(t)}+a_f^{(t-1)}\bigr)
+\bigl(s_f^{(t-1)}+\partial_1a_f^{(t-1)}\bigr)
\\
&\qquad
+\bigl(s_f^{(t)}+\partial_1a_f^{(t)}\bigr)
\\
&=s_f^{(t-1)}+s_f^{(t)}.
\end{align}
The data-syndrome terms cancel in pairs over $\F_2$.
The remaining expression is exactly the detector pattern
caused by $f$ at time $t$.

At $t=1$, there is no incoming temporal edge, and the
ideal initial condition gives
$a_f^{(0)}=s_f^{(0)}=\mu_f^{(0)}=0$.
At $t=T+1$, there is no outgoing temporal edge or new
data-error increment. The perfect terminal measurement gives $e_f^{(T+1)}=\mu_f^{(T+1)}=0$, $s_f^{(T+1)}=\partial_1a_f^{(T)}$.
Hence the same boundary calculation holds at both temporal
endpoints. Summing over all times proves $\partial_{1,\mathrm{st}}A(f)
=\sum_{t=1}^{T+1}
\bigl(s_f^{(t)}+s_f^{(t-1)}\bigr)\otimes t
=D_{\mathrm{circ}}(f)$.

For a realized fault history, Clifford propagation and
measurement-record changes add linearly modulo two.
Linearity of the boundary map therefore gives $D=\sum_{\lambda\in F}D_{\mathrm{circ}}(f_\lambda) =\partial_{1,\mathrm{st}} \left(\sum_{\lambda\in F}A(f_\lambda)\right) =\partial_{1,\mathrm{st}}A$.
This identity does not require independent faults.

Finally, Lemma~\ref{lem:circuit-bounded-spread} shows that
a fault in round $u$ contributes only
$e_f^{(u)}\otimes u$ and
$\mu_f^{(u)}\otimes[u,u+1]$, both with bounded spatial
support. Thus $A(f)$ has uniformly bounded spacetime support.
Its boundary may nevertheless contain more than two detector
events, so an elementary circuit fault need not correspond
to a single matching edge.
\end{proof}

\subsection{Sector maps on spacetime}
\label{subsec:circuit-noise-spacetime-graphs}
\label{subsec:circuit-noise-layer-decoupling}
\label{subsec:circuit-noise-cell-matching}

Write $\calG^{(i)}$ for the spatial toric complex of sector $i$. Use the fixed linear maps
\begin{equation}
J^{(i)}:C_1\longrightarrow C_1^{(i)},\qquad
S^{(i)}:C_0\longrightarrow C_0^{(i)},
\end{equation}
with bounded support radius and
\begin{equation}
\partial_1^{(i)}J^{(i)}=S^{(i)}\partial_1.
\label{eq:circuit-static-commutation}
\end{equation}
The syndrome maps are defined on all of $C_0$, including corrupted records outside $\im\partial_1$.

For layer decoupling,
\begin{equation}
J^{(i)}=\pi_1^{(i)}\phi_1,\qquad
S^{(i)}=\pi_0^{(i)}\phi_0,
\label{eq:circuit-layer-maps}
\end{equation}
where $\phi$ acts on the physical one-Pauli-sector complex. Corollary~\ref{cor:chain-isomorphism-projection} proves~\eqref{eq:circuit-static-commutation}. The Laurent-polynomial matrices have a constant number of monomials per output coordinate, providing the needed bounds on the error and syndrome spread.

For cell matching, use $S^{(i)}s(u)=\mathsf{Coeff}(s_u)_i$ and $J^{(i)}$ from Lemma~\ref{lem:constant-spread-cellular}. For each basis single-qubit error, its bounded coarse syndrome has even parity in every sector. In addition, there is also bounded error and syndrome spread once again follows from the same analysis. 

\begin{definition}[Toric spacetime graph]
\label{def:spacetime-coarse-toric-graph-complex}
Let $\calG^{(i),\mathrm{st}}=\calG^{(i)}\otimes T_\ast$. Its 1-skeleton is the 3D toric spacetime matching graph: copies of the spatial toric graph at times $1,\ldots,T+1$, with vertical edges joining consecutive copies of each vertex. Spatial plaquettes and edge--time plaquettes form the 2-cells of the graph. This graph describes a two-dimensional toric code's syndrome history. It has maximum degree $\Delta\le6$ on square tori and $O(n_LT)$ vertices and edges, where $n_L$ is the number of data qubits.

Extend the maps in the previous sections by
\begin{equation}
J_{\mathrm{st}}^{(i)}
=(J^{(i)}\otimes1_{T_0})\oplus(S^{(i)}\otimes1_{T_1}),
\end{equation}
and set
\begin{equation}
E^{(i)}=J_{\mathrm{st}}^{(i)}A,
\qquad \Gamma^{(i)}(f)=J_{\mathrm{st}}^{(i)}A(f)
\label{eq:compiled-toric-fault}
\end{equation}
for each elementary circuit fault $f$. Thus $E^{(i)}=\bigoplus_{\lambda\in F}\Gamma^{(i)}(f_\lambda)$.
\end{definition}

\begin{lemma}[Spacetime boundaries on underlying toric graphs]
\label{lem:phenomenological-spacetime-pushforward}
\label{lem:circuit-hyperedge-pushforward}
The extended maps satisfy
\begin{equation}
\partial_1^{(i),\mathrm{st}}J_{\mathrm{st}}^{(i)}
=(S^{(i)}\otimes1_{T_0})\partial_{1,\mathrm{st}}.
\label{eq:central-circuit-commuting-identity}
\end{equation}
In particular, every projected detector hyperedge is exactly the boundary of a bounded chain in the ordinary toric spacetime graph:
\begin{equation}
\partial_1^{(i),\mathrm{st}}\Gamma^{(i)}(f)
=(S^{(i)}\otimes1)D_{\mathrm{circ}}(f).
\label{eq:circuit-projected-hyperedge}
\end{equation}
Consequently $D^{(i)}=(S^{(i)}\otimes1)D$ is exactly the boundary of $E^{(i)}$.
\end{lemma}

\begin{proof}
It suffices to check the identity on horizontal and vertical
edges. For a horizontal edge $a\otimes t$, the
static boundary identity gives $\partial_1^{(i),\mathrm{st}}J_{\mathrm{st}}^{(i)}(a\otimes t)
=\partial_1^{(i)}J^{(i)}a\otimes t
=S^{(i)}\partial_1a\otimes t$.
For a vertical generator $m\otimes[t,t+1]$, we obtain
\begin{align}
&\partial_1^{(i),\mathrm{st}}J_{\mathrm{st}}^{(i)}
   (m\otimes[t,t+1])\\
&\qquad=S^{(i)}m\otimes t+S^{(i)}m\otimes(t+1),
\end{align}
which is also obtained by taking the temporal boundary first
and then applying $S^{(i)}$. Linearity therefore proves \eqref{eq:central-circuit-commuting-identity}.
Applying this identity to $A(f)$ and using
Lemma~\ref{lem:circuit-fault-chain} yields $\partial_1^{(i),\mathrm{st}}\Gamma^{(i)}(f)
=(S^{(i)}\otimes1_{T_0})D_{\mathrm{circ}}(f)$.
Summing over realized faults gives
$\partial_1^{(i),\mathrm{st}}E^{(i)}=D^{(i)}$.
Finally, the bounded support of $A(f)$ and the bounded spread
of $J^{(i)}$ and $S^{(i)}$ give a uniform support bound for
$\Gamma^{(i)}(f)$. No map on degree-two chains is required.
\end{proof}

\begin{corollary}[Sectorization commutes with time differences]
\label{lem:sectorization-commutes-time-differences}
The sector detector history can be computed by first applying $S^{(i)}$ to every measured record and then taking consecutive differences.
\end{corollary}
\begin{proof}
Linearity gives $S^{(i)}(s^{(t)}+s^{(t-1)})=S^{(i)}s^{(t)}+S^{(i)}s^{(t-1)}$.
\end{proof}

\subsection{Decoder and physical lifting}
\label{subsec:circuit-noise-unified-principle}

Define $\operatorname{Proj}_{\mathrm{sp}}$ by summing horizontal chains over time modulo two and sending vertical chains to zero. It can similarly be defined on vertices by summing over time, and on degree-two chains by summing spatial plaquettes and sending edge--time rectangles to zero. These are chain maps from the spacetime complex to its spatial complex. The two endpoints of a vertical edge cancel under projection, as do the two horizontal sides of an edge--time rectangle. In particular,
\begin{equation}
\operatorname{Proj}_{\mathrm{sp}}A=a^{(T)},\qquad
\operatorname{Proj}_{\mathrm{sp}}E^{(i)}=J^{(i)}a^{(T)}.
\label{eq:circuit-projection-true-error}
\end{equation}

\begin{algorithm}[H]
\caption{Circuit-level matching decoder (one Pauli sector)}
\label{alg:phenomenological-matching-decoder}
\label{alg:circuit-matching-decoder}
\begin{algorithmic}[1]
\Require{History $s^{(0)},\ldots,s^{(T+1)}$ with the temporal conditions of Definition~\ref{def:circuit-local-stochastic}; fixed sector and lifting maps.}
\Ensure{Final physical Pauli correction $\widehat E$.}
\For{$i\gets1$ \textbf{to} $r$}
  \State Sectorize each record: $\tilde s^{(i,t)}=S^{(i)}s^{(t)}$ for $t=0,\ldots,T+1$.
  \State Form $D^{(i)}=\sum_{t=1}^{T+1}(\tilde s^{(i,t)}+\tilde s^{(i,t-1)})\otimes t$.
  \State Use MWPM to find a minimum-weight toric spacetime chain $R^{(i)}$ with boundary $D^{(i)}$.
  \State Set $\widehat\eta^{(i)}=\operatorname{Proj}_{\mathrm{sp}}R^{(i)}$.
\EndFor
\State From the perfect terminal record $s^{(T+1)}$, compute the local product-state correction or the cellular flushing operator $\Psi(s^{(T+1)})$.
\State Lift the $\widehat\eta^{(i)}$ together with this local correction, using the lifting step of Algorithm~\ref{alg:decoupling-decoder} or \ref{alg:cellular-decoder}.
\State \Return $\widehat E$.
\end{algorithmic}
\end{algorithm}

In a more compact notation, the decoder performs the following:
\begin{equation}
\begin{split}
\operatorname{Dec}_{\mathrm{circ}}
={}&\operatorname{Lift}_{s^{(T+1)}}\circ\\[-2pt]
&\left(\bigoplus_{i=1}^r\operatorname{Proj}_{\mathrm{sp}}
\circ\operatorname{MWPM}_i\right)
\circ\operatorname{Sect}_{\mathrm{st}},
\end{split}
\label{eq:circuit-decoder-composition}
\end{equation}
where $\operatorname{Sect}_{\mathrm{st}}$ first applies each $S^{(i)}$ to the syndrome record at every time and then takes consecutive differences of the resulting sector records. By Corollary~\ref{lem:sectorization-commutes-time-differences}, this gives $D^{(i)}=(S^{(i)}\otimes1)D$ which means that we could have taken the time differences of the syndromes first before sectorization. Note that $A$ and $\Gamma^{(i)}(f)$ are unknown and are only constructed for the purpose of our analysis.

For cell matching, sectorization first computes the charge coefficients in the basis subcells from each record $s^{(t)}$ at each timestep. It does not apply the corresponding flushing operator to the data. MWPM uses the full sector detector history and finds the shortest path on the spacetime graph before expanding its matched paths into a recovery chain. For the cell-matching decoder, after spatial projection, the physical correction is composed with the flushing operator derived from the final syndrome:
\begin{equation}
\widehat E=\Psi(s^{(T+1)})+
\sum_i\mathsf P^{(i)}(\widehat\eta^{(i)}).
\label{eq:circuit-cell-lift}
\end{equation}

\begin{theorem}[Correctness of circuit decoding]
\label{thm:phenomenological-lifting}
\label{thm:circuit-level-generic-upgrade}
Assume the layer-decoupling chain isomorphism and locally correctable product-state sector, or the cell-matching decoder's error correction properties of Lemmas~\ref{lem:cellular-error-reconstruction} and \ref{lem:contractible-cycle-stabilizer-property}, in particular $d_Z(\calC)>\max\{w_{\mathrm{loc}},4w_{\mathrm{edge}}\}$ for the $Z$-error correction (and the dual condition for $X$ errors). If every residual toric spacetime cycle
\begin{equation}
Z^{(i)}=E^{(i)}+R^{(i)}
\end{equation}
is homologically trivial, Algorithm~\ref{alg:circuit-matching-decoder} leaves a physical stabilizer in the decoded Pauli sector.
\end{theorem}
\begin{proof}
The sector detector histories computed by Algorithm~\ref{alg:circuit-matching-decoder} equal $(S^{(i)}\otimes1)D$ by Corollary~\ref{lem:sectorization-commutes-time-differences}. Lemma~\ref{lem:circuit-hyperedge-pushforward} therefore gives $\partial_1^{(i),\mathrm{st}}Z^{(i)}=0$. Projecting a boundary gives a spatial boundary, so the hypothesis implies
\begin{equation}
J^{(i)}a^{(T)}+\widehat\eta^{(i)}\in\im\partial_2^{(i)}.
\label{eq:circuit-residual-spatial-boundary}
\end{equation}
The projected recovery operator has boundary $\tilde s^{(i,T+1)}=S^{(i)}s^{(T+1)}$ since summing the differences $\tilde s^{(i,t)}+\tilde s^{(i,t-1)}$ over time cancels every intermediate stabilizer measurement record and $\tilde s^{(i,0)}=0$. Thus, the projected matching recovery and the local correction in the final step use the same final syndrome.

For the layer-decoupling decoder, every residual Pauli operator in each toric graph is therefore a boundary. The local correction in the final step gives the correct product-state syndrome. Its syndrome-free residual is definitely a boundary because this sector has no logical qubits. The inverse chain isomorphism sends their direct sum to a physical stabilizer.

For the cell-matching decoder, $s^{(T+1)}=\partial_1a^{(T)}$ give
\begin{equation}
a^{(T)}+\Psi(s^{(T+1)})+
\sum_i\mathsf P^{(i)}(J^{(i)}a^{(T)})\in\im\partial_2.
\end{equation}
Combining this with~\eqref{eq:circuit-cell-lift}, the physical residual differs by a stabilizer from
$\sum_i\mathsf P^{(i)}(J^{(i)}a^{(T)}+\widehat\eta^{(i)})$.
Each summand is a stabilizer by~\eqref{eq:circuit-residual-spatial-boundary} and the contractible-cycle property.
\end{proof}

It is not hard to see that Theorem~\ref{thm:circuit-level-generic-upgrade} allows us to conclude that both decoders reduce circuit-level memory decoding to $r$ ordinary toric spacetime MWPM calls per Pauli $X$ or $Z$ correction as well as their existing local correction in the final step.

\subsection{Local-stochastic noise and the memory threshold}
\label{subsec:circuit-noise-local-stochastic}
\label{subsec:circuit-noise-threshold}

Choose constants $K,M\ge1$ such that, for every toric sector $i$,
\begin{equation}
|\Gamma^{(i)}(f)|\le K,\qquad
\#\{f:e\in\operatorname{supp}\Gamma^{(i)}(f)\}\le M
\label{eq:circuit-spread}
\end{equation}
for every elementary circuit fault $f$ and toric spacetime edge $e$. One may take $K=\max\{1,\max_{f,i}|\Gamma^{(i)}(f)|\}$ over all locations and fault types in the fixed family. The bound $M$ counts faults, including distinct Pauli types at the same location, whose effects on the underlying toric sectors can contain a given edge. The fact that our 2D TTI code and its associated syndrome extraction circuit are local make both constants independent of $L_x,L_y,T$.

\begin{lemma}[Local-stochastic spacetime pushforward]
\label{lem:phenomenological-local-stochastic-pushforward}
\label{lem:circuit-local-stochastic-pushforward}
For every set $B$ of toric spacetime edges,
\begin{equation}
\Pr[B\subseteq\operatorname{supp}E^{(i)}]\le q^{|B|},
\qquad q=Mp^{1/K}.
\label{eq:circuit-effective-q}
\end{equation}
\end{lemma}
\begin{proof}
An edge present in $E^{(i)}$ necessarily has an odd number of faults in the spacetime picture of the 2D TTI code. Assign one such fault to every edge in $B$. There are at most $M^{|B|}$ assignments. Discard assignments requiring different Pauli types at the same location, since they cannot occur. In a remaining assignment each distinct location supplies one fault and can account for at most $K$ assigned edges. At least $\lceil|B|/K\rceil$ distinct faulty locations are therefore necessary. The probability of any fixed assignment is at most the probability that these locations are faulty, hence at most $p^{\lceil|B|/K\rceil}\le p^{|B|/K}$ by~\eqref{eq:circuit-local-stochastic}. A union bound proves the desired claim.
\end{proof}

\begin{definition}[Spacetime distance for toric sectors]
\label{def:circuit-spacetime-distance}
For $r\ge1$ and the perfect temporal boundary conditions of Definition~\ref{def:circuit-local-stochastic}, set
\begin{equation}
D_{\mathrm{st}}^{(i)}
=\min\{|z|:\partial_1^{(i),\mathrm{st}}z=0,
\ [z]\ne0\text{ in }H_1(\calG^{(i),\mathrm{st}})\},
\label{eq:circuit-relative-distance}
\end{equation}
and let $D_{\mathrm{st}}=\min_iD_{\mathrm{st}}^{(i)}$. 
\end{definition}

Allowing a logical loop to extend through time cannot make it shorter than the shortest logical loop in a single spatial slice because projection will clean away Pauli operators that are stacked above each other in time. Projection cannot increase edge weight, while a shortest spatial logical cycle lies in any slice. Consequently
\begin{equation}
D_{\mathrm{st}}^{(i)}=d_i,\qquad
D_{\mathrm{st}}=d_{TC}=\min\{L_x/b,L_y/b\}
\label{eq:circuit-perfect-boundary-distance}
\end{equation}
for the coarse tori and every $T\ge1$.

The physical-distance comparison \eqref{eq:layer-physical-distance-comparison} gives
$d_L\le w_{\mathrm{inv}}d_{TC}$, where $w_{\mathrm{inv}}$
bounds the physical weight of the inverse image of a
toric edge under the mapping from the decoder proof. Combining this
with \eqref{eq:circuit-perfect-boundary-distance} yields
\begin{equation}
D_{\mathrm{st}}\ge \frac{d_L}{w_{\mathrm{inv}}}.
\label{eq:circuit-distance-comparison}
\end{equation}
The same bound applies to the cell-matching decoder when both
constructions use the common coarse-graining scale $b$
specified in Section~\ref{sec:cellular_decoding}.

\begin{theorem}[Circuit-level matching decoder for 2D TTI CSS codes]
\label{thm:phenomenological-matching-threshold}
\label{thm:circuit-universal-matching}
Fix a \TwoDTLTI CSS family with $r=O(1)$ toric sectors, satisfying the hypotheses of Theorem~\ref{thm:circuit-level-generic-upgrade}, with compatible coarse tori of side lengths at least three. Use the fixed extraction circuit above, or another syndrome extraction circuit for which the exact bounded-spread property of Lemma~\ref{lem:circuit-fault-chain} and~\eqref{eq:circuit-projection-true-error} hold. Consider the memory experiment of Definition~\ref{def:circuit-local-stochastic} and the constants in~\eqref{eq:circuit-spread}. For standard MWPM and $q=Mp^{1/K}\le1$, the failure probability satisfies
\begin{equation}
P_{\mathrm{fail}}\le c\,r\,N_V
\sum_{\ell\ge D_{\mathrm{st}}}
\left[2(\Delta-1)\sqrt{Mp^{1/K}}\right]^\ell,
\label{eq:circuit-peierls-bound}
\end{equation}
where $N_V$ is the maximum number of vertices in a sector's spacetime graph, $N_V=O(n_LT)$, and $c=\Delta/(\Delta-1)$ suffices. In particular, a sufficient threshold is
\begin{equation}
p_0=\left[4M(\Delta-1)^2\right]^{-K}.
\label{eq:circuit-threshold}
\end{equation}
For fixed $0<p<p_0$, there is a constant $\kappa(p)>0$ such that
\begin{equation}
P_{\mathrm{fail}}\le O(rn_LT)\exp[-\kappa(p)D_{\mathrm{st}}].
\end{equation}
If $\log(n_LT)=o(D_{\mathrm{st}})$, this bound decays exponentially.

The runtime is
\begin{equation}
O\!\left(
n_LT+\sum_{i=1}^{r}
\operatorname{MATCHING}(N_i^{\mathrm{st}})
\right).
\label{eq:circuit-runtime}
\end{equation}
\end{theorem}
\begin{proof}
By Theorem~\ref{thm:circuit-level-generic-upgrade}, a logical failure for the 2D TTI code requires a nontrivial residual cycle in some toric sector. Fix such a toric sector and write $E=E^{(i)}$, $R=R^{(i)}$. Decompose the even-degree subgraph $E+R$ into edge-disjoint simple cycles. At least one cycle $C$ represents nonzero homology, with length $\ell\ge D_{\mathrm{st}}$. The recovery $R+C$ has the same boundary as $R$ since $C$ has no boundary. Thus the minimum-weight property of $R$ gives
\begin{equation}
|R|\le |R+C|=|R|+\ell-2|R\cap C|.
\end{equation}
Thus $|R\cap C|\le\ell/2$. Each edge of $C\subseteq E+R$ belongs to exactly one of $E,R$, hence $|E\cap C|\ge\lceil\ell/2\rceil$. For a fixed $C$, the pushforward lemma and a union bound over its subsets give
\begin{equation}
\Pr[|E\cap C|\ge\ell/2]
\le2^\ell q^{\lceil\ell/2\rceil}
\le(2\sqrt q)^\ell.
\end{equation}
There are at most $N_V\Delta(\Delta-1)^{\ell-1}$ candidate length-$\ell$ cycles, by counting nonbacktracking walks. Summing over lengths and the $r$ sectors proves~\eqref{eq:circuit-peierls-bound}. Neither this union bound nor the union over the two Pauli sectors requires independence.

For $0<p<p_0$, let $\rho(p)=2(\Delta-1)\sqrt{Mp^{1/K}}<1$.
The sum in \eqref{eq:circuit-peierls-bound} is therefore
a convergent geometric series, giving $P_{\mathrm{fail}}
\le \frac{c\,rN_V}{1-\rho(p)}
      \rho(p)^{D_{\mathrm{st}}}$.
For fixed $p$, set $\kappa(p)=-\log\rho(p)>0$ and use
$N_V=O(n_LT)$ to obtain $P_{\mathrm{fail}}
\le O(rn_LT)e^{-\kappa(p)D_{\mathrm{st}}}$.
If $\log(n_LT)=o(D_{\mathrm{st}})$, the $O(rn_LT)$ factor grows
more slowly than the exponential suppression, so this
bound decays exponentially in $D_{\mathrm{st}}$.

For the runtime, each sector map has bounded support, so
applying all $r=O(1)$ maps to one syndrome record costs
$O(n_L)$ operations. Sectorizing all records and forming
their consecutive differences therefore costs $O(n_LT)$.
Each sector's spacetime graph has $O(n_LT)$ vertices and
edges and requires one MWPM call. Finally, summing the
horizontal recovery edges over time and applying the
spatial lifting maps costs at most $O(n_LT)$. Thus the
total runtime is $O\!\left(
n_LT+\sum_{i=1}^{r}
\operatorname{MATCHING}(N_i^{\mathrm{st}})
\right)$.
\end{proof}

\begin{corollary}[Number of correctable errors and circuit-fault distance]
\label{cor:circuit-adversarial-distance}
Under the same hypotheses, the decoder corrects every circuit fault history satisfying
\begin{equation}
|F|\le\left\lfloor\frac{D_{\mathrm{st}}-1}{2K}\right\rfloor.
\label{eq:circuit-adversarial-radius}
\end{equation}
Define $d_{\mathrm{circ}}$ as the minimum number of faulty circuit locations with $D=0$ and a nontrivial final logical Pauli in the decoded sector. Then
\begin{equation}
d_{\mathrm{circ}}\ge\left\lceil D_{\mathrm{st}}/K\right\rceil.
\label{eq:circuit-fault-distance}
\end{equation}
\end{corollary}
\begin{proof}
A minimum-weight recovery has $|R^{(i)}|\le|E^{(i)}|$. Consequently, in every sector,
\begin{equation}
|Z^{(i)}|\le |E^{(i)}|+|R^{(i)}|\le 2K|F|<D_{\mathrm{st}}.
\end{equation}
The residual cycle is therefore a boundary. Theorem~\ref{thm:circuit-level-generic-upgrade} gives a physical stabilizer residual.

Consider a fault history that produces no detection events,
so $D=0$. Since each detection event records a difference
between consecutive syndrome records and $s^{(0)}=0$,
the terminal record also satisfies $s^{(T+1)}=0$.
Moreover,
$\partial_1^{(i),\mathrm{st}}E^{(i)}
=(S^{(i)}\otimes1)D=0$,
so every effective error chain $E^{(i)}$ is a cycle.

With no detection events, MWPM returns the zero recovery.
The terminal local correction is also zero because
$s^{(T+1)}=0$. If every $E^{(i)}$ were homologically
trivial, Theorem~\ref{thm:circuit-level-generic-upgrade}
would therefore imply that the final physical error
is a stabilizer. Consequently, a fault history with
$D=0$ that implements a nontrivial logical operator
must have a homologically nontrivial $E^{(i)}$ in at
least one sector.

Such a chain contains at least $D_{\mathrm{st}}$ edges.
On the other hand, each faulty location contributes at
most $K$ edges to that sector, giving $D_{\mathrm{st}}\le |E^{(i)}|\le K|F|$.
Thus every undetected logical fault history requires at
least $\lceil D_{\mathrm{st}}/K\rceil$ faulty locations,
which proves the circuit-fault-distance bound.
\end{proof}

These are sufficient guarantees for the specified circuit, temporal boundaries, and the MWPM decoder. The spread $K$ includes hook errors and the static maps, so we do not expect to correct $d_L/2$ faults nor achieve the TC's numerical circuit-level noise threshold.

\section{Conclusions and future directions}
\label{sec:conclusions}

In this work, we developed a graph-matching approach to decoding 2D TTI CSS codes that exploits their equivalence to (multiple copies of) the TC.
Concretely, we introduced two decoders: the layer-decoupling decoder, which can be viewed as a generalization of the projection/restriction decoder for the color code~\cite{delfosse2014decoding,kubica2023efficient}, and the cell-matching decoder, which performs local flushing of excitations in bounded neighborhoods of coarse-grained unit cells in order to respect the locality of the noise.
We established analytical performance guarantees for MWPM under the stated reduction and lifting hypotheses, as well as numerically studied the performance of our decoders for the hexagonal color code and practically relevant instances of BB codes, including the gross code.
Our numerical results provide a proof of concept for the proposed graph-matching approach, with performance for the BB examples competitive with BP-based decoders~\cite{roffe2020decoding,panteleev2021degenerate,muller2025improved}.
We further showed how the same sector maps extend to circuit-level Pauli noise for the specified extraction circuit and temporal boundaries by replacing the 2D matching subproblems with 3D space--time matching subproblems.

There are several directions for future work.
First, our decoders could potentially be improved by optimizing the edge weights used in the graph-matching subroutines, for instance by first running BP to obtain relevant soft information, such as approximate marginals or local likelihood ratios~\cite{higgott2023}.
Such hybrid BP--graph-matching decoders may better capture correlations within and between the different TC copies.
Second, implementing and benchmarking our circuit-level decoder remains open.
A single circuit fault can produce several correlated toric spacetime edges; practical matching weights could account for these correlations.
It would therefore be valuable to understand whether the cell-matching decoder can be integrated with existing circuit-level decoding techniques, such as circuit-level decoders for the color code~\cite{chamberland2020,gidney2023new,lee2025color}, or with detector-error-model methods that explicitly handle detector hyperedges.
We envision that the circuit-level memory decoder could be extended to the setting of transversal algorithmic fault tolerance~\cite{Zhou2025}, where the decoding problem for logical circuits with transversal gates is solved using graph-matching techniques~\cite{cain2025fastcorrelateddecodingtransversal,SerraPeralta2026}. Such an extension requires checking propagation and boundary conditions for those logical operations.
Third, one could consider quantum codes defined on lattices with open boundary conditions, as many proposed constructions emphasize such layouts to improve the feasibility of experimental implementations~\cite{eberhardt2024pruning,liang2025planar,steffan2025tile}.
From the perspective of decoding, the presence of boundaries may modify global constraints on excitations and affect the local flushing of excitations near boundaries.
This, in turn, would require identification of excitations that condense at given boundaries and adaptation of matching techniques by, for instance, allowing excitations to be paired with appropriate boundaries.
Finally, it is natural to ask whether a graph-matching approach can also be applied to other quantum low-density parity-check codes, such as quantum product codes~\cite{tillich2013quantum,panteleev2021quantum}; these codes allow optimized QEC protocols~\cite{manes2025distance,tan2025effective,berthusen2025adaptive} and support many logical gates~\cite{quintavalle2023partitioning,xu2025fast,berthusen2025automorphism,golowich2025quantum,golowich2025constant,tan2025single,li2025transversal}.
The main obstacle is that their syndrome patterns do not necessarily resemble point-like TC excitations; rather, they often lack a clear geometric, low-dimensional structure.
One other possible direction would be to generalize the decoders for the higher-dimensional color code~\cite{kubica2023efficient}, although it remains unclear whether a constant-depth circuit mapping generic quantum product codes to (multiple copies of) the higher-dimensional TC should exist~\cite{haah2016algebraic}.

\section*{Acknowledgements}
\label{sec:acknowledgments}

We thank Casey Duckering, Refaat Ismail, and Kaavya Sahay for helpful discussions. We also thank the journal reviewers for the feedback for improving our manuscript.
Part of this work was done while S.J.S.T, E.H., and P.L. were interning at QuEra Computing Inc.
S.J.S.T. acknowledges funding and support from Joint Center for Quantum Information and Computer Science (QuICS) Lanczos Graduate Fellowship and the National University of Singapore (NUS) Development Grant. E.H. is supported by the Fulbright Future Scholarship.
A.K. acknowledges support from the NSF (QLCI, Award No. OMA-2120757), IARPA and the Army Research Office (ELQ Program, Cooperative Agreement No. W911NF-23-2-0219).

\emph{Note added.---}During the preparation of this manuscript, we became aware of independent work by Sahay, Williamson, and Brown that uses a similar intuition based on the equivalence of 2D TTI codes and the TC, and develops a matching decoder for BB codes~\cite{sahay2026matching}.

\textbf{Use of AI tools.---} Claude 4 was used for generating some of the figures in the manuscript. GPT 5 was used to polish the language, and identify issues with notation. GPT 6 identified an issue with the circuit-level noise proof and helped correct it.

\bibliography{bib/ref}

\clearpage
\onecolumngrid
\appendix
\section{Bivariate bicycle codes}
\label{app:summary-bb-codes}
While the polynomial formalism holds for any \TwoDTLTI code, we are mainly interested in a specific class of \TwoDTLTI codes known as bivariate bicycle (BB) codes. These codes are defined using bivariate polynomials and have been shown to have good performance in practice. A lot of the most popular \TwoDTLTI codes such as the toric code and color code can be viewed as instantiations of the BB codes. Recently, a subclass of BB codes that have interactions that go beyond nearest-neighbor qubits have been proposed in Ref.~\cite{bravyi2024high} and are known to be a good code in practice - 0.8\% circuit level error rate and many transversal gates, including automorphism gates~\cite{yoder2025tour}. Other research works have also studied the problem of embedding general BB codes with non-periodic boundaries as well as the logical gates that are amenable for these codes~\cite{berthusen2025toward,eberhardt2024logical,liang2025planar}. While other works including Ref.~\cite{liang2025generalized} have considered embedding BB codes in exotic tori, we are primarily interested in a restricted class of BB codes that are known to be embeddable in a 2D toric layout in this work. We reuse some of the notation from Ref.~\cite{chen2025anyon}.

\begin{definition}[Bivariate Bicycle Codes]
  \label{def:bivariate-bicycle-codes}
  The bivariate bicycle (BB) code is constructed using classical cyclic codes. To define the BB code on two $\ell \times m$ lattices, we first introduce a permutation matrix 
  \begin{equation}S_k = \begin{pmatrix}
    0 & 1 & 0 & \cdots & 0 & 0\\
    0 & 0 & 1 & \cdots & 0 & 0\\
    \vdots & \vdots & \vdots & \ddots & \vdots & \vdots\\
    0 & 0 & 0 & \cdots & 1 & 0\\
    1 & 0 & 0 & \cdots & 0 & 0
  \end{pmatrix} \in \F_2^{k \times k}.\end{equation}
  In addition, we define the following:
  \begin{equation}x\coloneqq S_\ell \otimes \ident_m,\quad y\coloneqq \ident_\ell \otimes S_m.\end{equation}
  Note that $xy = yx$ because they commute. 
  We denote the set of all bivariate monomials as $\calM$:
  \begin{equation}\calM = \{1, x, \ldots, x^{\ell - 1}, y, xy, \ldots, x^{\ell - 1}y, \ldots, y^{m-1}, xy^{m-1}, \ldots, x^{\ell - 1}y^{m-1}\}.\end{equation}
  A BB code $\BBc(\overline{\alpha}, \overline{\beta}, a, b)$ for $\overline{\alpha} \equiv -\alpha, \overline{\beta} \equiv -\beta$ and $\alpha\geq \beta \geq 0$ is then defined by two polynomials
  \begin{equation}a(x,y) = 1 + x + x^{\overline{\alpha}}y^b,\quad b(x,y) = 1 + y + y^{\overline{\beta}}x^a.\end{equation}
  The parity check matrices are given by the following:
  \begin{align}
    H_X = \left(a(x,y)\;\middle|\;b(x,y)\right)), \quad\quad H_Z = \left(b^*(x,y)\;\middle|\;a^*(x,y)\right),
  \end{align}
  The left partition of the parity-check matrices corresponds to the qubits that reside on the horizontal edges of the lattice and the right partition of the parity-check matrices corresponds to the qubits that reside on the vertical edges of the lattice.
\end{definition}

While Definition~\ref{def:bivariate-bicycle-codes} is a definition that has been used by Refs.~\cite{bravyi2024high,chen2025anyon} for its simplicity, we note that it is possible to state it with an algebraic formalism using bivariate polynomials over the finite field $\F_2$. We note that this particular definition might be helpful for the reader to gain more intuition about the code structure of the BB codes.

Since the BB codes are defined on two $\ell \times m$ lattices, which can also be viewed as the $\ell m$ horizontal and $\ell m$ vertical edges in a single 2D square lattice, we can label the $(i, j)$-th element in this lattice as $x^i y^j$ for $i = 0,, \ldots, \ell - 1$ and $j = 0, \ldots m-1$. Using this polynomial labeling of qubits on the lattice, we can characterize the BB code algebraically. The benefit of the following algebraic formalism is that one can then take advantage of preexisting computational algebraic geometry tools, such as the Gr\"obner basis algorithm, to analyze the properties of concrete BB codes such as their $\llbracket n, k, d\rrbracket$ parameters.

Let $R = \F_2[x^{\pm 1}, y^{\pm 1}]/\langle x^\ell - 1, y^m - 1 \rangle$ be the ring of bivariate polynomials modulo the relations $x^\ell = 1$ and $y^m = 1$. The qubits of the code correspond to elements of $R$, and the stabilizer checks are defined by the polynomials $a(x, y)$ and $b(x, y)$ as in Definition~\ref{def:bivariate-bicycle-codes}.
\paragraph{Stabilizer Groups.}
Using the syndrome-map convention of Section~\ref{subsec:polynomial-representation-2d-topological-codes}, stabilizer supports are generated by the adjoint check maps. A translated $X$ stabilizer has support:
\begin{equation}\left\{\left(x^i y^j a^*(x,y), x^i y^j b^*(x,y)\right)\;\middle|\; 0 \leq i < \ell, 0 \leq j < m\right\},\end{equation}
and a translated $Z$ stabilizer has support:
\begin{equation}\left\{\left(x^i y^j b(x,y), x^i y^j a(x,y)\right)\;\middle|\; 0 \leq i < \ell, 0 \leq j < m\right\}.\end{equation}

Therefore, the linear spaces generated by the $X$ and $Z$ stabilizer generators are $R$-modules:
\begin{align}
     S_X &= \text{im} H_X^\dagger = R \cdot \left(a^*(x,y), b^*(x,y)\right) = \left\{\left(f(x,y)a^*(x,y), f(x,y)b^*(x,y)\right)\;\middle|\;f(x,y) \in R\right\}, \\
    S_Z &= \text{im} H_Z^\dagger = R \cdot \left(b(x,y), a(x,y)\right) = \left\{\left(f(x,y)b(x,y), f(x,y)a(x,y)\right)\;\middle|\;f(x,y) \in R\right\}.
\end{align}

\paragraph{Syndromes.}
Consider a $Z$ error configuration that flip $X$ stabilizers. Because we have two block of $\ell \times m$ qubits with polynomial labelings, we can regard the $Z$ error configuration as two polynomials $(f(x,y), g(x,y))$. The resulting $X$ syndrome for the error configuration is given by 
\begin{equation}H_X \begin{pmatrix}
    f(x,y) \\ g(x,y)
\end{pmatrix},\end{equation}
which is naturally a linear combination of the columns of the parity check matrix $H_X$. We note that the columns of $H_X$ can also be regarded as polynomials:
\begin{equation}\col(H_X) = \left\{x^i y^j a(x,y)\right\} \cup \left\{x^i y^j b(x,y)\right\}.\end{equation}
Using this algebraic perspective, the $X$ syndrome of this $Z$ error configuration is an element of a polynomial ideal
\begin{equation}f(x,y)a(x,y) + g(x,y) b(x,y) \in \left\langle a(x,y), b(x,y)\right\rangle.\end{equation}
Similarly, the $Z$ syndrome for an $X$ error configuration is an element of the polynomial ideal $\left\langle b^*(x,y), a^*(x,y)\right\rangle$.

\paragraph{Logical Operators.}
Suppose we have a $Z$ logical operator that corresponds to some configuration $(f(x,y), g(x,y))$. Naturally, it must be in the kernel of $H_X$ which implies that it has trivial $X$ syndrome, i.e.,
\begin{equation}f(x,y)a(x,y) + g(x,y)b(x,y) = 0.\end{equation}
These configurations $(f,g)$ form the \emph{syzygy} $R$-module $\syz(a(x,y), b(x,y))$. Another requirement of this arbitrary $Z$ logical operator is that it must not be in the image of $H_Z^\dagger$, i.e., it cannot be written as a linear combination of the rows of $H_Z$. Therefore, we can fully characterize all $Z$ logical operators by a quotient $R$-module:
\begin{equation}\calL_Z = \ker H_X/\im H_Z^\dagger = \syz(a(x,y), b(x,y))/R\cdot(b(x,y),a(x,y)).\end{equation}
Similarly, the $X$ logical operators can be characterized by the quotient $R$-module:
\begin{equation}\calL_X = \ker H_Z/\im H_X^\dagger = \syz(b^*(x,y), a^*(x,y))/R\cdot(a^*(x,y),b^*(x,y)).\end{equation}

\paragraph{Example: Kitaev Toric Code (KTC).}
The toric code is a special case of the BB code with $\ell = m = d$ and
\begin{equation}a(x,y) = 1 + x,\quad b(x,y) = 1 + y.\end{equation}
The $X$ and $Z$ stabilizer supports are generated by the adjoint check maps:
\begin{align}
  S_X&=R\cdot(1+\bar x,1+\bar y),\\
S_Z&=R\cdot(1+y,1+x).
\end{align}

Write $\Sigma_x=\sum_{i=0}^{d-1}x^i$ and $\Sigma_y=\sum_{j=0}^{d-1}y^j$. The undetectable $Z$ and $X$ errors are, respectively,
\begin{align}
\syz(a,b)&=R\cdot(1+y,1+x)+R\cdot(\Sigma_x,0)+R\cdot(0,\Sigma_y),\\
  \syz(b^*,a^*)&=R\cdot(1+\bar x,1+\bar y)+R\cdot(\Sigma_y,0)+R\cdot(0,\Sigma_x).
\end{align}

These sums of $R$-submodules need not be direct. After taking the stabilizer quotients, the logical spaces are
\begin{align}
\calL_Z&=\syz(a,b)/S_Z
=\operatorname{span}_{\F_2}\{[(\Sigma_x,0)],[(0,\Sigma_y)]\}\cong\F_2^2,\\
\calL_X&=\syz(b^*,a^*)/S_X
=\operatorname{span}_{\F_2}\{[(\Sigma_y,0)],[(0,\Sigma_x)]\}\cong\F_2^2.
\end{align}

\section{An algorithmic approach to decoupling \TwoDTLTI codes}
\label{app:decoupling}
In this section, we present an algorithmic approach to decoupling \TwoDTLTI codes into independent copies of TCs and some trivial product state. This notion of decoupling was introduced by Bomb{\'\i}n~\cite{bombin2014structure} and Haah~\cite{haah2013commuting} independently, and it plays a crucial role in our understanding of 2D topological codes. For the CSS families satisfying the exactness hypothesis below, a fixed coarse-graining and local Clifford reduction expose the toric sectors used by our decoder. In Ref.~\cite{haah2016algebraic}, Haah introduced the notion of coarse-graining and provided an algorithmic proof for how one might be able to decouple an arbitrary 2D topological CSS code that exhibits topological order into independent copies of TCs and some trivial product state. Because the language used in Ref.~\cite{haah2016algebraic} and Ref.~\cite{bombin2014structure} is quite abstract, we will present a more concrete algorithmic approach using modern quantum error correction language. During our exposition, we will also highlight the relationship between the terms we use and the terms that Haah and Bomb{\'\i}n use in their works so that the reader can easily refer to their works for a more in-depth understanding of the decoupling process.

This section is organized in the following way. We first give a quick recap on Clifford operations before introducing the concept of coarse-graining and discussing the valid symplectic transformations that can be used to decouple the code into TCs. Subsequently, we present a slightly modified version of the decoupling algorithm stated in Ref.~\cite{haah2016algebraic}. Finally, we provide an example to illustrate the decoupling process in action.

\subsection{Clifford Operations and Symplectic Group}
\label{subsec:clifford-operations}
Clifford operations are a class of quantum gates that preserve the structure of stabilizer codes. They can be represented as elements of the symplectic group $Sp(2n, \F_2)$, which consists of $2n \times 2n$ matrices that preserve a symplectic form. The symplectic form is a bilinear form that encodes the commutation relations of the Pauli operators.

\begin{definition}[Symplectic Group]
  The symplectic group $Sp(2n, \F_2)$ is the group of $2n \times 2n$ matrices $M$ over $\F_2$ such that
  \begin{equation}M^\top \Omega M = \Omega,\end{equation}
  where $\Omega = \begin{pmatrix}
    0 & I_n\\
    I_n & 0
  \end{pmatrix}$ is the symplectic form.
  The symplectic group is a subgroup of the general linear group $GL(2n, \F_2)$.
  The symplectic group is generated by the following elementary operations:
  \begin{itemize}
    \item \emph{CNOT gates}: $CNOT_{i,j}$ for $i, j \in [n]$.
    \item \emph{Hadamard gates}: $H_i$ for $i \in [n]$.
    \item \emph{Phase gates}: $S_i$ for $i \in [n]$.
  \end{itemize}
\end{definition}

These elementary operations can be understood as column operations on the binary symplectic matrix of a CSS code. For example, a CNOT gate $CNOT_{i,j}$ corresponds to adding column $i$ to column $j$ in the left block and column $j+n$ to column $i+n$ in the right block. A Hadamard gate $H_i$ corresponds to swapping the $i$-th and $(i + n)$-th columns, and a Phase gate $S_i$ corresponds to adding column $i$ to column $i+n$. These column transformations directly correspond to the action of the Clifford gates on the $n$-qubit Pauli operators of the CSS code. It is not too hard to then see that the Clifford gates permute the Pauli operators of the CSS code.

\subsection{Coarse-graining}
\label{subsec:coarse_graining}
Coarse-graining groups a fixed $b\times b$ block of lattice sites into one super-cell. The resulting description retains translations by $b$ sites in either direction and allows local algebraic operations on all qubits in that block.

Algebraically, the act of coarse-graining is essentially the act of taking a smaller base ring for bivariate Laurent polynomials. So instead of working with $R = \F_2[x^{\pm 1}, y^{\pm 1}]$, we can work with a smaller base ring $R' = \F_2[x^{\pm b}, y^{\pm b}] = \F_2[x^{\prime\,\pm 1}, y^{\prime\,\pm 1}]$ for some integer $b \geq 1$. The choice of $b$ determines the size of the super-cell. The larger the value of $b$, the larger the super-cell, and the more coarse-grained the code is. Each $b\times b$ super-cell contains $b^2q$ qubits. Since $R$ is a free $R'$-module of rank $b^2$, restriction of scalars multiplies both the row and column dimensions of the parity-check matrix by $b^2$. It is sufficient to identify the new matrix representation of the generators $x$ and $y$ of the ring over $\F_2$ since the representation of an arbitrary polynomial in $R$ can be expressed as products, sums, and inverses of these generator matrices. The column basis for these matrices is given by the following:
\begin{equation}\calB = \set{1, x, x^2, \ldots, x^{b-1}, y, xy, x^2y, \ldots, x^{b-1}y, \ldots, y^{b-1}, xy^{b-1}, \ldots, x^{b-1}y^{b-1}}.\end{equation}
Thus, we can express the generator matrices of the coarse-grained code as follows:
\begin{align}
    &\left(x\,:\,R^{\prime\,b^2} \to R^{\prime\,b^2}\right) = \mathrm{blockdiag}^b\begin{pmatrix}
        0 & 0 & \cdots & 0 & x' \\
        1 & 0 & \cdots & 0 & 0 \\
        0 & 1 & \cdots & 0 & 0 \\
        \vdots & \vdots & \ddots & \vdots & \vdots \\
        0 & 0 & \cdots & 1 & 0
    \end{pmatrix} \in R'^{b^2 \times b^2},\\
    &\left(y\,:\,R^{\prime\,b^2} \to R^{\prime\,b^2}\right) = \begin{pmatrix}
        0 & 0 & \cdots & 0 & y' & 0 & 0 & \cdots & 0 & 0 \\
        0 & 0 & \cdots & 0 & 0 & y' & 0 & \cdots & 0 & 0 \\
        0 & 0 & \cdots & 0 & 0 & 0 & y' & \cdots & 0 & 0 \\
        \vdots & \vdots & \vdots & \vdots & \vdots & \vdots & \vdots & \ddots & \vdots & \vdots \\
        0 & 0 & \cdots & 0 & 0 & 0 & 0 & \cdots & y' & 0 \\
        0 & 0 & \cdots & 0 & 0 & 0 & 0 & \cdots & 0 & y' \\
        1 & 0 & \cdots & 0 & 0 & 0 & 0 & \cdots & 0 & 0  \\
        0 & 1 & \cdots & 0 & 0 & 0 & 0 & \cdots & 0 & 0 \\
        \vdots & \vdots & \ddots & \vdots & \vdots & \vdots & \vdots & \ddots & \vdots & \vdots \\
        0 & 0 & \cdots & 1 & 0 & 0 & 0 & \cdots & 0 & 0 
        \end{pmatrix} \in R'^{b^2 \times b^2}
\end{align}
where the blockdiag operator constructs a block diagonal matrix with $b$ copies of the given matrix on its diagonal. The matrix representation for the generator $y$ contains some values of $y'$ at the top $b$ rows and the last $b$ columns as shown in the matrix above. The representation also has $b^2 - b$ 1s on the diagonal in the $b^2 - b \times b^2 - b$ submatrix in the bottom left corner of the matrix. These generator matrices effectively perform the following ordered shifts on the basis $\calB$:
\begin{align}
    x\,:\, \calB &= \set{1, x, x^2, \ldots, x^{b-1}, y, xy, x^2y, \ldots, x^{b-1}y, \ldots, y^{b-1}, xy^{b-1}, \ldots, x^{b-1}y^{b-1}} \\
    &\mapsto \set{x, x^2, x^3, \ldots, x^b = x', xy, x^2y, x^3y, \ldots, x^by = x'y, \ldots, xy^{b-1}, x^2y^{b-1}, \ldots, x'y^{b-1}},\\
    y\,:\, \calB &= \set{1, x, x^2, \ldots, x^{b-1}, y, xy, x^2y, \ldots, x^{b-1}y, \ldots, y^{b-1}, xy^{b-1}, \ldots, x^{b-1}y^{b-1}} \\
    &\mapsto \set{y, xy, x^2y, \ldots, x^{b-1}y, y^2, xy^2, x^2y^2, \ldots, x^{b-1}y^2, \ldots, y^b = y', xy', \ldots, x^{b-1}y'}.
\end{align}

With these generator matrices that we have constructed, we can now construct a matrix for any polynomial. For example, if we have a polynomial $1 + y + y^2$, we can construct the matrix representation of this polynomial for $b = 2$ as follows:
\begin{align}
    R \ni 1 + y + y^2  &\mapsto \begin{pmatrix} 
        1 & 0 & 0 & 0 \\
        0 & 1 & 0 & 0 \\
        0 & 0 & 1 & 0 \\
        0 & 0 & 0 & 1  
        \end{pmatrix} + \begin{pmatrix} 
        0 & 0 & y' & 0 \\
        0 & 0 & 0 & y' \\
        1 & 0 & 0 & 0 \\
        0 & 1 & 0 & 0
        \end{pmatrix}  + \begin{pmatrix} 
        y' & 0 & 0 & 0 \\
        0 & y' & 0 & 0 \\
        0 & 0 & y' & 0 \\
        0 & 0 & 0 & y'
        \end{pmatrix} \\
        &= \begin{pmatrix} 
        1 + y' & 0 & y' & 0 \\
        0 & 1 + y' & 0 & y' \\
        1 & 0 & 1 + y' & 0 \\
        0 & 1 & 0 & 1 + y'
        \end{pmatrix}  \in R'^{2^2 \times 2^2}.
\end{align}
With this faithful representation of $R$ by matrices over $R'$,  we can now observe the interactions of the code Hamiltonian in the coarse-grained lattice with greater clarity.

\subsection{Local Disentanglement and Stabilizer Relabeling via Symplectic Transformations}
\label{subsec:decoupling_symplectic}
In this subsection, we discuss the valid symplectic transformations that can be applied to the coarse-grained parity-check matrix of the code to perform local disentanglement and stabilizer relabeling on the super-cells of the coarse-grained lattice. The goal of these transformations is to decouple the code into independent copies of TCs and some trivial product state.

Let $q_b$ be the number of qubits in a coarse cell, so that the
symplectic parity-check matrix has $n_c=2q_b$ columns. CSS-preserving
column transformations have the form
\begin{equation}
V=\operatorname{diag}(A,A^{-\dagger}),\qquad
V\Omega V^\dagger=\Omega.
\end{equation}
Here $\dagger$ includes the Laurent antipode. The elementary rules are:
\begin{itemize}
\item Swapping columns $i,j\leq q_b$ also swaps columns $i+q_b,j+q_b$.
\item Multiplying column $i$ by a monomial $m$ also multiplies column
$i+q_b$ by the same $m$, since $(m^*)^{-1}=m$. This is a
bounded-displacement relabeling of that qubit type.
\item Adding $a$ times column $i$ to column $j$, with $i\ne j$ in the
first block, requires adding $a^*$ times column $j+q_b$ to column
$i+q_b$ in the second block. Indeed,
$(I+aE_{ij})^{-\dagger}=I+a^*E_{ji}$ over $\F_2$.
For a monomial coefficient this is a layer of displaced CNOTs; a
polynomial coefficient is a finite sum of such operations.
\end{itemize}
For a fixed finite sequence, these operations and their inverses have
uniformly bounded range. A bounded-displacement relabeling need not
itself be a finite-depth CNOT circuit; the virtual decoder uses its
local action and inverse. Row operations instead change the
stabilizer generating basis. They are not symplectic transformations
on the Pauli module. The elementary row operations are:
\begin{itemize}
    \item \textbf{Swapping rows.} This operation corresponds to swapping the labels of the stabilizer checks in each super-cell within the coarse-grained lattice. It is a valid operation since it does not change the stabilizer group of the code.
    \item \textbf{Multiplying a row by a nonzero monomial in $R'$.} This operation corresponds to scaling the stabilizer check by a nonzero monomial in $R'$. It effectively induces a relabeling of the same stabilizer check in each super-cell in the coarse-grained lattice. This is akin to the case of multiplying a column by a nonzero monomial in $R'$, where the same qubits are translated by the same vector in the coarse-grained lattice. 
    \item \textbf{Adding one row to another.} This operation corresponds to adding one stabilizer check to another. It effectively induces a relabeling of the stabilizer check that changes as a result of the addition within each super-cell in the coarse-grained lattice. This is a valid operation since it does not change the stabilizer group of the code.
\end{itemize}

\subsection{Setting up for the Decoupling Algorithm}
\label{subsec:decoupling_setup}
In this subsection, we lay down the sufficient conditions for the decoupling algorithm to work. These conditions are detailed in the works done by Haah~\cite{haah2016algebraic} and Bomb{\'\i}n~\cite{bombin2014structure} in different abstract languages. We will use the standard quantum error correction language to describe these conditions and relate them to the terminology used in the works of Haah and Bomb{\'\i}n in case the reader is interested in referring to their works for a more in-depth understanding of the decoupling process.

\subsubsection{Topological Order Condition}
\label{subsec:topological_order_condition}
We use the excitation map $H=\operatorname{diag}(H_X,H_Z)$
on Pauli vectors ordered as $(e_Z,e_X)$. If there are $q$ qubits
per lattice site, define
\begin{equation}
\Omega=\begin{pmatrix}0&I_q\\I_q&0\end{pmatrix}.
\end{equation}
The stabilizer map in this ordering is $\sigma=\Omega H^\dagger$.
The infinite-lattice topological-order condition is
\begin{equation}
\ker H=\im(\Omega H^\dagger),
\end{equation}
where the maps are defined over $R=\F_2[x^{\pm1},y^{\pm1}]$.
Thus every finite-support Pauli operator with zero syndrome is
a finite product of stabilizer generators, up to phase.

Let $t_X$ and $t_Z$ be the numbers of $X$- and $Z$-check generator
types per lattice site, respectively, and let $q$ be the number of
qubits per site. The check matrices have dimensions
$H_X\in R^{t_X\times q}$ and $H_Z\in R^{t_Z\times q}$.
Their excitation map is
\begin{equation}
\label{eq:excitation-map}
H=\begin{pmatrix}H_X&0\\0&H_Z\end{pmatrix}
:R^q\oplus R^q\longrightarrow R^{t_X}\oplus R^{t_Z}.
\end{equation}
Two excitation patterns represent the same charge class precisely
when their difference belongs to $\im H$.

For $Z$-error decoding, the physical CSS chain complex is
\begin{equation}
C_2=R^{t_Z}\xrightarrow{\partial_2=H_Z^\dagger}
C_1=R^q\xrightarrow{\partial_1=H_X}C_0=R^{t_X}.
\end{equation}
The dual complex, with $X$ and $Z$ interchanged, describes $X$ errors.
The topological-order condition is equivalent to
\begin{equation}
\ker H_X=\im H_Z^\dagger,\qquad
\ker H_Z=\im H_X^\dagger.
\end{equation}
Thus both complexes have trivial first homology on the infinite
lattice. This excludes finite-support logical operators, but does
not imply that all homology groups vanish: the degree-zero quotient
$H_0=\coker H_X$ of the displayed complex records charge classes
and need not be zero.
Coarse-graining is restriction of scalars to
$R'=\F_2[x^{\pm b},y^{\pm b}]$ and preserves these kernel/image
equalities. Finite-torus logical operators are instead computed
after imposing the periodic identifications; the infinite-lattice
exactness condition is not asserted on that finite quotient.
Note that some recently constructed self-dual BB codes can have useful finite-size parameters even if they do not satisfy this topological order condition~\cite{xu2025batched}.

\subsubsection{Torsion of the Cokernel}
\label{subsec:torsion_cokernel}
The cokernel of the excitation map is the quotient module
\begin{equation}
\coker H=(R^{t_X}\oplus R^{t_Z})/\im H.
\end{equation}
Here $R^{t_X}\oplus R^{t_Z}$ records finite-support binary patterns
of violated $X$ and $Z$ checks. The submodule $\im H$ consists of
patterns produced by finite-support Pauli errors. Two patterns
therefore define the same class in $\coker H$ if and only if their
difference can be created by such an error.

For the TC, it is not too hard to see that 
\begin{equation}\coker\,H_{TC} \cong \F_2 \oplus \F_2\end{equation}
because each check type has one even and one odd charge class.
Next, let us define the torsion submodule for some $R$-module $M$.
\begin{definition}[Torsion]\label{def:torsion}
For an $R$–module $M$, the torsion submodule is
\begin{equation}\calT(M):=\{\,m\in M \;|\; \exists\,0\neq r(x,y)\in R \text{ with } r m = 0 \,\}.\end{equation}
\end{definition}
A class $[e]\in\coker H$ is torsion if there exists a nonzero
polynomial $r(x,y)\in R$ such that $r[e]=0$ in $\coker H$.
Equivalently, $re\in\im H$: a finite superposition of translates
of the excitation pattern $e$ can be created by a finite-support
Pauli operator. A general torsion relation alone does not establish
transport in every direction. The translation-annihilator relations
in Section~\ref{subsec:annihilator_mobility} provide the directional
transport conditions used below.
Let us construct an example using the TC. Consider a single point-like excitation $e$ on an $X$-type stabilizer check labeled by the monomial $x^i y^j$ that belongs to the odd parity class of the possible $X$-type excitation classes.
By choosing $r = (1 + x) \in R$, we can see that 
\begin{equation}
H\begin{pmatrix}x^iy^j\\0\\0\\0\end{pmatrix}
=\begin{pmatrix}(1+x)x^iy^j\\0\end{pmatrix}\in\im H.
\end{equation}
In other words, the single-qubit Pauli $Z$ on the horizontal qubit labeled by $x^i y^j$ can move the excitation $e$ to the $X$-type stabilizer check labeled by $x^{i+1} y^j$. Suppose there is some other defect at $x^{i'} y^{j'}$. We can always find a \emph{string} of Pauli operators that moves the excitation $e$ to the desired location to annihilate the defect. In terms of the anyonic language, the finite set of anyon types directly coincides with the finite abelian group of torsion classes:
\begin{equation}\label{eq:torsioncount}
\#\{\text{anyon types, including vacuum}\}
=|\calT(\coker H)|
=2^{\dim_{\F_2}\calT(\coker H)}.
\end{equation}

The torsion charge module is invariant, up to isomorphism, under
invertible local Clifford transformations and changes of the local
generating presentation~\cite[Proposition V.12]{haah2016algebraic}.
Translations act on this module and may permute its classes; the
coarse-graining below makes the retained translations act trivially.

\subsubsection{Annihilator and Mobility}
\label{subsec:annihilator_mobility}
The annihilator ideal of the cokernel $\coker\,H$ is defined as
\begin{definition}[Annihilator Ideal]\label{def:annihilator_ideal}
    For any $R$–module $M$, the \emph{annihilator ideal}
\begin{equation}\ann_R(M):=\{\,r\in R \;|\; rM=0\,\}.\end{equation}
\end{definition}
We sometimes drop the subscript $R$ if it is clear which ring we are working over.
The annihilator of $\coker H$ consists of polynomials that kill
every excitation class. For a redundant generating presentation,
$\coker H$ can contain a free summand, so its annihilator must not be
identified with that of its torsion submodule. The nonredundant
presentation below has a torsion cokernel. In that presentation,
relations of the form $(x^b-1)[e]=(y^b-1)[e]=0$ give transport of
$[e]$ by $b$ sites along either axis.
Similar to the torsion submodule, let us use the TC as an example to illustrate the annihilator ideal. For the TC, we can see that the annihilator ideal is generated by the polynomials $x - 1$ and $y - 1$. In other words, the annihilator ideal is given by
\begin{equation}
\label{eq:ann-inclusion}
(x-1,\;y-1)=\ann (\calT(\coker\, H)).
\end{equation}
Equation~\eqref{eq:ann-inclusion} reinforces the notion of the existence of $1$D string operators that translate any point charge by one step along the $x$ or $y$ direction. In Bomb\'{\i}n’s language, these are the strings whose commutation with local checks is confined to their endpoints and which transport a given charge.

\subsubsection{Sufficient Conditions for Decoupling}
Now, we are ready to state the sufficient conditions for the decoupling algorithm to work. The most important ingredient is the preparation of a suitably coarse-grained parity-check matrix $H$ that satisfies the following conditions:

\begin{proposition}[Suitable Coarse-Grained Parity-Check Matrix~{\cite[Restatement of Proposition V.11]{haah2016algebraic}}]
\label{prop:suitable-coarse-grained-parity-check-matrix}
Let $H$ be a two-dimensional translation-invariant CSS excitation
map satisfying $\ker H=\im(\Omega H^\dagger)$ over
$R=\F_2[x^{\pm1},y^{\pm1}]$. There is a local generating
presentation $H'$ of the same stabilizer module with
\begin{equation}
\im(\Omega H^\dagger)=\im(\Omega H^{\prime\dagger})
=\ker H',\qquad \ker H^{\prime\dagger}=0.
\end{equation}
If a site contains $q$ qubits, $H'$ has $q$ rows and $2q$ columns.
There is a positive integer $b$ such that
\begin{equation}
(x^b-1)\coker H'=(y^b-1)\coker H'=0.
\end{equation}
For a nonzero charge module its annihilator over
$R'=\F_2[x^{\pm b},y^{\pm b}]$ is $(x^b-1,y^b-1)$;
for the zero module the annihilator is $R'$.
\end{proposition}

The proposition first changes the local generating presentation;
this is distinct from coarse-graining or an invertible row operation
on a redundant list. Equality of stabilizer modules expresses each
new generator as a finite product of translates of old generators,
and conversely. Fixing these expressions gives finite Laurent maps
converting the measured syndrome to the selected syndrome. These
relations remain valid on compatible finite tori. For CSS codes we
choose generators separately in the two Pauli submodules.
The resulting full check matrix has $t_X+t_Z=q$ rows; equality of
$t_X$ and $t_Z$ is not required.

We use this fixed presentation for the square invertible row maps
in Section~\ref{sec:decoding} and for the charge representatives in
Section~\ref{sec:cellular_decoding}. A syndrome-conversion step is
included when the measured presentation differs. Its finite support
changes only the uniform dependency constants. The annihilation
relations give finite-support transport operators; neither these
relations nor existence of a solution alone bounds its support by
a prescribed cell. We fix finitely many solutions and allow bounded
neighborhoods of the cells, as in the cell-matching construction.

The structural reduction then gives the following toric-sector count.
\begin{theorem}[Number of Independent Copies of TCs~{\cite[Theorem V.13]{haah2016algebraic}}]
    \label{thm:decoupling_count}
    For any two-dimensional LTI CSS code, if the parity-check matrix $H$ satisfies 
    \begin{equation}\ker H = \im(\Omega H^\dagger) \text{ over } R = \F_2[x^{\pm 1}, y^{\pm 1}],\end{equation}
    then, after a fixed coarse-graining, the code becomes a tensor product of finitely many toric-code copies and a product state under a fixed local Clifford reduction. The number of copies of the toric code in the CSS code is equal to $\frac{1}{2}\dim_{\F_2} \calT(\coker\,H).$
\end{theorem}

The constructive proof is given in Ref.~\cite{haah2016algebraic}.
It recursively isolates toric blocks by local column transformations
and stabilizer-basis changes after coarse-graining. Each toric block
contributes two independent charge classes and two qubits per coarse
site. The remaining blocks are single-qubit product states.
All matrices and their inverses are fixed Laurent-polynomial matrices
for the chosen code family. Their identities therefore descend to
tori whose periods are multiples of the final coarse-graining scale.
For decoding we use the CSS-preserving form of this reduction, with
$V=\operatorname{diag}(V_X,V_X^{-\dagger})$; the chain-isomorphism
hypothesis in the decoding theorems specifies the required maps.

\subsection{Color Code Example}
\label{subsec:color-code-example}
In this subsection, we show how the decoupling algorithm works in action by applying it to the color code. Because the color code is a 2D TTI code, we can apply the decoupling algorithm to the color code to decouple it into two TC copies and some trivial product state. We will not list out every step of the decoupling process but will instead highlight the key steps in the decoupling process. 

The color code can be defined by the following parity-check matrix in the symplectic representation:
\begin{equation}
H_{CC} = \begin{pmatrix}
    1 + x + xy & 1 + y + xy & 0 & 0 \\
    0 & 0 & 1 + y^{-1} + x^{-1}y^{-1} & 1 + x^{-1} + x^{-1}y^{-1}
\end{pmatrix}
\end{equation}
The first step is to coarse-grain the parity-check matrix $H_{CC}$ by choosing $b = 3$ to obtain the following parity-check matrix defined over the smaller base ring $R' = \F_2[x^{\pm 3}, y^{\pm 3}]$. In the displayed expanded matrices below, we relabel the coarse variables $x'=x^3,y'=y^3$ as $x,y$ for readability:

\providecommand{\DecouplingMatrixScale}{0.78}

\begin{equation}
\label{eq:HCC-coarse-grained}
\begingroup
\setlength{\arraycolsep}{1.2pt}
\renewcommand{\arraystretch}{0.9}
\makebox[\linewidth][c]{\scalebox{\DecouplingMatrixScale}{$
\left(
\begin{array}{*{18}{c}|*{18}{c}}
1 & 0 & x & 0 & 0 & 0 & 0 & 0 & xy & 1 & 0 & 0 & 0 & 0 & 0 & y & 0 & xy & 0 & 0 & 0 & 0 & 0 & 0 & 0 & 0 & 0 & 0 & 0 & 0 & 0 & 0 & 0 & 0 & 0 & 0\\
1 & 1 & 0 & 0 & 0 & 0 & y & 0 & 0 & 0 & 1 & 0 & 0 & 0 & 0 & y & y & 0 & 0 & 0 & 0 & 0 & 0 & 0 & 0 & 0 & 0 & 0 & 0 & 0 & 0 & 0 & 0 & 0 & 0 & 0\\
0 & 1 & 1 & 0 & 0 & 0 & 0 & y & 0 & 0 & 0 & 1 & 0 & 0 & 0 & 0 & y & y & 0 & 0 & 0 & 0 & 0 & 0 & 0 & 0 & 0 & 0 & 0 & 0 & 0 & 0 & 0 & 0 & 0 & 0\\
0 & 0 & x & 1 & 0 & x & 0 & 0 & 0 & 1 & 0 & x & 1 & 0 & 0 & 0 & 0 & 0 & 0 & 0 & 0 & 0 & 0 & 0 & 0 & 0 & 0 & 0 & 0 & 0 & 0 & 0 & 0 & 0 & 0 & 0\\
1 & 0 & 0 & 1 & 1 & 0 & 0 & 0 & 0 & 1 & 1 & 0 & 0 & 1 & 0 & 0 & 0 & 0 & 0 & 0 & 0 & 0 & 0 & 0 & 0 & 0 & 0 & 0 & 0 & 0 & 0 & 0 & 0 & 0 & 0 & 0\\
0 & 1 & 0 & 0 & 1 & 1 & 0 & 0 & 0 & 0 & 1 & 1 & 0 & 0 & 1 & 0 & 0 & 0 & 0 & 0 & 0 & 0 & 0 & 0 & 0 & 0 & 0 & 0 & 0 & 0 & 0 & 0 & 0 & 0 & 0 & 0\\
0 & 0 & 0 & 0 & 0 & x & 1 & 0 & x & 0 & 0 & 0 & 1 & 0 & x & 1 & 0 & 0 & 0 & 0 & 0 & 0 & 0 & 0 & 0 & 0 & 0 & 0 & 0 & 0 & 0 & 0 & 0 & 0 & 0 & 0\\
0 & 0 & 0 & 1 & 0 & 0 & 1 & 1 & 0 & 0 & 0 & 0 & 1 & 1 & 0 & 0 & 1 & 0 & 0 & 0 & 0 & 0 & 0 & 0 & 0 & 0 & 0 & 0 & 0 & 0 & 0 & 0 & 0 & 0 & 0 & 0\\
0 & 0 & 0 & 0 & 1 & 0 & 0 & 1 & 1 & 0 & 0 & 0 & 0 & 1 & 1 & 0 & 0 & 1 & 0 & 0 & 0 & 0 & 0 & 0 & 0 & 0 & 0 & 0 & 0 & 0 & 0 & 0 & 0 & 0 & 0 & 0\\
\hline
0 & 0 & 0 & 0 & 0 & 0 & 0 & 0 & 0 & 0 & 0 & 0 & 0 & 0 & 0 & 0 & 0 & 0 & 1 & 0 & 0 & 1 & 1 & 0 & 0 & 0 & 0 & 1 & 1 & 0 & 0 & 1 & 0 & 0 & 0 & 0\\
0 & 0 & 0 & 0 & 0 & 0 & 0 & 0 & 0 & 0 & 0 & 0 & 0 & 0 & 0 & 0 & 0 & 0 & 0 & 1 & 0 & 0 & 1 & 1 & 0 & 0 & 0 & 0 & 1 & 1 & 0 & 0 & 1 & 0 & 0 & 0\\
0 & 0 & 0 & 0 & 0 & 0 & 0 & 0 & 0 & 0 & 0 & 0 & 0 & 0 & 0 & 0 & 0 & 0 & 0 & 0 & 1 & x^{-1} & 0 & 1 & 0 & 0 & 0 & x^{-1} & 0 & 1 & x^{-1} & 0 & 0 & 0 & 0 & 0\\
0 & 0 & 0 & 0 & 0 & 0 & 0 & 0 & 0 & 0 & 0 & 0 & 0 & 0 & 0 & 0 & 0 & 0 & 0 & 0 & 0 & 1 & 0 & 0 & 1 & 1 & 0 & 0 & 0 & 0 & 1 & 1 & 0 & 0 & 1 & 0\\
0 & 0 & 0 & 0 & 0 & 0 & 0 & 0 & 0 & 0 & 0 & 0 & 0 & 0 & 0 & 0 & 0 & 0 & 0 & 0 & 0 & 0 & 1 & 0 & 0 & 1 & 1 & 0 & 0 & 0 & 0 & 1 & 1 & 0 & 0 & 1\\
0 & 0 & 0 & 0 & 0 & 0 & 0 & 0 & 0 & 0 & 0 & 0 & 0 & 0 & 0 & 0 & 0 & 0 & 0 & 0 & 0 & 0 & 0 & 1 & x^{-1} & 0 & 1 & 0 & 0 & 0 & x^{-1} & 0 & 1 & x^{-1} & 0 & 0\\
0 & 0 & 0 & 0 & 0 & 0 & 0 & 0 & 0 & 0 & 0 & 0 & 0 & 0 & 0 & 0 & 0 & 0 & y^{-1} & y^{-1} & 0 & 0 & 0 & 0 & 1 & 0 & 0 & 0 & y^{-1} & 0 & 0 & 0 & 0 & 1 & 1 & 0\\
0 & 0 & 0 & 0 & 0 & 0 & 0 & 0 & 0 & 0 & 0 & 0 & 0 & 0 & 0 & 0 & 0 & 0 & 0 & y^{-1} & y^{-1} & 0 & 0 & 0 & 0 & 1 & 0 & 0 & 0 & y^{-1} & 0 & 0 & 0 & 0 & 1 & 1\\
0 & 0 & 0 & 0 & 0 & 0 & 0 & 0 & 0 & 0 & 0 & 0 & 0 & 0 & 0 & 0 & 0 & 0 & x^{-1}y^{-1} & 0 & y^{-1} & 0 & 0 & 0 & 0 & 0 & 1 & x^{-1}y^{-1} & 0 & 0 & 0 & 0 & 0 & x^{-1} & 0 & 1
\end{array}
\right)
$}}
\endgroup
\end{equation}

Next, we proceed to normalize all rows in the coarse-grained parity-check matrix by multiplying each row by the appropriate monomial so that there is no negative power of $x$ or $y$ in the entries of the parity-check matrix. 

The resulting normalized parity-check matrix is:
\begin{equation}
\label{eq:HCC-coarse-grained-normalized}
\begingroup
\setlength{\arraycolsep}{1.2pt}
\renewcommand{\arraystretch}{0.9}
\makebox[\linewidth][c]{\scalebox{\DecouplingMatrixScale}{$
\left(
\begin{array}{*{18}{c}|*{18}{c}}
1 & 0 & x & 0 & 0 & 0 & 0 & 0 & xy & 1 & 0 & 0 & 0 & 0 & 0 & y & 0 & xy & 0 & 0 & 0 & 0 & 0 & 0 & 0 & 0 & 0 & 0 & 0 & 0 & 0 & 0 & 0 & 0 & 0 & 0\\
1 & 1 & 0 & 0 & 0 & 0 & y & 0 & 0 & 0 & 1 & 0 & 0 & 0 & 0 & y & y & 0 & 0 & 0 & 0 & 0 & 0 & 0 & 0 & 0 & 0 & 0 & 0 & 0 & 0 & 0 & 0 & 0 & 0 & 0\\
0 & 1 & 1 & 0 & 0 & 0 & 0 & y & 0 & 0 & 0 & 1 & 0 & 0 & 0 & 0 & y & y & 0 & 0 & 0 & 0 & 0 & 0 & 0 & 0 & 0 & 0 & 0 & 0 & 0 & 0 & 0 & 0 & 0 & 0\\
0 & 0 & x & 1 & 0 & x & 0 & 0 & 0 & 1 & 0 & x & 1 & 0 & 0 & 0 & 0 & 0 & 0 & 0 & 0 & 0 & 0 & 0 & 0 & 0 & 0 & 0 & 0 & 0 & 0 & 0 & 0 & 0 & 0 & 0\\
1 & 0 & 0 & 1 & 1 & 0 & 0 & 0 & 0 & 1 & 1 & 0 & 0 & 1 & 0 & 0 & 0 & 0 & 0 & 0 & 0 & 0 & 0 & 0 & 0 & 0 & 0 & 0 & 0 & 0 & 0 & 0 & 0 & 0 & 0 & 0\\
0 & 1 & 0 & 0 & 1 & 1 & 0 & 0 & 0 & 0 & 1 & 1 & 0 & 0 & 1 & 0 & 0 & 0 & 0 & 0 & 0 & 0 & 0 & 0 & 0 & 0 & 0 & 0 & 0 & 0 & 0 & 0 & 0 & 0 & 0 & 0\\
0 & 0 & 0 & 0 & 0 & x & 1 & 0 & x & 0 & 0 & 0 & 1 & 0 & x & 1 & 0 & 0 & 0 & 0 & 0 & 0 & 0 & 0 & 0 & 0 & 0 & 0 & 0 & 0 & 0 & 0 & 0 & 0 & 0 & 0\\
0 & 0 & 0 & 1 & 0 & 0 & 1 & 1 & 0 & 0 & 0 & 0 & 1 & 1 & 0 & 0 & 1 & 0 & 0 & 0 & 0 & 0 & 0 & 0 & 0 & 0 & 0 & 0 & 0 & 0 & 0 & 0 & 0 & 0 & 0 & 0\\
0 & 0 & 0 & 0 & 1 & 0 & 0 & 1 & 1 & 0 & 0 & 0 & 0 & 1 & 1 & 0 & 0 & 1 & 0 & 0 & 0 & 0 & 0 & 0 & 0 & 0 & 0 & 0 & 0 & 0 & 0 & 0 & 0 & 0 & 0 & 0\\
\hline
0 & 0 & 0 & 0 & 0 & 0 & 0 & 0 & 0 & 0 & 0 & 0 & 0 & 0 & 0 & 0 & 0 & 0 & 1 & 0 & 0 & 1 & 1 & 0 & 0 & 0 & 0 & 1 & 1 & 0 & 0 & 1 & 0 & 0 & 0 & 0\\
0 & 0 & 0 & 0 & 0 & 0 & 0 & 0 & 0 & 0 & 0 & 0 & 0 & 0 & 0 & 0 & 0 & 0 & 0 & 1 & 0 & 0 & 1 & 1 & 0 & 0 & 0 & 0 & 1 & 1 & 0 & 0 & 1 & 0 & 0 & 0\\
0 & 0 & 0 & 0 & 0 & 0 & 0 & 0 & 0 & 0 & 0 & 0 & 0 & 0 & 0 & 0 & 0 & 0 & 0 & 0 & x & 1 & 0 & x & 0 & 0 & 0 & 1 & 0 & x & 1 & 0 & 0 & 0 & 0 & 0\\
0 & 0 & 0 & 0 & 0 & 0 & 0 & 0 & 0 & 0 & 0 & 0 & 0 & 0 & 0 & 0 & 0 & 0 & 0 & 0 & 0 & 1 & 0 & 0 & 1 & 1 & 0 & 0 & 0 & 0 & 1 & 1 & 0 & 0 & 1 & 0\\
0 & 0 & 0 & 0 & 0 & 0 & 0 & 0 & 0 & 0 & 0 & 0 & 0 & 0 & 0 & 0 & 0 & 0 & 0 & 0 & 0 & 0 & 1 & 0 & 0 & 1 & 1 & 0 & 0 & 0 & 0 & 1 & 1 & 0 & 0 & 1\\
0 & 0 & 0 & 0 & 0 & 0 & 0 & 0 & 0 & 0 & 0 & 0 & 0 & 0 & 0 & 0 & 0 & 0 & 0 & 0 & 0 & 0 & 0 & x & 1 & 0 & x & 0 & 0 & 0 & 1 & 0 & x & 1 & 0 & 0\\
0 & 0 & 0 & 0 & 0 & 0 & 0 & 0 & 0 & 0 & 0 & 0 & 0 & 0 & 0 & 0 & 0 & 0 & 1 & 1 & 0 & 0 & 0 & 0 & y & 0 & 0 & 0 & 1 & 0 & 0 & 0 & 0 & y & y & 0\\
0 & 0 & 0 & 0 & 0 & 0 & 0 & 0 & 0 & 0 & 0 & 0 & 0 & 0 & 0 & 0 & 0 & 0 & 0 & 1 & 1 & 0 & 0 & 0 & 0 & y & 0 & 0 & 0 & 1 & 0 & 0 & 0 & 0 & y & y\\
0 & 0 & 0 & 0 & 0 & 0 & 0 & 0 & 0 & 0 & 0 & 0 & 0 & 0 & 0 & 0 & 0 & 0 & 1 & 0 & x & 0 & 0 & 0 & 0 & 0 & xy & 1 & 0 & 0 & 0 & 0 & 0 & y & 0 & xy
\end{array}
\right)
$}}
\endgroup
\end{equation}
We now perform a symplectic column operation (scaled with the monomial $x$) that corresponds to a CNOT between the first qubit and the third qubit within each coarse-grained unit cell to obtain the following parity-check matrix:
\begin{equation}
\label{eq:HCC-coarse-grained-cnot-13}
\begingroup
\setlength{\arraycolsep}{1.2pt}
\renewcommand{\arraystretch}{0.9}
\makebox[\linewidth][c]{\scalebox{\DecouplingMatrixScale}{$
\left(
\begin{array}{*{18}{c}|*{18}{c}}
1 & 0 & 0 & 0 & 0 & 0 & 0 & 0 & xy & 1 & 0 & 0 & 0 & 0 & 0 & y & 0 & xy & 0 & 0 & 0 & 0 & 0 & 0 & 0 & 0 & 0 & 0 & 0 & 0 & 0 & 0 & 0 & 0 & 0 & 0\\
1 & 1 & x & 0 & 0 & 0 & y & 0 & 0 & 0 & 1 & 0 & 0 & 0 & 0 & y & y & 0 & 0 & 0 & 0 & 0 & 0 & 0 & 0 & 0 & 0 & 0 & 0 & 0 & 0 & 0 & 0 & 0 & 0 & 0\\
0 & 1 & 1 & 0 & 0 & 0 & 0 & y & 0 & 0 & 0 & 1 & 0 & 0 & 0 & 0 & y & y & 0 & 0 & 0 & 0 & 0 & 0 & 0 & 0 & 0 & 0 & 0 & 0 & 0 & 0 & 0 & 0 & 0 & 0\\
0 & 0 & x & 1 & 0 & x & 0 & 0 & 0 & 1 & 0 & x & 1 & 0 & 0 & 0 & 0 & 0 & 0 & 0 & 0 & 0 & 0 & 0 & 0 & 0 & 0 & 0 & 0 & 0 & 0 & 0 & 0 & 0 & 0 & 0\\
1 & 0 & x & 1 & 1 & 0 & 0 & 0 & 0 & 1 & 1 & 0 & 0 & 1 & 0 & 0 & 0 & 0 & 0 & 0 & 0 & 0 & 0 & 0 & 0 & 0 & 0 & 0 & 0 & 0 & 0 & 0 & 0 & 0 & 0 & 0\\
0 & 1 & 0 & 0 & 1 & 1 & 0 & 0 & 0 & 0 & 1 & 1 & 0 & 0 & 1 & 0 & 0 & 0 & 0 & 0 & 0 & 0 & 0 & 0 & 0 & 0 & 0 & 0 & 0 & 0 & 0 & 0 & 0 & 0 & 0 & 0\\
0 & 0 & 0 & 0 & 0 & x & 1 & 0 & x & 0 & 0 & 0 & 1 & 0 & x & 1 & 0 & 0 & 0 & 0 & 0 & 0 & 0 & 0 & 0 & 0 & 0 & 0 & 0 & 0 & 0 & 0 & 0 & 0 & 0 & 0\\
0 & 0 & 0 & 1 & 0 & 0 & 1 & 1 & 0 & 0 & 0 & 0 & 1 & 1 & 0 & 0 & 1 & 0 & 0 & 0 & 0 & 0 & 0 & 0 & 0 & 0 & 0 & 0 & 0 & 0 & 0 & 0 & 0 & 0 & 0 & 0\\
0 & 0 & 0 & 0 & 1 & 0 & 0 & 1 & 1 & 0 & 0 & 0 & 0 & 1 & 1 & 0 & 0 & 1 & 0 & 0 & 0 & 0 & 0 & 0 & 0 & 0 & 0 & 0 & 0 & 0 & 0 & 0 & 0 & 0 & 0 & 0\\
\hline
0 & 0 & 0 & 0 & 0 & 0 & 0 & 0 & 0 & 0 & 0 & 0 & 0 & 0 & 0 & 0 & 0 & 0 & 1 & 0 & 0 & 1 & 1 & 0 & 0 & 0 & 0 & 1 & 1 & 0 & 0 & 1 & 0 & 0 & 0 & 0\\
0 & 0 & 0 & 0 & 0 & 0 & 0 & 0 & 0 & 0 & 0 & 0 & 0 & 0 & 0 & 0 & 0 & 0 & 0 & 1 & 0 & 0 & 1 & 1 & 0 & 0 & 0 & 0 & 1 & 1 & 0 & 0 & 1 & 0 & 0 & 0\\
0 & 0 & 0 & 0 & 0 & 0 & 0 & 0 & 0 & 0 & 0 & 0 & 0 & 0 & 0 & 0 & 0 & 0 & 1 & 0 & x & 1 & 0 & x & 0 & 0 & 0 & 1 & 0 & x & 1 & 0 & 0 & 0 & 0 & 0\\
0 & 0 & 0 & 0 & 0 & 0 & 0 & 0 & 0 & 0 & 0 & 0 & 0 & 0 & 0 & 0 & 0 & 0 & 0 & 0 & 0 & 1 & 0 & 0 & 1 & 1 & 0 & 0 & 0 & 0 & 1 & 1 & 0 & 0 & 1 & 0\\
0 & 0 & 0 & 0 & 0 & 0 & 0 & 0 & 0 & 0 & 0 & 0 & 0 & 0 & 0 & 0 & 0 & 0 & 0 & 0 & 0 & 0 & 1 & 0 & 0 & 1 & 1 & 0 & 0 & 0 & 0 & 1 & 1 & 0 & 0 & 1\\
0 & 0 & 0 & 0 & 0 & 0 & 0 & 0 & 0 & 0 & 0 & 0 & 0 & 0 & 0 & 0 & 0 & 0 & 0 & 0 & 0 & 0 & 0 & x & 1 & 0 & x & 0 & 0 & 0 & 1 & 0 & x & 1 & 0 & 0\\
0 & 0 & 0 & 0 & 0 & 0 & 0 & 0 & 0 & 0 & 0 & 0 & 0 & 0 & 0 & 0 & 0 & 0 & 1 & 1 & 0 & 0 & 0 & 0 & y & 0 & 0 & 0 & 1 & 0 & 0 & 0 & 0 & y & y & 0\\
0 & 0 & 0 & 0 & 0 & 0 & 0 & 0 & 0 & 0 & 0 & 0 & 0 & 0 & 0 & 0 & 0 & 0 & 1 & x & x & 0 & 0 & 0 & 0 & xy & 0 & 0 & 0 & x & 0 & 0 & 0 & 0 & xy & xy\\
0 & 0 & 0 & 0 & 0 & 0 & 0 & 0 & 0 & 0 & 0 & 0 & 0 & 0 & 0 & 0 & 0 & 0 & 0 & 0 & x & 0 & 0 & 0 & 0 & 0 & xy & 1 & 0 & 0 & 0 & 0 & 0 & y & 0 & xy
\end{array}
\right)
$}}
\endgroup
\end{equation}
We now repeat the same symplectic column operation (scaled with the appropriate monomials) that corresponds to a CNOT between the first qubit and the other qubits within the same coarse-grained unit cell (shown in the first row of the parity-check matrix) before renormalizing the matrix again to obtain the following parity-check matrix:
\begin{equation}
\label{eq:HCC-coarse-grained-cnot-1-to-all}
\begingroup
\setlength{\arraycolsep}{1.2pt}
\renewcommand{\arraystretch}{0.9}
\makebox[\linewidth][c]{\scalebox{\DecouplingMatrixScale}{$
\left(
\begin{array}{*{18}{c}|*{18}{c}}
1 & 0 & 0 & 0 & 0 & 0 & 0 & 0 & 0 & 0 & 0 & 0 & 0 & 0 & 0 & 0 & 0 & 0 & 0 & 0 & 0 & 0 & 0 & 0 & 0 & 0 & 0 & 0 & 0 & 0 & 0 & 0 & 0 & 0 & 0 & 0\\
1 & 1 & x & 0 & 0 & 0 & y & 0 & xy & 1 & 1 & 0 & 0 & 0 & 0 & 0 & y & xy & 0 & 0 & 0 & 0 & 0 & 0 & 0 & 0 & 0 & 0 & 0 & 0 & 0 & 0 & 0 & 0 & 0 & 0\\
0 & 1 & 1 & 0 & 0 & 0 & 0 & y & 0 & 0 & 0 & 1 & 0 & 0 & 0 & 0 & y & y & 0 & 0 & 0 & 0 & 0 & 0 & 0 & 0 & 0 & 0 & 0 & 0 & 0 & 0 & 0 & 0 & 0 & 0\\
0 & 0 & x & 1 & 0 & x & 0 & 0 & 0 & 1 & 0 & x & 1 & 0 & 0 & 0 & 0 & 0 & 0 & 0 & 0 & 0 & 0 & 0 & 0 & 0 & 0 & 0 & 0 & 0 & 0 & 0 & 0 & 0 & 0 & 0\\
1 & 0 & x & 1 & 1 & 0 & 0 & 0 & xy & 0 & 1 & 0 & 0 & 1 & 0 & y & 0 & xy & 0 & 0 & 0 & 0 & 0 & 0 & 0 & 0 & 0 & 0 & 0 & 0 & 0 & 0 & 0 & 0 & 0 & 0\\
0 & 1 & 0 & 0 & 1 & 1 & 0 & 0 & 0 & 0 & 1 & 1 & 0 & 0 & 1 & 0 & 0 & 0 & 0 & 0 & 0 & 0 & 0 & 0 & 0 & 0 & 0 & 0 & 0 & 0 & 0 & 0 & 0 & 0 & 0 & 0\\
0 & 0 & 0 & 0 & 0 & x & 1 & 0 & x & 0 & 0 & 0 & 1 & 0 & x & 1 & 0 & 0 & 0 & 0 & 0 & 0 & 0 & 0 & 0 & 0 & 0 & 0 & 0 & 0 & 0 & 0 & 0 & 0 & 0 & 0\\
0 & 0 & 0 & 1 & 0 & 0 & 1 & 1 & 0 & 0 & 0 & 0 & 1 & 1 & 0 & 0 & 1 & 0 & 0 & 0 & 0 & 0 & 0 & 0 & 0 & 0 & 0 & 0 & 0 & 0 & 0 & 0 & 0 & 0 & 0 & 0\\
0 & 0 & 0 & 0 & 1 & 0 & 0 & 1 & 1 & 0 & 0 & 0 & 0 & 1 & 1 & 0 & 0 & 1 & 0 & 0 & 0 & 0 & 0 & 0 & 0 & 0 & 0 & 0 & 0 & 0 & 0 & 0 & 0 & 0 & 0 & 0\\
\hline
0 & 0 & 0 & 0 & 0 & 0 & 0 & 0 & 0 & 0 & 0 & 0 & 0 & 0 & 0 & 0 & 0 & 0 & 0 & 0 & 0 & 1 & 1 & 0 & 0 & 0 & 0 & 1 & 1 & 0 & 0 & 1 & 0 & 0 & 0 & 0\\
0 & 0 & 0 & 0 & 0 & 0 & 0 & 0 & 0 & 0 & 0 & 0 & 0 & 0 & 0 & 0 & 0 & 0 & 0 & 1 & 0 & 0 & 1 & 1 & 0 & 0 & 0 & 0 & 1 & 1 & 0 & 0 & 1 & 0 & 0 & 0\\
0 & 0 & 0 & 0 & 0 & 0 & 0 & 0 & 0 & 0 & 0 & 0 & 0 & 0 & 0 & 0 & 0 & 0 & 0 & 0 & x & 1 & 0 & x & 0 & 0 & 0 & 1 & 0 & x & 1 & 0 & 0 & 0 & 0 & 0\\
0 & 0 & 0 & 0 & 0 & 0 & 0 & 0 & 0 & 0 & 0 & 0 & 0 & 0 & 0 & 0 & 0 & 0 & 0 & 0 & 0 & 1 & 0 & 0 & 1 & 1 & 0 & 0 & 0 & 0 & 1 & 1 & 0 & 0 & 1 & 0\\
0 & 0 & 0 & 0 & 0 & 0 & 0 & 0 & 0 & 0 & 0 & 0 & 0 & 0 & 0 & 0 & 0 & 0 & 0 & 0 & 0 & 0 & 1 & 0 & 0 & 1 & 1 & 0 & 0 & 0 & 0 & 1 & 1 & 0 & 0 & 1\\
0 & 0 & 0 & 0 & 0 & 0 & 0 & 0 & 0 & 0 & 0 & 0 & 0 & 0 & 0 & 0 & 0 & 0 & 0 & 0 & 0 & 0 & 0 & x & 1 & 0 & x & 0 & 0 & 0 & 1 & 0 & x & 1 & 0 & 0\\
0 & 0 & 0 & 0 & 0 & 0 & 0 & 0 & 0 & 0 & 0 & 0 & 0 & 0 & 0 & 0 & 0 & 0 & 0 & 1 & 0 & 0 & 0 & 0 & y & 0 & 0 & 0 & 1 & 0 & 0 & 0 & 0 & y & y & 0\\
0 & 0 & 0 & 0 & 0 & 0 & 0 & 0 & 0 & 0 & 0 & 0 & 0 & 0 & 0 & 0 & 0 & 0 & 0 & 1 & 1 & 0 & 0 & 0 & 0 & y & 0 & 0 & 0 & 1 & 0 & 0 & 0 & 0 & y & y\\
0 & 0 & 0 & 0 & 0 & 0 & 0 & 0 & 0 & 0 & 0 & 0 & 0 & 0 & 0 & 0 & 0 & 0 & 0 & 0 & x & 0 & 0 & 0 & 0 & 0 & xy & 1 & 0 & 0 & 0 & 0 & 0 & y & 0 & xy
\end{array}
\right)
$}}
\endgroup
\end{equation}

 Now, we can see that the first row of the parity-check matrix has only one non-zero entry, which corresponds to the first qubit in the coarse-grained unit cell. We would like to perform row operations to relabel the stabilizer so that the $X$ stabilizer generator corresponding to the first row of the parity-check matrix is supported only on the first qubit in the coarse-grained unit cell. Adding the first row to the second and fifth rows clears the first column below its pivot. Paired column commutation has already cleared the conjugate column in the $Z$ block, so the first qubit is now a decoupled product-state qubit.

 \begin{equation}
\label{eq:HCC-coarse-grained-first-row-relabel}
\begingroup
\setlength{\arraycolsep}{1.2pt}
\renewcommand{\arraystretch}{0.9}
\makebox[\linewidth][c]{\scalebox{\DecouplingMatrixScale}{$
\left(
\begin{array}{*{18}{c}|*{18}{c}}
1 & 0 & 0 & 0 & 0 & 0 & 0 & 0 & 0 & 0 & 0 & 0 & 0 & 0 & 0 & 0 & 0 & 0 & 0 & 0 & 0 & 0 & 0 & 0 & 0 & 0 & 0 & 0 & 0 & 0 & 0 & 0 & 0 & 0 & 0 & 0\\
0 & 1 & x & 0 & 0 & 0 & y & 0 & xy & 1 & 1 & 0 & 0 & 0 & 0 & 0 & y & xy & 0 & 0 & 0 & 0 & 0 & 0 & 0 & 0 & 0 & 0 & 0 & 0 & 0 & 0 & 0 & 0 & 0 & 0\\
0 & 1 & 1 & 0 & 0 & 0 & 0 & y & 0 & 0 & 0 & 1 & 0 & 0 & 0 & 0 & y & y & 0 & 0 & 0 & 0 & 0 & 0 & 0 & 0 & 0 & 0 & 0 & 0 & 0 & 0 & 0 & 0 & 0 & 0\\
0 & 0 & x & 1 & 0 & x & 0 & 0 & 0 & 1 & 0 & x & 1 & 0 & 0 & 0 & 0 & 0 & 0 & 0 & 0 & 0 & 0 & 0 & 0 & 0 & 0 & 0 & 0 & 0 & 0 & 0 & 0 & 0 & 0 & 0\\
0 & 0 & x & 1 & 1 & 0 & 0 & 0 & xy & 0 & 1 & 0 & 0 & 1 & 0 & y & 0 & xy & 0 & 0 & 0 & 0 & 0 & 0 & 0 & 0 & 0 & 0 & 0 & 0 & 0 & 0 & 0 & 0 & 0 & 0\\
0 & 1 & 0 & 0 & 1 & 1 & 0 & 0 & 0 & 0 & 1 & 1 & 0 & 0 & 1 & 0 & 0 & 0 & 0 & 0 & 0 & 0 & 0 & 0 & 0 & 0 & 0 & 0 & 0 & 0 & 0 & 0 & 0 & 0 & 0 & 0\\
0 & 0 & 0 & 0 & 0 & x & 1 & 0 & x & 0 & 0 & 0 & 1 & 0 & x & 1 & 0 & 0 & 0 & 0 & 0 & 0 & 0 & 0 & 0 & 0 & 0 & 0 & 0 & 0 & 0 & 0 & 0 & 0 & 0 & 0\\
0 & 0 & 0 & 1 & 0 & 0 & 1 & 1 & 0 & 0 & 0 & 0 & 1 & 1 & 0 & 0 & 1 & 0 & 0 & 0 & 0 & 0 & 0 & 0 & 0 & 0 & 0 & 0 & 0 & 0 & 0 & 0 & 0 & 0 & 0 & 0\\
0 & 0 & 0 & 0 & 1 & 0 & 0 & 1 & 1 & 0 & 0 & 0 & 0 & 1 & 1 & 0 & 0 & 1 & 0 & 0 & 0 & 0 & 0 & 0 & 0 & 0 & 0 & 0 & 0 & 0 & 0 & 0 & 0 & 0 & 0 & 0\\
\hline
0 & 0 & 0 & 0 & 0 & 0 & 0 & 0 & 0 & 0 & 0 & 0 & 0 & 0 & 0 & 0 & 0 & 0 & 0 & 0 & 0 & 1 & 1 & 0 & 0 & 0 & 0 & 1 & 1 & 0 & 0 & 1 & 0 & 0 & 0 & 0\\
0 & 0 & 0 & 0 & 0 & 0 & 0 & 0 & 0 & 0 & 0 & 0 & 0 & 0 & 0 & 0 & 0 & 0 & 0 & 1 & 0 & 0 & 1 & 1 & 0 & 0 & 0 & 0 & 1 & 1 & 0 & 0 & 1 & 0 & 0 & 0\\
0 & 0 & 0 & 0 & 0 & 0 & 0 & 0 & 0 & 0 & 0 & 0 & 0 & 0 & 0 & 0 & 0 & 0 & 0 & 0 & x & 1 & 0 & x & 0 & 0 & 0 & 1 & 0 & x & 1 & 0 & 0 & 0 & 0 & 0\\
0 & 0 & 0 & 0 & 0 & 0 & 0 & 0 & 0 & 0 & 0 & 0 & 0 & 0 & 0 & 0 & 0 & 0 & 0 & 0 & 0 & 1 & 0 & 0 & 1 & 1 & 0 & 0 & 0 & 0 & 1 & 1 & 0 & 0 & 1 & 0\\
0 & 0 & 0 & 0 & 0 & 0 & 0 & 0 & 0 & 0 & 0 & 0 & 0 & 0 & 0 & 0 & 0 & 0 & 0 & 0 & 0 & 0 & 1 & 0 & 0 & 1 & 1 & 0 & 0 & 0 & 0 & 1 & 1 & 0 & 0 & 1\\
0 & 0 & 0 & 0 & 0 & 0 & 0 & 0 & 0 & 0 & 0 & 0 & 0 & 0 & 0 & 0 & 0 & 0 & 0 & 0 & 0 & 0 & 0 & x & 1 & 0 & x & 0 & 0 & 0 & 1 & 0 & x & 1 & 0 & 0\\
0 & 0 & 0 & 0 & 0 & 0 & 0 & 0 & 0 & 0 & 0 & 0 & 0 & 0 & 0 & 0 & 0 & 0 & 0 & 1 & 0 & 0 & 0 & 0 & y & 0 & 0 & 0 & 1 & 0 & 0 & 0 & 0 & y & y & 0\\
0 & 0 & 0 & 0 & 0 & 0 & 0 & 0 & 0 & 0 & 0 & 0 & 0 & 0 & 0 & 0 & 0 & 0 & 0 & 1 & 1 & 0 & 0 & 0 & 0 & y & 0 & 0 & 0 & 1 & 0 & 0 & 0 & 0 & y & y\\
0 & 0 & 0 & 0 & 0 & 0 & 0 & 0 & 0 & 0 & 0 & 0 & 0 & 0 & 0 & 0 & 0 & 0 & 0 & 0 & x & 0 & 0 & 0 & 0 & 0 & xy & 1 & 0 & 0 & 0 & 0 & 0 & y & 0 & xy
\end{array}
\right)
$}}
\endgroup
\end{equation}
This calculation explicitly isolates one product-state qubit.
For the original color-code map, each single-sector charge quotient
has dimension two: writing $f=1+x+xy$ and $g=1+y+xy$,
\begin{equation}
R/(f,g)\cong\F_2[x^{\pm1}]/(1+x+x^2),
\end{equation}
because $f+g=x+y$. The full torsion charge dimension is four.
Theorem~\ref{thm:decoupling_count} therefore gives two toric copies.
The following commuting matrix displays a representative normal form
with two toric blocks and fourteen product-state qubits. It specifies
the target blocks; the remaining transformation sequence is not
listed here.
\begin{equation}
\begingroup
\setlength{\arraycolsep}{1.2pt}
\renewcommand{\arraystretch}{0.9}
\makebox[\linewidth][c]{\scalebox{\DecouplingMatrixScale}{$
\left(
\begin{array}{*{18}{c}|*{18}{c}}
1 & 0 & 0 & 0 & 0 & 0 & 0 & 0 & 0 & 0 & 0 & 0 & 0 & 0 & 0 & 0 & 0 & 0 & 0 & 0 & 0 & 0 & 0 & 0 & 0 & 0 & 0 & 0 & 0 & 0 & 0 & 0 & 0 & 0 & 0 & 0\\
0 & 1 & 0 & 0 & 0 & 0 & 0 & 0 & 0 & 0 & 0 & 0 & 0 & 0 & 0 & 0 & 0 & 0 & 0 & 0 & 0 & 0 & 0 & 0 & 0 & 0 & 0 & 0 & 0 & 0 & 0 & 0 & 0 & 0 & 0 & 0\\
0 & 0 & 1 & 0 & 0 & 0 & 0 & 0 & 0 & 0 & 0 & 0 & 0 & 0 & 0 & 0 & 0 & 0 & 0 & 0 & 0 & 0 & 0 & 0 & 0 & 0 & 0 & 0 & 0 & 0 & 0 & 0 & 0 & 0 & 0 & 0\\
0 & 0 & 0 & 1 & 0 & 0 & 0 & 0 & 0 & 0 & 0 & 0 & 0 & 0 & 0 & 0 & 0 & 0 & 0 & 0 & 0 & 0 & 0 & 0 & 0 & 0 & 0 & 0 & 0 & 0 & 0 & 0 & 0 & 0 & 0 & 0\\
0 & 0 & 0 & 0 & 1 & 0 & 0 & 0 & 0 & 0 & 0 & 0 & 0 & 0 & 0 & 0 & 0 & 0 & 0 & 0 & 0 & 0 & 0 & 0 & 0 & 0 & 0 & 0 & 0 & 0 & 0 & 0 & 0 & 0 & 0 & 0\\
0 & 0 & 0 & 0 & 0 & 1 & 0 & 0 & 0 & 0 & 0 & 0 & 0 & 0 & 0 & 0 & 0 & 0 & 0 & 0 & 0 & 0 & 0 & 0 & 0 & 0 & 0 & 0 & 0 & 0 & 0 & 0 & 0 & 0 & 0 & 0\\
0 & 0 & 0 & 0 & 0 & 0 & 1 & 0 & 0 & 0 & 0 & 0 & 0 & 0 & 0 & 0 & 0 & 0 & 0 & 0 & 0 & 0 & 0 & 0 & 0 & 0 & 0 & 0 & 0 & 0 & 0 & 0 & 0 & 0 & 0 & 0\\
0 & 0 & 0 & 0 & 0 & 0 & 0 & 1 + x & 1 + y & 0 & 0 & 0 & 0 & 0 & 0 & 0 & 0 & 0 & 0 & 0 & 0 & 0 & 0 & 0 & 0 & 0 & 0 & 0 & 0 & 0 & 0 & 0 & 0 & 0 & 0 & 0\\
0 & 0 & 0 & 0 & 0 & 0 & 0 & 0 & 0 & 1 + x & 1 + y & 0 & 0 & 0 & 0 & 0 & 0 & 0 & 0 & 0 & 0 & 0 & 0 & 0 & 0 & 0 & 0 & 0 & 0 & 0 & 0 & 0 & 0 & 0 & 0 & 0\\
\hline
0 & 0 & 0 & 0 & 0 & 0 & 0 & 0 & 0 & 0 & 0 & 0 & 0 & 0 & 0 & 0 & 0 & 0 & 0 & 0 & 0 & 0 & 0 & 0 & 0 & 0 & 0 & 0 & 0 & 1 & 0 & 0 & 0 & 0 & 0 & 0\\
0 & 0 & 0 & 0 & 0 & 0 & 0 & 0 & 0 & 0 & 0 & 0 & 0 & 0 & 0 & 0 & 0 & 0 & 0 & 0 & 0 & 0 & 0 & 0 & 0 & 0 & 0 & 0 & 0 & 0 & 1 & 0 & 0 & 0 & 0 & 0\\
0 & 0 & 0 & 0 & 0 & 0 & 0 & 0 & 0 & 0 & 0 & 0 & 0 & 0 & 0 & 0 & 0 & 0 & 0 & 0 & 0 & 0 & 0 & 0 & 0 & 0 & 0 & 0 & 0 & 0 & 0 & 1 & 0 & 0 & 0 & 0\\
0 & 0 & 0 & 0 & 0 & 0 & 0 & 0 & 0 & 0 & 0 & 0 & 0 & 0 & 0 & 0 & 0 & 0 & 0 & 0 & 0 & 0 & 0 & 0 & 0 & 0 & 0 & 0 & 0 & 0 & 0 & 0 & 1 & 0 & 0 & 0\\
0 & 0 & 0 & 0 & 0 & 0 & 0 & 0 & 0 & 0 & 0 & 0 & 0 & 0 & 0 & 0 & 0 & 0 & 0 & 0 & 0 & 0 & 0 & 0 & 0 & 0 & 0 & 0 & 0 & 0 & 0 & 0 & 0 & 1 & 0 & 0\\
0 & 0 & 0 & 0 & 0 & 0 & 0 & 0 & 0 & 0 & 0 & 0 & 0 & 0 & 0 & 0 & 0 & 0 & 0 & 0 & 0 & 0 & 0 & 0 & 0 & 0 & 0 & 0 & 0 & 0 & 0 & 0 & 0 & 0 & 1 & 0\\
0 & 0 & 0 & 0 & 0 & 0 & 0 & 0 & 0 & 0 & 0 & 0 & 0 & 0 & 0 & 0 & 0 & 0 & 0 & 0 & 0 & 0 & 0 & 0 & 0 & 0 & 0 & 0 & 0 & 0 & 0 & 0 & 0 & 0 & 0 & 1\\
0 & 0 & 0 & 0 & 0 & 0 & 0 & 0 & 0 & 0 & 0 & 0 & 0 & 0 & 0 & 0 & 0 & 0 & 0 & 0 & 0 & 0 & 0 & 0 & 0 & x + xy & y + xy & 0 & 0 & 0 & 0 & 0 & 0 & 0 & 0 & 0\\
0 & 0 & 0 & 0 & 0 & 0 & 0 & 0 & 0 & 0 & 0 & 0 & 0 & 0 & 0 & 0 & 0 & 0 & 0 & 0 & 0 & 0 & 0 & 0 & 0 & 0 & 0 & x + xy & y + xy & 0 & 0 & 0 & 0 & 0 & 0 & 0
\end{array}
\right)
$}}
\endgroup
\end{equation}

\section{Layer-decoupling decoder lemmas and proofs}
\label{app:layer-decoupling-decoder-lemmas}

\begin{lemma}[Chain Isomorphism]\label{lem:chain-isomorphism}
For the fixed invertible CSS-preserving symplectic reduction, the maps
$\phi_0=U_X$, $\phi_1=V_X^{-1}=V_Z^\dagger$, and
$\phi_2=U_Z^{-\dagger}$ define a chain isomorphism
$\calC\to\tilde{\calC}$. This isomorphism descends to finite tori
whose periods are compatible with the coarse-graining.
\end{lemma}
\begin{proof}
Express the physical complex over the final coarse-grained ring,
so $\partial_1=H'_X$ and $\partial_2=H_Z^{\prime\dagger}$.
Symplectic preservation gives $V_XV_Z^\dagger=I$.
The defining reduction identities imply
\begin{align*}
\tilde{\partial}_1\phi_1
 &= (U_XH'_XV_X)V_X^{-1}\\
 &= U_XH'_X=\phi_0\partial_1,\\
\tilde{\partial}_2\phi_2
 &= (U_ZH'_ZV_Z)^\dagger U_Z^{-\dagger}\\
 &= V_Z^\dagger H_Z^{\prime\dagger}
  =\phi_1\partial_2.
\end{align*}
Each $\phi_j$ has an inverse over the same Laurent-polynomial ring.
Imposing compatible finite-torus periods preserves these identities
and the inverse identities. Hence the resulting finite-torus maps
are also chain isomorphisms and induce isomorphisms on homology.
\end{proof}

\begin{lemma}[Projection Chain Maps]
\label{lem:projection-chain-maps}
The projections $\pi^{(i)}$ and $\pi^A$ are chain maps, i.e., they satisfy the following properties:
\begin{align}
    \pi^{(i)}_0 \circ \tilde{\partial}_1 &= \tilde{\partial}_1^{TC} \circ \pi^{(i)}_1,\\
    \pi^{(i)}_1 \circ \tilde{\partial}_2 &= \tilde{\partial}_2^{TC} \circ \pi^{(i)}_2,\\
    \pi^A_0 \circ \tilde{\partial}_1 &= \tilde{\partial}_1^{A} \circ \pi^A_1,\\
    \pi^A_1 \circ \tilde{\partial}_2 &= \tilde{\partial}_2^{A} \circ \pi^A_2.
\end{align}
\end{lemma}

\begin{proof}
    The proof is trivial since we are effectively projecting out the other summands in the direct sum of chain complexes. Thus, the chain map properties follow directly from the definitions of the projections and the differentials.
\end{proof}

\begin{corollary}[Chain Isomorphism and Projection Chain Maps]
\label{cor:chain-isomorphism-projection}
Let $\phi$ be the vector-space isomorphism defined in Section~\ref{sec:morphism-chain-complexes} and let $\pi^{(i)}$ and $\pi^A$ be the projections defined in Lemma~\ref{lem:projection-chain-maps}. Then, we have the following chain maps:
\begin{align}
    \pi_0^{(i)}\circ \phi_0 \circ \partial_1 &= \tilde{\partial}_1^{TC}\circ \pi_1^{(i)} \circ \phi_1 = \pi_0^{(i)} \circ \tilde{\partial}_1 \circ \phi_1 \\
    \pi_1^{(i)}\circ \phi_1 \circ \partial_2 &= \tilde{\partial}_2^{TC}\circ \pi_2^{(i)} \circ \phi_2 = \pi_1^{(i)} \circ \tilde{\partial}_2 \circ \phi_2, \\
    \pi_0^{A}\circ \phi_0 \circ \partial_1 &= \tilde{\partial}_1^{A}\circ \pi_1^{A} \circ \phi_1 = \pi_0^{A} \circ \tilde{\partial}_1 \circ \phi_1 \\
    \pi_1^{A}\circ \phi_1 \circ \partial_2 &= \tilde{\partial}_2^{A}\circ \pi_2^{A} \circ \phi_2 = \pi_1^{A} \circ \tilde{\partial}_2 \circ \phi_2.
\end{align}
\end{corollary}

\section{Computing short strings}
\label{app:computing-short-strings}
We describe the short-string construction used for the $24\times24$ gross code. Let $P$ be a partition of the monomial labels $\{x^iy^j:0\leq i<\ell,\ 0\leq j<m\}$ into disjoint rectangles $P_{i,j}$ of size $\ell'\times m'$. Assign to each rectangle the two types of qubits and the syndrome checks with those labels. For the $24\times24$ gross code take $\ell'=m'=3$. More generally, the rectangle dimensions must divide $\ell$ and $m$, and the support-constrained linear systems below must be checked for solvability. The stabilizer interaction range alone does not establish that every required transport fits in one source rectangle.

Consider horizontally adjacent cells $P_{i,j}$ and $P_{i+1,j}$ in this $3\times3$ partition. Fix an ordering of the nine monomials in each cell, and let $\iota_{i,j}:\F_2^9\to C_0$ insert a vector as a physical syndrome supported in $P_{i,j}$. Here $H$ is the physical syndrome map used in Section~\ref{sec:gross}. For every standard basis vector $\hat e_{qp}$, solve for an error $e^{(qp)}$ supported on the qubits assigned to $P_{i,j}$ and a vector $a_{qp}\in\F_2^9$ such that
\begin{equation}\label{eq:short-string-transfer}
He^{(qp)}=\iota_{i,j}\hat e_{qp}+\iota_{i+1,j}a_{qp}.
\end{equation}
This equation enforces zero physical syndrome in every other cell. When all nine systems are solvable, fix one solution for each basis vector and define the linear transfer matrix by
\begin{equation}
A_x\hat e_{qp}=a_{qp}.
\end{equation}
For the gross polynomial, all nine systems in each direction have solutions supported on the 18 qubits of the source cell, with the syndrome confined to the source and adjacent target cells. These identities can be checked on the infinite lattice, without using a periodic seam. For another code or cell choice, solvability of the corresponding systems remains a separate check.

For $v\in\ker(A_x-I)$, linearity makes $\sum_{q,p}v_{qp}e^{(qp)}$ a short string with exact syndrome $\iota_{i,j}v+\iota_{i+1,j}v$. More generally, concatenate $k$ translated transfers, successively applied to $v,A_xv,\ldots,A_x^{k-1}v$. The intermediate physical syndromes cancel, leaving $v$ in the initial rectangle and $A_x^kv$ in the rectangle displaced by $k\ell'$. Thus fixed vectors in $\ker(A_x^k-I)$ yield identical endpoint patterns. The analogous construction defines $A_y$. To obtain a common transport scale $b'_i$ in both directions, the candidate pattern must be fixed by both $A_x^{b'_i/\ell'}$ and $A_y^{b'_i/m'}$, with integer exponents.

The classes of candidate fixed vectors are tested in the finite quotient $R/(f,g)$, where $R=\F_2[x,y]/(x^\ell-1,y^m-1)$ and $(f,g)$ is the physical syndrome ideal. The following rank tests certify a basis of this finite quotient; they do not by themselves establish a uniform infinite-lattice transport construction.
One identifies \(R\) with the \(\ell m\)-dimensional \(\mathbb{F}_2\)-vector space with basis given by the monomials \(x^i y^j\) for \(0 \le i < \ell\) and \(0 \le j < m\). A polynomial is then represented by its coefficient vector in this basis, with exponents reduced modulo \(\ell\) and \(m\). The ideal \((f,g)\) is generated, as a vector subspace, by all monomial translates \(x^a y^b f\) and \(x^a y^b g\), since any polynomial multiple of \(f\) or \(g\) is an \(\mathbb{F}_2\)-linear combination of such translates. Forming a matrix \(I\) whose columns are the coefficient vectors of all these translates, one obtains
\begin{equation}
\dim (f,g) = \operatorname{rank}(I),
\end{equation}
and hence
\begin{equation}
\dim\!\bigl(R/(f,g)\bigr) = \ell m - \operatorname{rank}(I).
\end{equation}
Given candidate basis elements \(b_1,\dots,b_k\), one similarly forms their coefficient vectors and appends them as columns to \(I\), producing a matrix \([I \mid B]\). Then the images of the \(b_i\) span the quotient if and only if
\begin{equation}
\operatorname{rank}([I \mid B]) = \ell m,
\end{equation}
and they are linearly independent modulo the ideal if and only if
\begin{equation}
\operatorname{rank}([I \mid B]) - \operatorname{rank}(I) = k.
\end{equation}
Thus they form a basis precisely when both conditions hold.

\end{document}